\newif\ifllncs  
\newif\iffull   

\fulltrue

\ifllncs
    \documentclass[runningheads,a4paper]{llncs}

\else
    \documentclass[11pt]{article}
\fi

\usepackage{amsmath, amsfonts, amssymb, amsthm}
\usepackage{color}
\usepackage{url}
\usepackage{qcircuit}
\usepackage{algorithm}
\usepackage{algpseudocode}
\usepackage[colorlinks=true,
            linkcolor=blue,
            urlcolor=blue,
            citecolor=blue]{hyperref}
\usepackage{graphicx}
\usepackage{multirow}
\usepackage{multicol}
\usepackage{tikz}
\usepackage{authblk}

\usepackage{subcaption}
\usepackage{booktabs}
\usepackage{pgfplots}
\pgfplotsset{compat=1.18}
\usetikzlibrary{calc,arrows.meta}

\newcommand{\F}{\mathbb{F}}

\newcommand{\floor}[1]{\left\lfloor #1 \right\rfloor}
\newcommand{\ceil}[1]{\left\lceil #1 \right\rceil}

\newcommand{\bra}[1]{\langle #1|}
\newcommand{\ket}[1]{|#1 \rangle}

\newcommand{\QFT}{\mathsf{QFT}}

\newcommand{\triang}{\quad\,\blacktriangleright\ \,}

\newcommand{\gateout}[1]{*+<.6em>{#1\ } \POS ="i",[0,0]+R *{\triang},"i"+UR;"i"+UL **\dir{-};"i"+DL **\dir{-};"i"+DR **\dir{-};"i"+UR **\dir{-},"i" \qw}
\newcommand{\multigateout}[2]{*+<1em,.9em>{\hphantom{#2}} \POS [0,0]="i",[0,0].[#1,0]="e", !C *{#2}, [0,0]+R *{\triang}, "e"+UR;"e"+UL **\dir{-};"e"+DL **\dir{-};"e"+DR **\dir{-};"e"+UR **\dir{-},"i" \qw}

\newcommand{\multigateouttwo}[2]{*+<1em,.9em>{\hphantom{#2}} \POS [0,0]="i",[0,0].[#1,0]="e", !C *{#2}, [0,0]+R *{\triang}, [#1,0]+R *{\triang}, "e"+UR;"e"+UL **\dir{-};"e"+DL **\dir{-};"e"+DR **\dir{-};"e"+UR **\dir{-},"i" \qw}

\newcommand{\multigateouttwohi}[2]{*+<1em,.9em>{\hphantom{#2}} \POS [0,0]="i",[0,0].[#1,0]="e", !C *{#2}, [0,0]+R *{\triang},  [1,0]+R *{\triang},"e"+UR;"e"+UL **\dir{-};"e"+DL **\dir{-};"e"+DR **\dir{-};"e"+UR **\dir{-},"i" \qw}

\iffull
    \newtheorem{theorem}{Theorem}[section]
    \newtheorem{lemma}[theorem]{Lemma}

\fi

\graphicspath{{sec_circuit/}, {sec_point_add/}}

\title{Quantum Algorithm for Elliptic Curve Discrete Logarithms with Space-Efficient Point Addition\thanks{This paper supersedes our earlier preprint \href{https://arxiv.org/pdf/2604.02311}{arXiv:2604.02311}. Compared with the earlier version, the present paper reduces the space complexity from $5n + O(\log_2 n)$ to $3n + O(\log_2 n)$ for affine point addition and from $3n + O(\log_2 n)$ to $2n + O(\log_2 n)$ for modular inversion. It also improves the overall Toffoli counts of ECDLP from $O(n^3)$ to $O(n^3/\log n)$. 
}}

\iffull
   \author[1]{Han Luo\thanks{Equal contribution.}}
   \author[2,3]{Ziyi Yang\protect\footnotemark[2]}
   \author[4]{Jingquan Luo\protect\footnotemark[2]}
   \author[2,3]{Ziruo Wang}
   \author[2,3]{Yuexin Su}
   \author[5,6]{Xiaoming Sun}
   \author[4]{\\Lvzhou Li\thanks{Corresponding author. Email: lilvzh@mail.sysu.edu.cn}}
   \author[2,3]{Tongyang Li\thanks{Corresponding author. Email: tongyangli@pku.edu.cn}} 
   \affil[1]{Institute for Interdisciplinary Information Sciences, Tsinghua University}
   \affil[2]{Center on Frontiers of Computing Studies, Peking University}
   \affil[3]{School of Computer Science, Peking University}
   \affil[4]{Institute of Quantum Computing  and Software, School of Computer Science and Engineering, Sun Yat-Sen University}
   \affil[5]{State Key Lab of Processors, Institute of Computing Technology, Chinese Academy of Sciences}
   \affil[6]{School of Computer Science and Technology, University of Chinese Academy of Sciences}
\else
   \author{Anonymous submission}
\fi

\date{}

\begin{document}

\maketitle

\vspace{-6mm}
\begin{abstract}
The Elliptic Curve Discrete Logarithm Problem (ECDLP) is a fundamental problem in cryptography, and reducing the resource requirements of quantum algorithms for solving ECDLP is an important goal. In this work, we present a space-efficient quantum algorithm for solving the ECDLP over prime fields, achieving an implementation with only
$3n+6\lfloor \log_2 n \rfloor+O(1)$ logical qubits and $1056n^3/\log_2 n+O(n^2)$ Toffoli gates, where $n$ is the bit-length of the prime. For a 256-bit prime-field curve, our construction requires only 835 logical qubits, reducing the previous best estimates of 1098 and 1175 logical qubits by Chevignard et al.~[EUROCRYPT 2026] and Babbush et al.~[ArXiv Preprint 2026], respectively.

The key to our improvement is a new space-efficient reversible modular inversion circuit, which addresses the dominant space bottleneck in affine-coordinate point addition. Starting from the extended Euclidean algorithm (EEA), we refine the register-sharing technique of Proos and Zalka by introducing length registers and location-controlled arithmetic to compactly store and update intermediate variables. We further optimize the reversible update procedures and construct the corresponding controlled arithmetic circuits, resulting in a modular inversion circuit implemented by only
$2n+6\lfloor \log_2 n \rfloor+O(1)$ logical qubits and
$229n^2+O(n\log_2 n)$ Toffoli gates.
This modular inversion circuit together with mid-circuit measurements and classical feed-forward operations provides a space-efficient controlled affine point-addition circuit and a complete implementation of Shor's algorithm for ECDLP.
\vspace{6mm}
\end{abstract}

\section{Introduction}

Elliptic curve cryptography (ECC), first proposed in the mid-1980s as an alternative framework for public-key cryptography \cite{koblitz1987elliptic,miller1985use}, has become one of the foundational primitives in modern cryptographic systems.
The security of ECC-based cryptosystems relies on the computational difficulty of the Elliptic Curve Discrete Logarithm Problem (ECDLP).
Specifically, given an elliptic curve defined over a finite field, a base point $P$ on that curve, and another point $Q$ resulting from scalar multiplication of $P$ by an unknown integer $m$, the ECDLP asks for the recovery of the scalar $m$ from the two points.
Despite decades of research, no classical algorithm is known that solves the ECDLP in polynomial time of bit length; the best generic attacks, such as Pollard's rho method, require $\Theta(\sqrt{p})$ group operations, where $p$ denotes the field size of the curve which is exponentially large \cite{galbraith2016recent}.

Compared to classical public-key systems such as RSA \cite{rivest1978method} and finite-field Diffie-Hellman \cite{diffie1976new}, ECC offers equivalent security with substantially smaller key sizes.
For example, according to NIST recommendations \cite{barker2020nist}, a 256-bit elliptic curve provides a security level of 128 bits, which is comparable to that of RSA with a 3072-bit modulus.
This efficiency advantage has led to the widespread adoption of ECC in both theoretical cryptographic constructions and real-world applications. Specifically, ECC is used for key exchange \cite{diffie1976new} and digital signature schemes \cite{elgamal1985public,johnson2001elliptic} in widely deployed protocols, including transport layer security \cite{blake2006elliptic}, secure shell \cite{stebila2009elliptic}, as well as cryptocurrencies such as Bitcoin \cite{blake1999standards,nakamoto2008peer}. 

In contrast to the computational hardness of ECDLP on classical computers, the security landscape changes fundamentally in the presence of large-scale quantum computers.
In 1994, Shor gave polynomial-time quantum algorithms for integer factorization and discrete logarithms \cite{shor1994algorithms}. Because the ECDLP is a discrete-logarithm problem in a cyclic subgroup of an elliptic-curve group, it can be efficiently solved quantumly by Shor’s algorithm whenever the elliptic-curve group operation can be implemented efficiently. Proos and Zalka subsequently developed an implementation with resource analysis for prime-field elliptic curves \cite{proos2003shor}. Thus, sufficiently large fault-tolerant quantum computers would break conventional ECC based on the hardness of the ECDLP.
At the same time, the practical implementation of quantum algorithms for cryptography problems remains severely constrained by the limitations of near-term quantum hardware. For instance, implementing Shor's algorithm at cryptographically relevant scales requires a fault-tolerant quantum computer with a large number of logical qubits and circuits, as well as sufficiently low physical error rates (or equivalently, achievable logical error rates after error correction) \cite{gidney2021factor,gidney2025factor}.

Motivated by these challenges, a substantial body of work has focused on reducing the quantum resource requirements of Shor's algorithm for integer factorization \cite{regev2025efficient,ragavan2024space,kahanamoku2025jacobi,chevignard2025reducing} and the ECDLP \cite{roetteler2017quantum,haner2020improved,gu2025resource,cryptoeprint:2026/106,cryptoeprint:2026/280,babbush2026ECC,schrottenloher2026optimized}.
In recent studies, three primary resource metrics have received the most attention: (i) space, measured by the number of logical qubits; (ii) quantum gate complexity; and (iii) circuit depth.
These resources are strongly coupled, and optimizing one metric often leads to increased costs in others \cite{roetteler2017quantum,haner2020improved,babbush2026ECC,schrottenloher2026optimized}.
For each single metric, the state-of-the-art results are:

\begin{itemize}
\item In terms of space complexity, Chevignard, Fouque, and Schrottenloher~\cite{cryptoeprint:2026/280} proposed a low-space variant of Shor's ECDLP algorithm based on residue-number-system arithmetic and compressed projective-coordinate outputs.
Their construction achieves the state-of-the-art asymptotic qubit complexity among existing implementations, requiring only \(3.12n+O(\sqrt n)\) logical qubits for an $n$-bit prime field.
However, this reduction in space comes at the cost of a significantly larger gate count of \(\widetilde O(n^4)\).

\item In terms of gate complexity, Babbush et al.~\cite{babbush2026ECC} reported substantially improved resource estimates for solving ECDLP of the curve secp256k1 using \(1175\) logical qubits and \(2^{26.27}\) Toffoli gates in the space-optimized setting, or \(1425\) logical qubits and \(2^{25.94}\) Toffoli gates in the gate-optimized setting. However, their circuit-level details are not publicly disclosed. 
Schrottenloher~\cite{schrottenloher2026optimized} subsequently gave an explicit and reproducible point addition architecture with comparable performance and slightly smaller Toffoli counts for secp256k1.

\item In terms of circuit depth, Kim et al.~\cite{cryptoeprint:2026/106} developed depth-optimized point-addition circuits based on QCSA-based Montgomery multiplication and low-depth binary-EEA inversion. They achieved the lowest circuit
depth with adequate ancilla qubits. In terms of the product of qubit count and full
depth, their qubit-optimized version achieves a 37\%--40\% improvement compared to the previous result by Häner et al.~\cite{haner2020improved}.
\end{itemize}

In this work, we study the space complexity of ECDLP. In particular, our focus is the modular inversion subroutine, which is the dominant source of workspace in affine-coordinate point addition. Our goal is to answer the following question:
\begin{center}
    {\em Can the quantum space requirements for solving the ECDLP be further reduced in practice?}
\end{center}

\subsection{Contributions}

The idea of reducing the space complexity of quantum ECDLP algorithms can be traced back to Proos and Zalka~\cite{proos2003shor}, who introduced a register-sharing strategy and obtained an asymptotic space bound of \(5n+O(\sqrt n)\) for solving the ECDLP over an \(n\)-bit prime field. However, their analysis was given at a high level of abstraction and did not provide explicit reversible circuit constructions or detailed resource estimates; in particular, issues such as reversible implementation and ancilla management were not fully addressed, and the approximate treatment of modular inversion introduces a fidelity loss of \(O(n^{-1})\). 

Our main contribution is to turn this low-space strategy into an exact and fully explicit circuit-level construction. Since modular inversion is the dominant contributor to the overall space, we focus on designing a more space-efficient reversible implementation of this component. Building on the algorithmic framework and register-sharing idea of Proos and Zalka, we refine the representation of the EEA state, introduce length registers and location-controlled arithmetic, and provide explicit circuit constructions with detailed resource estimates. As summarized in Table~\ref{tab:comp_qubit}, our construction achieves the best asymptotic logical-qubit complexity reported so far for ECDLP in the standard quantum circuit model.

\iffull
\begin{table}[ht]
\centering
\small
\begin{tabular}{|c|c|c|}
    \hline
    \textbf{Reference} & \textbf{Space of point addition} & \textbf{Space of modular inversion}\\
    \hline
    \cite{roetteler2017quantum} & $9n + 2\log_2 n + O(1)$ & $7n + 2\log_2 n + O(1)$\\
    \hline
    \cite{haner2020improved} & $8n + 10.2\log_2 n + O(1)$ & $7n + \log_2 n + O(1)$ \\
    \hline
    \cite{cryptoeprint:2026/280} & $3.12n+o(n)$ & not based on modular inversion\\
    \hline
    \cite{schrottenloher2026optimized} & $4.12n + O(\sqrt{n})$ & $2.12n + O(\sqrt{n})$\\
    \hline
    \textbf{This work} & $3n + 6\log_2 n + O(1)$ & $2n + 6\log_2 n + O(1)$ \\
    \hline
\end{tabular}
\caption{Comparison of the number of logical qubits required for ECDLP and modular inversion.}
\label{tab:comp_qubit}
\end{table}
\fi

\ifllncs
\begin{table}[ht]
\centering
\small
\begin{tabular}{|c|c|c|}
    \hline
    \textbf{Source} & \textbf{Number of logical qubits} & \textbf{Remark} \\
    \hline
    \cite{proos2003shor} without register sharing & $5n + 4\log_2 n + O(1)$ & \multirow{2}{*}{\begin{tabular}[c]{@{}c@{}}without explicit \\ circuit implementation\end{tabular}} \\
    \cline{1-2}
    \cite{proos2003shor} with register sharing & $3n + 8\sqrt{n} + 4\log_2 n + O(1)$ & \\
    \hline
    \cite{roetteler2017quantum} & $7n + 2\log_2 n + O(1)$ & \multirow{3}{*}{\begin{tabular}[c]{@{}c@{}}with explicit \\ circuit implementation\end{tabular}}  \\
    \cline{1-2}
    \cite{haner2020improved} & $7n + \log_2 n + O(1)$ & \\
    \cline{1-2}
    \textbf{This work} & $3n + 4\log_2 n + O(1)$ & \\
    \hline
\end{tabular}
\caption{Comparison of the number of logical qubits required for ECDLP.}
\label{tab:comp_qubit}
\end{table}
\fi

Second, we provide resource estimates for our quantum circuit implementation. Specifically, we show that a modular inversion circuit can be implemented using at most $229n^2 + O(n\log_2 n)$ Toffoli gates and $286n^2 + O(n\log_2 n)$ CNOT gates. By integrating this modular inversion circuit into an affine point addition circuit with mid-circuit measurement (Figures \ref{fig:measurement_point_addition} and \ref{fig:measurement_idiv}), we obtain a complexity of $1051n^2 + O(n\log_2 n)$ Toffoli gates per point addition. Furthermore, incorporating the point addition circuit into the complete quantum algorithm for solving the ECDLP using the signed-window technique \cite[Section 5.1]{haner2020improved} yields a total Toffoli-gate complexity of $1056n^3/\log_2 n + O(n^2)$.

We complement the asymptotic comparison with numerical gate-count estimates
and concrete secp256k1 resource estimates. Figure~\ref{fig:toffli_counts} reports the Toffoli and CNOT counts of our modular-inversion and point-addition
circuits as functions of the field size \(n\) and Table~\ref{tab:comp_qubit_concrete} compares the resulting secp256k1 estimates with recent quantum ECDLP implementations. For secp256k1, our construction uses 835 logical qubits, saving more than 250 logical qubits compared to 1098 qubits in the previous best space-optimized estimate of Chevignard et al.~\cite{cryptoeprint:2026/280}. Furthermore, our Toffoli count is two orders of magnitude lower than~\cite{cryptoeprint:2026/280}.
\begin{figure}[ht]
    \centering
    \begin{subfigure}{0.48\linewidth}
        \centering
        \includegraphics[width=\linewidth]{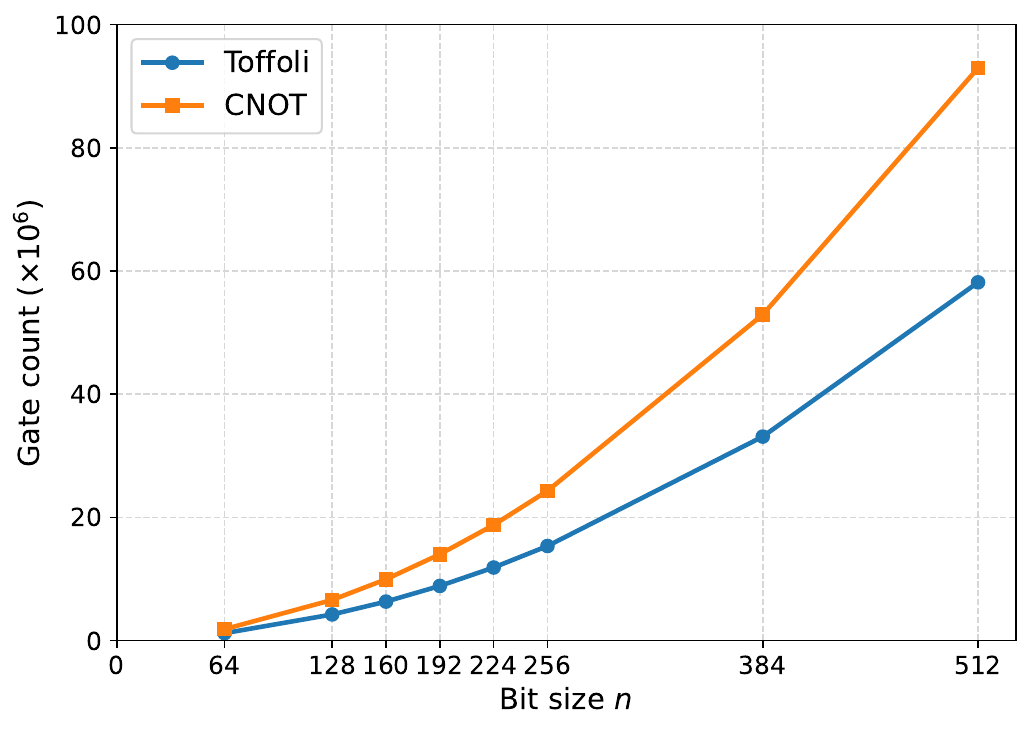}
        \caption{modular inversion count}
        \label{fig:modular_inversion}
    \end{subfigure}
    \hfill
    \begin{subfigure}{0.48\linewidth}
        \centering
        \includegraphics[width=\linewidth]{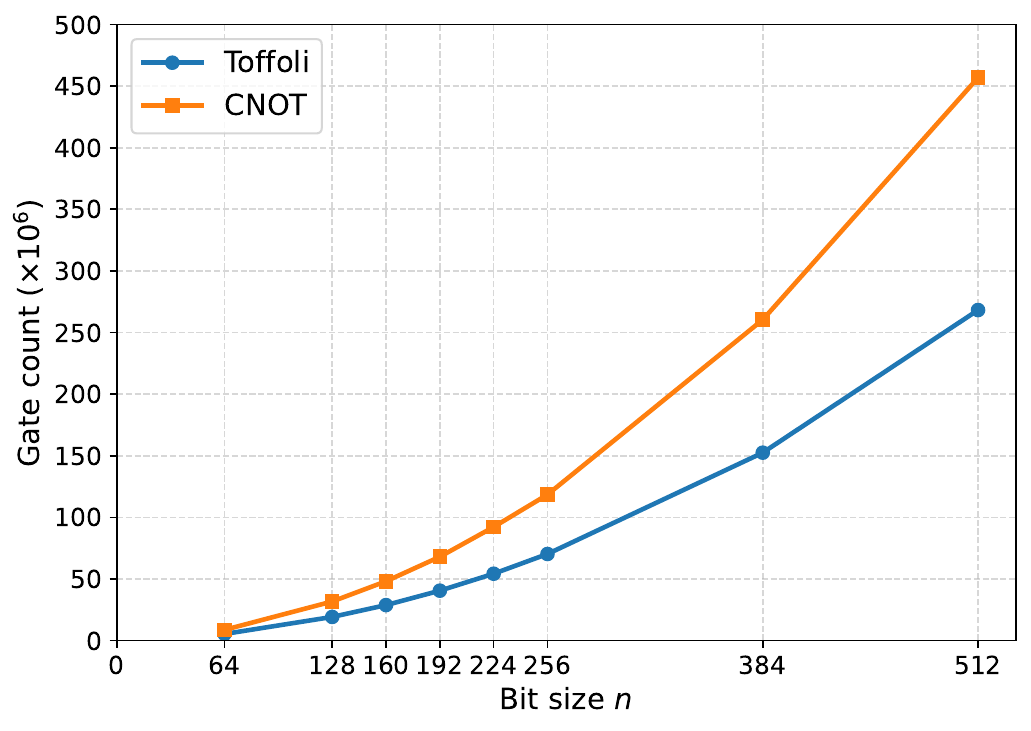}
        \caption{point addition count}
        \label{fig:point_addition}
    \end{subfigure}
    
    \caption{Toffoli gate count and CNOT gate count for modular inversion and point addition.}
    \label{fig:toffli_counts}
\end{figure}

\begin{table}[ht]
\centering
\renewcommand{\arraystretch}{1.05}
\begin{tabular}{@{}l l c c@{}}
\toprule
\textbf{Reference} & \textbf{Type} & \textbf{Qubits} & \textbf{Toffolis} \\
\midrule
\multirow{2}{*}{\cite{babbush2026ECC}}
& Space-optimized for secp256k1 & 1175 & $2^{26.27}$ \\
& Gate-optimized for secp256k1  & 1425 & $2^{25.94}$ \\
\midrule
\multirow{2}{*}{\cite{schrottenloher2026optimized}}
& Space-optimized for secp256k1 & 1192 & $2^{26.11}$ \\
& Gate-optimized for secp256k1  & 1446 & $2^{25.78}$ \\
\midrule
\cite{cryptoeprint:2026/280}
& Space-optimized for secp256k1 & 1098 & $2^{38.10}$ \\
\midrule
\textbf{This work}
& Space-optimized for secp256k1 & 835 & $2^{30.88}$ \\
\bottomrule
\end{tabular}
\caption{Comparison of the concrete costs of circuits for secp256k1.}
\label{tab:comp_qubit_concrete}
\end{table}
\subsection{Overview}
We present an overview of our quantum circuit for solving ECDLP.
The standard implementation of Shor's algorithm consists of $2n+2$ controlled elliptic-curve point additions followed by a quantum Fourier transform (QFT), as illustrated in Figure~\ref{fig:shor_QFT}.
Replacing the QFT with its semiclassical variant reduces the logical-qubit requirement without altering the algorithmic functionality, as shown in Figure~\ref{fig:shor_semiQFT}.

Each elliptic-curve point addition is decomposed into reversible modular arithmetic operations, including modular addition, multiplication, squaring, and division.
Among these, modular division is the dominant contributor to the logical-qubit cost.
Figure~\ref{fig:measurement_point_addition} illustrates the structure of our point-addition circuit, where these operations are encapsulated as modular arithmetic blocks.
In particular, the in-place modular division block is realized by the circuit shown in Figure~\ref{fig:measurement_idiv}, which consists of a forward execution of the extended Euclidean algorithm (EEA) for modular inversion, followed by its uncomputation.

We present each component from a top-down perspective as follows.

\paragraph{Space-efficient reversible algorithm for modular inversion (Section~\ref{sec:EEA}).}
Our modular inversion procedure, based on the EEA, is designed to implement the following quantum transformation:
\begin{align}
    \ket{x}_X\ket{0}_A\ket{0}_S
    \longmapsto
    \ket{x^{-1}\bmod p}_X\ket{0}_A\ket{\Gamma(x)}_S,
    \qquad x\in\mathbb{F}_p^\times,
    \label{eq:inv_api}
\end{align}
where $A$ is an $n$-qubit workspace register, $S$ denotes the remaining ancillary registers, and $\ket{\Gamma(x)}$ is the forward EEA state retained for later uncomputation.

We first revisit the four-phase framework of Proos and Zalka~\cite{proos2003shor} (see Section~\ref{subsec:4phase}), which is well suited for quantum implementation because the number of steps is linear in $n$.
Rather than relying on the (inaccurate) $4.5n$ upper bound stated in~\cite{proos2003shor}, we establish a rigorous bound showing that a complete execution of the EEA requires at most $4\lceil cn\rceil$ steps for every input $x$, where $c = 3 / \log_2(2 + \sqrt{3})$. Specifically, we consider a different roadmap: instead of computing the upper
bound of step number $N$, we fix the total number of steps and then compute the minimum possible value of $p$. To do this, we formulate the problem as finding a slowest growth pattern for products of matrices $M(q)$ with a fixed weight sum $N/4 = \sum_{i = 1}^m (\lfloor log_{2}{q_i}\rfloor + 1).$ Motivated by the asymptotically extremal alternating quotient pattern $(1,2,1,2,\ldots)$ whose normalized growth rate is $(2+\sqrt 3)^{1/3}$ per unit weight, we construct a finite polyhedral antinorm certificate, following the lower spectral radius framework of~\cite{guglielmi2013exact}, to prove that every admissible quotient sequence grows at least this rapidly. Details can be found in Appendix~\ref{app:step_num}.

Section~\ref{subsec:regshare} then develops an explicit and refined register-sharing strategy at the algorithmic level.
In the original construction of~\cite{proos2003shor}, a uniform register allocation is used to accommodate all inputs $x \in \mathbb{F}_p$ in superposition.
This worst-case allocation incurs an $O(\sqrt{n})$ space overhead and necessitates truncating rare outlier cases, introducing an $O(n^{-2})$ fidelity loss.
In contrast, we present an exact allocation scheme that guarantees a deterministic register bound for all inputs $x \in \mathbb{F}_p$.
This reduces the additional space overhead to $O(\log n)$ while preserving exact correctness.

Finally, starting from a baseline reversible realization of this space-efficient EEA, we introduce a sequence of algorithmic optimizations in Section~\ref{subsec:opt} that reorganize the update rules and eliminate redundant arithmetic operations.
We provide complete pseudocode (Algorithm~\ref{alg:opt}) together with representative classical execution traces (Table~\ref{tab:run_example}), which illustrate the reversibility and correctness of each step and clarify how the procedure can be translated into a quantum circuit.

\paragraph{Explicit quantum circuit construction for modular inversion (Section~\ref{sec:circuit}).}
To show that our modular inversion algorithm admits an explicit quantum circuit implementation, we present a detailed construction based on Algorithm~\ref{alg:opt}.
We begin with the overall circuit architecture (Figures~\ref{fig:all_circuit_1} and~\ref{fig:all_circuit_2} in Section~\ref{subsec:all_circuit}), where location-controlled operations arise as key non-standard components.

Before giving their explicit implementations, we establish the necessary tools.
In Section~\ref{subsec:actwindow}, we introduce step-dependent active windows, which restrict the relevant index range at each step and thereby reduce gate complexity.
We then introduce unary iteration (Section~\ref{subsec:unary_iteration}) as a primitive for implementing controlled indexed operations, which enables the realization of these location-controlled components.

Finally, in Section~\ref{subsec:blocks}, we present explicit constructions of the location-controlled circuit blocks, including swap, addition/subtraction, and length-update operations (Figures~\ref{fig:lc_swap}, \ref{fig:lc_add}, \ref{fig:length-write}, and~\ref{fig:length-z-transform}).
In addition to the circuit diagrams, we provide detailed explanations of the underlying techniques and design rationale.

\paragraph{Affine point addition with mid-circuit measurements (Section~\ref{sec:measurement_point_addition}).}
We further show that affine point addition can be implemented using our EEA inversion circuit with only $3n+O(\log n)$ logical qubits by allowing mid-circuit measurements and classical feed-forward operations. The main issue is the in-place division step
\[
    \ket{x}_X\ket{y}_Y\ket{0}_A
    \longmapsto
    \ket{x}_X\ket{y/x}_Y\ket{0}_A .
\]
To perform this map, the EEA inversion circuit in Equation~\eqref{eq:inv_api} is first
applied to $x$. This requires one clean $n$-qubit register $A$ and leaves the forward EEA
state $\Gamma(x)$ to be cleared later. After the quotient $y/x$ has been written into $A$,
however, the register $A$ is no longer clean, so the inverse EEA circuit has no available
$n$-qubit workspace.

As shown in Figure~\ref{fig:measurement_idiv}, we resolve this problem by
measuring the register containing the old value $y$ in the Hadamard basis and resetting it to zero. This makes the register $Y$ available as the clean workspace for the inverse EEA circuit, which recovers $x$ and clears $\Gamma(x)$. 
This measurement is valid because, after $y/x$ has been written into $A$, the old value $y$ is redundant and can be recovered from $x$ and $y/x$. Measuring $\ket{y}$ in the Hadamard basis gives a uniformly random outcome $b\in\{0,1\}^n$, and the only effect on the remaining registers is the phase $(-1)^{b\cdot y}$. Thus measuring the register $Y$ does not destroy coherence. The phase is removed by recomputing $y$ from $x$ and $y/x$, applying the corresponding Pauli $Z$ corrections, and uncomputing the recomputation. 
The same
measurement and correction procedure gives an in-place multiplication circuit. Consequently, affine point addition can be implemented using three $n$-qubit registers $X$, $Y$, and $A$, together with the $O(\log n)$ auxiliary qubits required by the EEA circuit.

We note that mid-circuit measurements have previously been used to reduce either the qubit count or the $T$-gate cost of quantum arithmetic circuits~\cite{gidney2019spooky,kornerup2025tight,luongo2024measurement,kim2021measurement}.
In particular, the lowest-qubit quantum algorithm currently known for the elliptic-curve discrete logarithm problem also makes extensive use of measurements to reduce the number of
qubits~\cite{cryptoeprint:2026/280}. That work does not keep the full elliptic-curve point, nor does it reduce the index registers using a semiclassical Fourier transform.
Instead of computing the full point, it computes a one-bit Legendre symbol from the projective coordinates of that point using residue-number-system arithmetic. Its measurements are used to release intermediate nodes in a spooky pebbling strategy for a binary tree~\cite{gidney2019spooky,kornerup2025tight}. In contrast, we keep the standard affine point-addition structure and use measurements only inside the in-place division and multiplication blocks, where they recycle a register
needed as workspace for running the inverse EEA circuit.

\paragraph{Resource estimation (Section \ref{sec:resource_estimation}).} Putting everything together, we quantify the space complexity and gate costs of the complete construction of Shor's ECDLP circuit.
We first show in Section~\ref{subsec:resource_space} that the modular inversion circuit requires \(2n + 6\lfloor \log_2 n \rfloor + 19\) logical qubits, and that the affine point addition circuit of Section~\ref{sec:measurement_point_addition} can therefore be implemented using \(3n + 6\lfloor \log_2 n \rfloor + 19\) logical qubits.
Section~\ref{subsec:inversion_costs} then combines the per-block gate costs (Section~\ref{subsec:blocks}) with the step-dependent active-window bounds (Section~\ref{subsec:actwindow}), to obtain an upper bound of \(229n^2 + O(n\log_2 n)\) Toffoli gates and \(286n^2 + O(n\log_2 n)\) CNOT gates for modular inversion, as well as an upper bound of \(1051n^2 + O(n\log_2 n)\) Toffoli gates for one affine point addition.

We complement these asymptotic bounds with concrete circuit generation and
blockwise compilation using Qiskit in Section~\ref{subsec:numerical}, reporting numerical
Toffoli and CNOT gate counts for prime-field sizes ranging from \(64\) to
\(512\) bits. Finally, in Section~\ref{subsec:component_costs}, we incorporate signed-window scalar
multiplication and the associated table look-ups into the complete Shor's ECDLP circuit. This yields an overall Toffoli complexity of
$1056n^3/\log_2 n + O(n^2).$ For the \texttt{secp256k1} curve, our construction requires \(835\) logical
qubits and \(2^{30.88}\) Toffoli gates. 

\subsection{Open questions} 

Our work leaves several open questions for future investigation:

\begin{itemize}
    \item \textbf{Gate-count optimization.}
    Our construction prioritizes logical-qubit efficiency and gives explicit reversible circuits for all arithmetic components. However, several building blocks, especially location-controlled arithmetic, length-register updates, and controlled modular operations, may still admit substantial improvements of their gate counts.
    
    \item \textbf{Depth analysis and space-depth trade-offs.}
    This paper focuses primarily on the number of logical qubits and Toffoli/CNOT gate counts, and does not provide a detailed analysis of circuit depth. A natural next step is to study the depth of our modular inversion and point-addition circuits, identify which parts of the computation can be parallelized, and quantify the resulting trade-offs between space, gate count, and depth.

    \item \textbf{Applications beyond elliptic-curve discrete logarithms.}
    We note that modular inversion is a general-purpose arithmetic primitive, not only restricted to ECDLP. It would be useful to investigate whether the space-efficient reversible EEA and register-sharing techniques developed here can improve other quantum algorithms that require modular division or inversion, such as Decoded Quantum Interferometry (DQI)~\cite{khattar2025verifiable}.
    
    \item \textbf{AI-assisted discovery of quantum arithmetic circuits.} 
    An important open question is whether AI systems can enable the automated design of quantum arithmetic circuits, supporting both the discovery of new resource-efficient circuits and the optimization of existing ones. One possible approach is to encode techniques such as windowing, unary iteration, register sharing and mid-circuit measurement as composable modules, and to develop AI-assisted evolutionary frameworks, inspired by systems such as AlphaEvolve~\cite{novikov2025alphaevolvecodingagentscientific} and SimpleTES~\cite{ye2026structuredscalingaidiscovery}, that iteratively explore their replacement and composition under automated correctness verification and resource-aware evaluation. 
\end{itemize}
\section{Preliminaries}

This section provides a very brief discussion of the basic concepts used in this work. We assume that readers have basic knowledge of quantum computation. For a general survey of quantum algorithms for algebraic problems, we recommend the survey written by Childs and van Dam~\cite{childs2010quantum}. 

\subsection{Quantum computing and the Toffoli counts}
\label{subsec:qc}
We write basis quantum states as $\ket{x}$, where $x$ is a bit string, and we model a quantum algorithm as a circuit composed of elementary gates applied to one or more qubits.
A key feature is that circuits can create and transform superpositions of basis states, enabling interference effects that are exploited by quantum algorithms.

At the fault-tolerant logical level, it is common to express circuits over a universal gate set such as Clifford+$T$ \cite{nielsen2010quantum}.
In many architectures the $T$ gate is substantially more expensive than Clifford operations, so circuit cost is frequently summarized by $T$-count, $T$-depth, or related metrics \cite{fowler2012surface}.
For design and validation, however, it is often convenient to use a purely reversible gate basis.

Accordingly, we express the main reversible components as \emph{Toffoli networks}. The Toffoli gate maps
\[
    \ket{x,y,z}\ \longmapsto\ \ket{x,y, z \oplus (x y)},
\]
and it forms a universal primitive for classical reversible computation \cite{nielsen2010quantum}.
A practical benefit of working at the Toffoli-network level is that such networks admit exact realizations over Clifford+$T$ \cite{amy2013meet}, avoiding approximation overhead associated with synthesizing arbitrary unitaries.
Moreover, Toffoli-based reversible circuits can be efficiently simulated on classical inputs at scales far beyond what is possible for general quantum-state simulation, which makes testing and debugging large arithmetic circuits significantly more tractable \cite{haner2016factoring}.
For these reasons, we specify the core arithmetic and elliptic-curve routines using Toffoli and CNOT gates.
\subsection{Elliptic curves and the ECDLP}
\label{subsec:ecdlp}

Let $p$ be a large prime and let $\mathbb{F}_p$ denote the finite field with $p$ elements.
An elliptic curve $E/\mathbb{F}_p$ (in short Weierstrass form) is the set of affine solutions $(x,y)\in \mathbb{F}_p^2$ to an equation $y^2 = x^3 + ax + b$, together with a distinguished point $\mathcal{O}$ at infinity.
The set of $\mathbb{F}_p$-rational points, together with $\mathcal{O}$, is denoted by $E(\mathbb{F}_p)$.

The points on $E(\mathbb{F}_p)$ form an abelian group under the elliptic-curve addition law, with identity element $\mathcal{O}$.
For affine points $P_1=(x_1,y_1)$ and $P_2=(x_2,y_2)$ that are not inverse to each other, the sum $P_3=P_1+P_2$ can be computed via a slope parameter $\lambda$: 
\[
    x_3 = \lambda^2 - x_1 - x_2, \qquad y_3 = \lambda(x_1 - x_3) - y_1,
\]
where $\lambda$ equals the chord slope $\frac{y_2 - y_1}{x_2 - x_1}$ when $P_1\neq P_2$, and the tangent slope $\frac{3x_1^2 + a}{2y_1}$ when $P_1 = P_2$.
Exceptional cases (e.g., $P_i = \mathcal{O}$ or $P_2 = -P_1$) are handled according to the standard group properties.

For an integer $m\ge 1$, the notation $[m]P$ denotes the $m$-fold sum of a point $P$ with itself, i.e. $[m]P = P + P + \cdots + P$, where $P$ occurs $m$ times.
This $m$-fold sum can be extended to all $m\in \mathbb{Z}$ by defining $[0]P = \mathcal{O}$ and $[-m]P = [m](-P)$ for $m \ge 1$.
The operation $m\mapsto [m]P$ is called \emph{scalar multiplication} and is the core primitive used in elliptic-curve cryptography.

Let $P\in E(\mathbb{F}_p)$ generates a cyclic subgroup $\langle P\rangle$ of order $r$, and let $Q\in\langle P\rangle$.
The \emph{elliptic-curve discrete logarithm problem (ECDLP)} is to recover the unique $m\in\{0,1,\ldots,r-1\}$ such that
\[
    Q = [m]P.
\]
In classical settings, the best generic algorithms require on the order of $\Theta(\sqrt{p})$ group operations \cite{galbraith2016recent}, which is exponential in the bit length $\log_2 p$.
\subsection{Shor's quantum algorithm for ECDLP}
\label{subsec:shor}

Shor's discrete-logarithm algorithm for solving ECDLP applies to elliptic curve over any finite abelian group and therefore to $E(\mathbb{F}_p)$.
Let $n = \lfloor\log_2 p\rfloor + 1$ be the bit-length of $p$.
The quantum procedure uses two exponent registers and one register storing an elliptic curve point.

Firstly, initialize two $(n+1)$-qubit registers\footnote{Hasse's bound \cite{hasse1936theorie} indicates that $\mathrm{ord}(P) \le \#E(\F_p) \le p + 1 + 2\sqrt{p}$, which can be represented by at most $n + 1$ bits.} to be $\ket{0^{n+1}, 0^{n+1}}$, and apply Hadamard gate to all the qubits to obtain a uniform superposition over pairs $(k, \ell)\in \{0, 1, \ldots, 2^{n+1} - 1\}^2$.
Next, coherently compute the elliptic curve group action into an accumulator as
\begin{equation}\label{eqn:shor}
    \sum_{k, \ell = 0}^{2^{n+1} - 1}
    \ket{k, \ell}\ket{\mathcal{O}}
    \;\longmapsto\;
    \sum_{k, \ell = 0}^{2^{n+1} - 1}
    \ket{k, \ell}\ket{[k]P + [\ell]Q}.
\end{equation}
After this step, apply Quantum Fourier Transform $\QFT_{2^{n+1}}$ to both exponent registers and then measure them.
Finally, a classical post-processing procedure uses the measurement outcomes to reconstruct the discrete logarithm $m$ with high probability, as shown in \cite{shor1994algorithms}.

In terms of resource, the dominant cost arises from the coherent group evaluation, i.e., the double-scalar multiplication on elliptic curve points.
To reduce the qubit cost associated with QFT, its full circuit on the exponent registers can be replaced by a semiclassical variant \cite{griffiths1996semiclassical}.
This semiclassical Fourier transform performs measurements during the computation, reusing qubits and applying conditional phase rotations based on previously observed outcomes.
The overall quantum circuits of Shor's algorithm for ECDLP using QFT (resp. semiclassical QFT) are shown in Figure \ref{fig:shor_QFT} (resp. Figure \ref{fig:shor_semiQFT}).
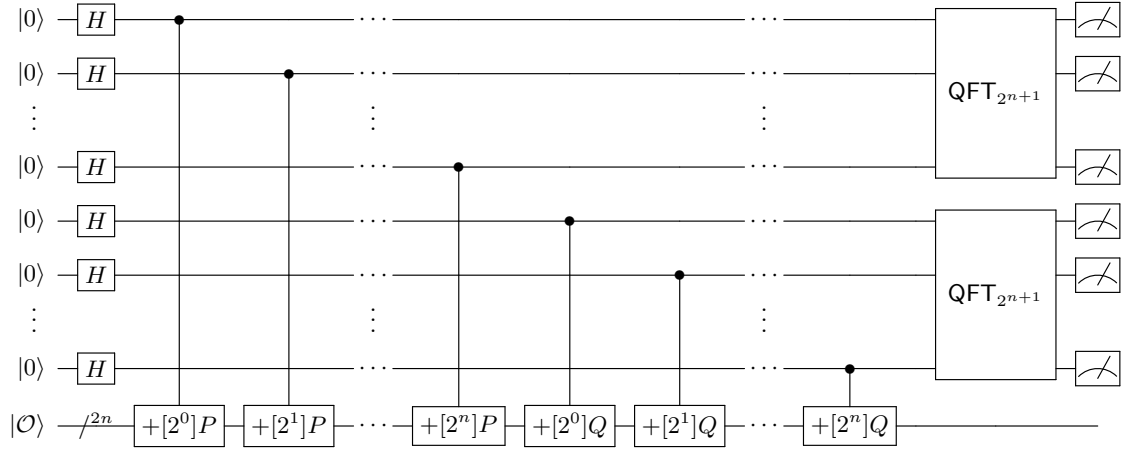
\begin{figure}[ht]
\centering
\footnotesize
\scalebox{\iffull 1 \else 0.85 \fi}{
\(\Qcircuit @C=0.8em @R=0.8em {
    & \lstick{\ket{0}} & \gate{H} & \ctrl{10} & \qw & \qw
    & \cdots & & \qw & \qw & \qw 
    & \qw & \cdots & & \qw & \qw 
    & \multigate{4}{\QFT_{2^{n+1}}} & \meter \\
    & \lstick{\ket{0}} & \gate{H} & \qw & \ctrl{9} & \qw 
    & \cdots & & \qw & \qw & \qw 
    & \qw & \cdots & & \qw & \qw 
    & \ghost{\QFT_{2^{n+1}}} & \meter \\
    & \lstick{\vdots\ } & & & & 
    & \vdots & & & &
    & & \vdots & & & 
    & & \\
    & & & & & & & & & & & & & & & & & \\
    & \lstick{\ket{0}} & \gate{H} & \qw & \qw & \qw 
    & \cdots & & \ctrl{6} & \qw & \qw 
    & \qw & \cdots & & \qw & \qw 
    & \ghost{\QFT_{2^{n+1}}} & \meter \\
    & \lstick{\ket{0}} & \gate{H} & \qw & \qw & \qw 
    & \cdots & & \qw & \ctrl{5} & \qw 
    & \qw & \cdots & & \qw & \qw 
    & \multigate{4}{\QFT_{2^{n+1}}} & \meter \\
    & \lstick{\ket{0}} & \gate{H} & \qw & \qw & \qw 
    & \cdots & & \qw & \qw & \ctrl{4} 
    & \qw & \cdots & & \qw & \qw 
    & \ghost{\QFT_{2^{n+1}}} & \meter \\
    & \lstick{\vdots\ } & & & & 
    & \vdots & & & &
    & & \vdots & & & 
    & & \\
    & & & & & & & & & & & & & & & & \\
    & \lstick{\ket{0}} & \gate{H} & \qw & \qw & \qw 
    & \cdots & & \qw & \qw & \qw 
    & \qw & \cdots & & \ctrl{1} & \qw 
    & \ghost{\QFT_{2^{n+1}}} & \meter \\
    & \lstick{\ket{\mathcal{O}}} & \qw\slash^{2n} & \gate{+[2^0]P} & \gate{+[2^1]P} & \qw 
    & \cdots & & \gate{+[2^n]P} & \gate{+[2^0]Q} & \gate{+[2^1]Q}
    & \qw & \cdots & & \gate{+[2^n]Q} & \qw 
    & \qw & \qw \\
}\)}
\caption{The overall quantum circuit of Shor’s algorithm for solving ECDLP using QFT. The qubits, from top to bottom, correspond to the exponent registers containing $k$ and $\ell$ in Equation \ref{eqn:shor} (ordered from lower-order bits to higher-order bits), and to the register that stores the elliptic curve point accumulator.
}
\label{fig:shor_QFT}
\end{figure}

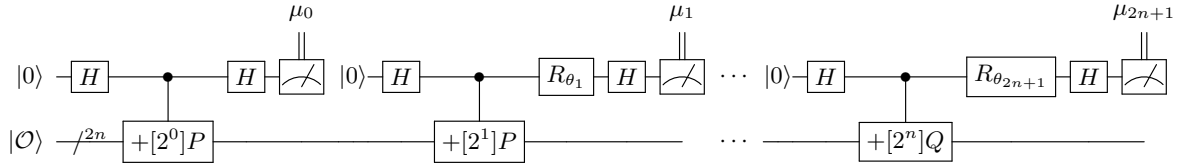
\begin{figure}[ht]
\centering
\footnotesize
\scalebox{\iffull 1 \else 0.8 \fi}{
\(\Qcircuit @C=0.6em @R=1.0em {
    & & & & & \ustick{\mu_0}\cwx[1] 
    & & & & &
    & & & \ustick{\mu_1}\cwx[1] & &
    & & & & &
    & & & & \ustick{\mu_{2n+1}}\cwx[1] \\
    & \lstick{\ket{0}} & \gate{H} & \ctrl{1} & \gate{H} & \meter 
    & & \ket{0} & & \gate{H} & \ctrl{1}
    & \gate{R_{\theta_1}} & \gate{H} & \meter & & \cdots
    & & & \ket{0} & & \gate{H}
    & \ctrl{1} & \gate{R_{\theta_{2n+1}}} & \gate{H} & \meter \\
    & \lstick{\ket{\mathcal{O}}} & \qw\slash^{2n} & \gate{+[2^0]P} & \qw & \qw
    & \qw & \qw & \qw & \qw & \gate{+[2^1]P} 
    & \qw & \qw & \qw & & \cdots
    & & & \qw & \qw & \qw 
    & \gate{+[2^n]Q} & \qw & \qw & \qw \\
}\)}
\caption{The overall quantum circuit of Shor’s algorithm for solving ECDLP using semiclassical QFT. The gates $R_{\theta_i}$ denote rotation gates with rotation angle $\theta_i = \sum_{j = 0}^{i-1} 2^{i-j}\mu_j$, where the values $\mu_j\in \{0, 1\}$ are outcomes obtained in previous measurements. By employing the semiclassical QFT, the exponent register requires only a single qubit, resulting in a reduction of $2n+1$ qubits compared to the standard implementation.}
\label{fig:shor_semiQFT}
\end{figure}
\section{Space-Efficient Extended Euclidean Algorithm}
\label{sec:EEA}

As discussed in previous sections, the dominant cost in designing quantum circuits for solving ECDLP via Shor's algorithm arises from coherent group additions on elliptic curve points. Among these operations, modular inversion is the most resource-consuming component. In this section, we present a space-efficient and reversible algorithm for modular inversion based on the Extended Euclidean Algorithm. This design is suitable for further translation into a quantum circuit, as detailed in Section \ref{sec:circuit}.

Suppose that $p > x$ are two positive integers. The well-known \emph{Euclidean algorithm} can be used to compute the greatest common divisor (GCD) of $p$ and $x$. The algorithm initializes $r_0 := p$ and $r_1 := x$, and for iterations $i = 1, 2, \ldots$, it repeatedly divides $r_{i-1}$ by $r_i$ to obtain the quotient $q_i = \lfloor r_{i-1} / r_i \rfloor$ and the remainder $r_{i+1} = r_{i-1} - q_i r_i$. The process terminates when $r_{k+1} = 0$, at which point the greatest common divisor of $p$ and $x$ equals $r_k$.

The \emph{Extended Euclidean Algorithm (EEA)} not only computes the GCD of $p$ and $x$, but also outputs the modular inverse $x^{-1} \bmod p$ when $p$ and $x$ are coprime. In addition to computing the sequence $r_i$, it also initializes $t_0 := 0$ and $t_1 := 1$, and for each iteration $i = 1, 2, \ldots$, updates $t_{i+1} = t_{i-1} + q_i t_i$ after obtaining $q_i = \lfloor r_{i-1} / r_i \rfloor$. When the algorithm first encounters $r_k = 1$, the modular inverse is given by $(-1)^{k-1} t_k \bmod p$.

When realizing the Extended Euclidean algorithm on a quantum computer to compute the modular inverse of $x$ with a prime modulus $p$ of $n$ binary digits, the primary challenge lies in handling superpositions of $x$. The number of iterations required varies from $1$ to $O(n)$ depending on the input, implying that a straightforward implementation would require $O(n)$ iterations. Each iteration involves divisions and multiplications on $n$-digit numbers, leading to an overall gate complexity of $O(n^3)$.

In what follows, we describe the dedicated four-phase algorithm framework first proposed in \cite{proos2003shor}, followed by our own algorithmic optimizations and the details of our optimized circuit implementation.

\subsection{The four-phase algorithm framework}
\label{subsec:4phase}

The four-phase algorithm framework was originally introduced in \cite{proos2003shor}. It operates on six quantum registers, denoted as $r_{i-1}$, $r_i$, $t_{i-1}$, $t_i$, $q_i$, and $\ell$, where the first five registers correspond to the variables used in the $i$-th iteration of EEA. To compute the quotient $q_i = \lfloor r_{i-1} / r_i \rfloor$, the algorithm performs a binary long-division procedure consisting solely of bit shifts and long additions or subtractions. The algorithm proceeds through the following four phases:
\begin{itemize}
    \item \textbf{Phase 1.} The register containing $r_i$ is repeatedly shifted left by one bit, while the shift counter $\ell$ is incremented at each step. This continues until the inequality $2^\ell r_i > r_{i-1}$ holds. Conceptually, this phase determines the largest power of two by which $r_i$ can be multiplied without exceeding $r_{i-1}$. Equivalently, after the loop terminates, $\ell_0 := \ell - 1$ gives the highest bit position in the binary representation of $q_i$.

    \item \textbf{Phase 2.} The register containing $r_i$ is then shifted right by one bit per step, with $\ell$ decremented accordingly. At each step, the registers corresponding to $r_{i-1}$ and $q_i$ are updated as
    \[
        (r_{i-1} - q' r_i, q') 
        \to \left(r_{i-1} - (q' + 2^\ell q_{i, \ell}) r_i, q' + 2^\ell q_{i, \ell}\right),
    \]
    depending on whether $r_{i-1} - q' r_i$ is at least $2^{\ell} r_i$. Here $q'$ denotes the current partial quotient, given by $q' = 2^{\ell_0} q_{i, \ell_0} + 2^{\ell_0 - 1} q_{i, \ell_0 - 1} + \cdots + 2^{\ell + 1} q_{i, \ell + 1}$, representing the most significant bits accumulated so far in the binary expansion
    \[
        q_i = 2^{\ell_0} q_{i, \ell_0} + 2^{\ell_0 - 1} q_{i, \ell_0 - 1} + \cdots + q_{i, 0}.
    \]
    This iterative process ultimately yields the transformation
    \[
        (r_{i-1}, 0) 
        \to (r_{i-1} - q_i r_i, q_i),
    \]
    completing the computation of $q_i$.

    \item \textbf{Phase 3.} The register containing $t_i$ is shifted left by one bit per step, and $\ell$ is incremented correspondingly. The registers associated with $t_{i-1}$ and $q_i$ are updated as
    \[
        (t_{i-1} + q'' t_i, q') 
        \to \left(t_{i-1} + (q'' + 2^\ell q_{i, \ell}) t_i, q'\right),
    \]
    where the update depends on the bit value $q_{i, \ell}$ of $q_i$. After each step, $q_{i, \ell}$ is reset to $0$, depending on whether $t_{i-1} + (q'' + 2^\ell q_{i, \ell}) t_i > 2^{\ell} t_i$. Here $q'$ and $q''$ represent the most and least significant portions of the quotient, respectively:
    \[
        q' = 2^{\ell_0} q_{i, \ell_0} + 2^{\ell_0 - 1} q_{i, \ell_0 - 1} + \cdots + 2^{\ell} q_{i, \ell}, \quad
        q'' = 2^{\ell - 1} q_{i, \ell - 1} + 2^{\ell - 2} q_{i, \ell - 2} + \cdots + q_{i, 0}.
    \]
    After completing all steps, the registers corresponding to $t_{i-1}$ and $q_i$ are updated as
    \[
        (t_{i-1}, q_i) 
        \to (t_{i-1} + q_i t_i, 0).
    \]

    \item \textbf{Phase 4.} Finally, the register containing $t_i$ is repeatedly shifted right by one bit, and $\ell$ is decremented at each step. The process terminates once $\ell = 0$.
\end{itemize}

Upon completion of all four phases, SWAP operations are applied between the register pairs corresponding to $(r_{i-1}, r_i)$ (which is currently $(r_{i+1}, r_i)$) and the register pairs corresponding to  $(t_{i-1}, t_i)$ (which is currently $(t_{i+1}, t_i)$), respectively. This completes one full iteration of the EEA implementation.

We emphasize that, during the execution of this four-phase algorithm, the correspondence between the iteration index in EEA and the phase number may vary depending on the input value $x$. In other words, the algorithm’s internal progression through iterations and phases is input-dependent. To avoid ambiguity in the subsequent descriptions, we use ``step'' to represent any sub-iteration occurring within any phase of an EEA iteration. The example in Table \ref{tab:step_status} illustrates how the step status evolves for different input values under this four-phase framework.
\begin{table}[htbp]
\centering
\footnotesize
\scalebox{\iffull 1 \else 0.85 \fi}{
\begin{tabular}{cccc}
\hline
\textbf{Input}  & $x_1$ & $x_2$ & $x_3$ \\
\hline
\textbf{Step 1} & Iteration 1, Phase 1 & Iteration 1, Phase 1      & Iteration 1, Phase 1 \\
\hline
\textbf{Step 2} & Iteration 1, Phase 1 & Iteration 1, Phase 1      & Iteration 1, Phase 1 \\
\hline
\textbf{Step 3} & Iteration 1, Phase 2 & Iteration 1, Phase 1      & Iteration 1, Phase 1 \\
\hline
\textbf{Step 4} & Iteration 1, Phase 2 & Iteration 1, Phase 1      & Iteration 1, Phase 2 \\
\hline
\textbf{Step 5} & Iteration 1, Phase 3 & Iteration 1, Phase 1      & Iteration 1, Phase 2 \\
\hline
\textbf{Step 6} & Iteration 1, Phase 3 & Iteration 1, Phase 2      & Iteration 1, Phase 2 \\
\hline
\textbf{Step 7} & Iteration 1, Phase 4 & Iteration 1, Phase 2      & Iteration 1, Phase 3 \\
\hline
\textbf{Step 8} & Iteration 1, Phase 4 (and SWAP) & Iteration 1, Phase 2      & Iteration 1, Phase 3 \\
\hline
\textbf{Step 9} & Iteration 2, Phase 1 & Iteration 1, Phase 2      & Iteration 1, Phase 3 \\
\hline
\textbf{Step 10} & Iteration 2, Phase 1 & Iteration 1, Phase 2      & Iteration 1, Phase 4 \\
\hline
\textbf{Step 11} & Iteration 2, Phase 1 & Iteration 1, Phase 3      & Iteration 1, Phase 4 \\
\hline
\textbf{Step 12} & Iteration 2, Phase 2 & Iteration 1, Phase 3      & Iteration 1, Phase 4 (and SWAP) \\
\hline
\textbf{Step 13} & Iteration 2, Phase 2 & Iteration 1, Phase 3      & Iteration 2, Phase 1 \\
\hline
\textbf{Step 14} & Iteration 2, Phase 2 & Iteration 1, Phase 3      & Iteration 2, Phase 1 \\
\hline
\textbf{Step 15} & Iteration 2, Phase 3 & Iteration 1, Phase 3      & Iteration 2, Phase 2 \\
\hline
\textbf{Step 16} & Iteration 2, Phase 3 & Iteration 1, Phase 4      & Iteration 2, Phase 2 \\
\hline
\textbf{$\cdots$} & $\cdots$ & $\cdots$ & $\cdots$ \\
\hline
\textbf{Step 100} & Iteration 9, Phase 3 & Iteration 5, Phase 2      & Iteration 8, Phase 4 \\
\hline
\textbf{$\cdots$} & $\cdots$ & $\cdots$ & $\cdots$ \\
\hline
\end{tabular}}
\caption{An example illustrating that the EEA iteration index and phase number may vary depending on the input value $x$.}
\label{tab:step_status}
\end{table}

Although the sequence of step statuses differs for each input $x$, the total number of steps required to reach an index $k$ such that $r_{k-1} = 1, r_k = 0$ can be bounded within the interval $[4n, 4\lceil cn\rceil]$ with $c = 3 / \log_2(2 + \sqrt3)$, as formally proven in Appendix \ref{app:step_num}. This bounded-step property is crucial for quantum implementation: it implies that a quantum circuit corresponding to one ``step'' can be repeated at most $4\lceil cn\rceil < 6.316n$ times to definitely obtain the modular inverse of any input $x$ in superposition. Since each step involves arithmetic operations such as shifts, additions, and subtractions, which require $O(n)$ quantum gates, the overall gate complexity of the algorithm is reduced to $O(n^2)$. 

Keeping the above discussion in mind, the forward space-efficient EEA subroutine for computing the modular inverse $x^{-1} \bmod p$, where $p$ is a prime represented with $n$ binary digits and $x \in \{1, 2, \ldots, p - 1\}$, is summarized in Algorithm \ref{alg:full_EEA}. In this context: (i) the register allocation strategy is described in Section \ref{subsec:regshare}; (ii) Algorithm \ref{alg:opt} provides the optimized stepwise iteration designed to minimize gate complexity and ensure quantum reversibility, as further discussed in Sections \ref{subsec:regshare} and \ref{subsec:opt}; (iii) when $x > p/2$, we run the EEA on $p-x$ instead of $x$ and record this choice in the iteration-parity qubit.
\begin{algorithm}[htbp]
\small
\caption{Forward space-efficient EEA inversion subroutine}
\begin{algorithmic}
\Require Quantum registers $\ket{x}\ket{0^n}\ket{0^m}$, where the first $n$-qubit register is the input/output register, the second $n$-qubit register is a large workspace, and the remaining $m$ qubits serve as auxiliary qubits.
\Ensure Output state $\ket{x^{-1}\bmod p}\ket{0^n}\ket{\Gamma(x)}$, where $\ket{\Gamma(x)}$ denotes the remaining forward EEA state.

\State Initialize register \texttt{Work1} with $(n + 3)$ qubits as $\ket{100, p}$ \Comment{Using $(n + 3)$ auxiliary qubits}
\State Initialize register \texttt{Work2} with $(n + 3)$ qubits as $\ket{000, x}$ \Comment{Pad $x$ to $n$ bits on the left}
\State Initialize \texttt{Length} registers to $\ket{\ell_t := 1, \ell_q := 0, \ell_{r'} := \text{length of } x, \ell_s := 0}$
\State Initialize \texttt{Control} registers to $\ket{\texttt{Phase1} := 0, \texttt{Phase2} := 0, \texttt{Sign} := 0, \texttt{Iter} := 0}$

\If{$x > p / 2$} \Comment{Ensures that the algorithm is invertible}
    \State $\texttt{Iter} \gets \texttt{Iter} \oplus 1$
\EndIf
\If{$\texttt{Iter} = 1$}
    \State $x \gets p - x$
\EndIf

\For{$i = 1, 2, \cdots, 4\lceil cn\rceil$} \Comment{Main iterative loop performing the optimized algorithm steps}
    \State Execute Algorithm~\ref{alg:opt} on the quantum registers
\EndFor

\State Keep the $t'$ part of \texttt{Work2} as the output value
\If{$\texttt{Iter} = 0$} \Comment{Convert $t'$ to the positive inverse modulo $p$}
    \State $t' \gets p - t'$
\EndIf
\State Reset the input-independent terminal contents of the large workspace to $\ket{0^n}$
\State Denote the remaining forward EEA state by $\ket{\Gamma(x)}$

\end{algorithmic}
\label{alg:full_EEA}
\end{algorithm}

\subsection{Space-efficient algorithm by register sharing}
\label{subsec:regshare}

In \cite{proos2003shor}, the authors introduced the concept of \emph{register sharing}. The central idea of register sharing is that, during the execution of the EEA, the sequence $\{r_i\}$ is monotonically decreasing while the sequence $\{t_i\}$ is monotonically increasing. This complementary behavior allows both values to occupy the same quantum register at different stages of computation, thereby reducing the overall space requirement. The identity
\[
    r_{i-1}t_i + r_it_{i-1} = p
\]
implies that two quantum registers, each consisting of $(n + 2)$ qubits, are sufficient: one to store the pair $(r_{i-1}, t_i)$ and the other to store $(r_i, t_{i-1})$.

We extend this register-sharing idea further. Specifically, we observe that a single quantum register of $(n + 2)$ qubits can be used to store the triple $(r_{i-1}, t_i, q_i)$ simultaneously, including the intermediate states that arise during phases 2 and 3 of the algorithm. These phases are responsible for computing $q_i = \lfloor r_{i-1} / r_i \rfloor$ and updating the variables
\[
    r_{i-1} \leftarrow r_{i-1} - q_i r_i, t_{i-1} \leftarrow t_{i-1} + q_i t_i.
\]
The detailed allocation of quantum registers is described as follows, and illustrated in Figure \ref{fig:regshare}.
\begin{itemize}
    \item \textbf{\texttt{Work1} register (of $(n + 3)$ qubits).}  
    The left-most $\ell_t + 1$ qubits encode the value $t := t_i$ in little-endian order (here we append a zero to the rightmost position of $t$ in order to simplify the circuit implementation). The following $\ell_q$ qubits store the intermediate quotient value $q := q'$ (only the effective most significant bits) in phases 2 and 3, in big-endian order. The remaining qubits hold the value $r := r_{i-1}$ (or its updated form $r := r_{i-1} - q'r_i$) in big-endian order.

    \item \textbf{\texttt{Work2} register (of $(n + 3)$ qubits).}
    The right-most $\ell_{r'}$ qubits store $r' := r_i$ in big-endian order, while the remaining qubits store the value $t' := t_{i-1}$ (and its intermediate update $t' := t_{i-1} + q''t_i$) in little-endian order. During phases 2 and 3, this register may be circularly shifted left by $\ell_s$ positions. In such cases, the value $t'$ may span both ends of the register, effectively splitting into two contiguous parts.

    \item \textbf{\texttt{Length} registers.}   
    Four auxiliary registers are used to manage variable-length storage. Each of them contains $\lfloor \log_2 n\rfloor + 1$ qubits. They store the length indicators $\ell_t, \ell_q, \ell_{r'}$, and the shift counter $\ell_s$ respectively. These \texttt{Length} registers mark the logical boundaries within the \texttt{Work1} and \texttt{Work2} registers and enable coherent manipulation of superposed inputs of varying lengths.

    \item \textbf{\texttt{Control} registers.}  
    Two single-qubit \texttt{Phase} registers indicate the current phase of the four-phase algorithmic framework, and one single-qubit \texttt{Iter} register that records the parity of the current EEA iteration number. Additionally, there are two auxiliary single-qubit registers, \texttt{Sign} and \texttt{Ctrl}, used for intermediate flag and control operations.
\end{itemize}

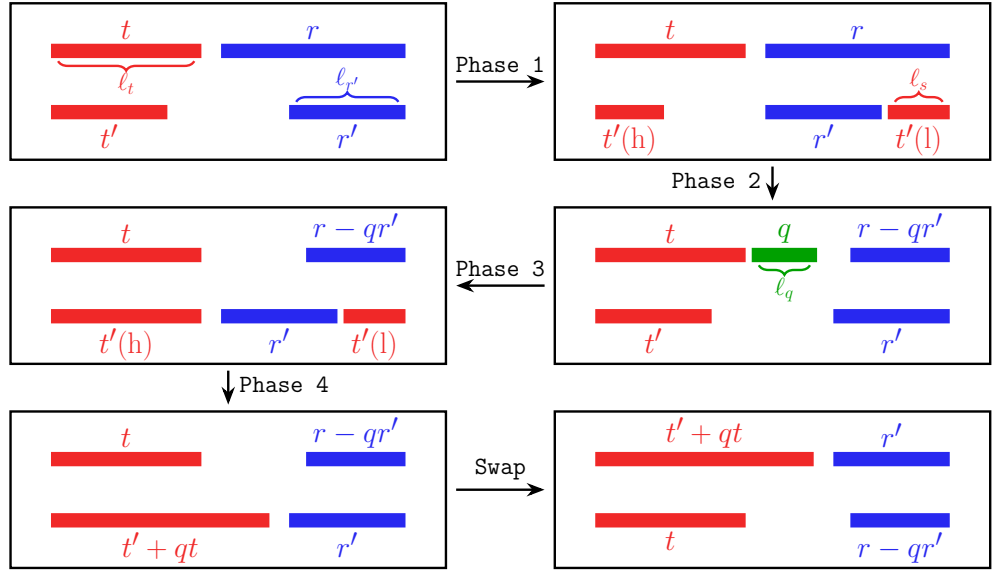
\begin{figure}[ht]
\centering
\scalebox{\iffull 0.9 \else 0.75 \fi}{
\begin{tikzpicture}[
    scale=0.5,
    transform shape,
    font=\Huge,
    >=Stealth,
    segR/.style={draw=myred, line width=6pt, line cap=rect},
    segB/.style={draw=myblue, line width=6pt, line cap=rect},
    segG/.style={draw=mygreen, line width=6pt, line cap=rect},
    braceR1/.style={decorate, decoration={brace, amplitude=4pt, mirror}, draw=myred, line width=1pt},
    braceR2/.style={decorate, decoration={brace, amplitude=4pt}, draw=myred, line width=1pt},
    braceB/.style={decorate, decoration={brace, amplitude=4pt}, draw=myblue, line width=1pt},
    braceG/.style={decorate, decoration={brace, amplitude=4pt, mirror}, draw=mygreen, line width=1pt},
    phaseArrow/.style={->, line width=1pt},
]

\usetikzlibrary{decorations.pathreplacing,arrows.meta,positioning,calc}
\usetikzlibrary{arrows.meta,decorations.pathreplacing,calc}

\definecolor{myred}{RGB}{240,40,40}
\definecolor{myblue}{RGB}{40,40,240}
\definecolor{mygreen}{RGB}{0,160,0}

\def\Xl{0.0}
\def\Xre{10.0}
\def\Ytop{0.0}
\def\Ybot{-1.8}
\def\rDist{1.4}

\def\Lt{4.0}
\def\Ltp{3.0}
\def\Lr{5.0}
\def\Lrp{3.0}
\def\Lrs{2.5}
\def\Lq{1.5}
\def\Ls{1.4}
\def\Gap{0.6}

\pgfmathsetmacro{\Xr}{\Xre-\Lr}     
\pgfmathsetmacro{\Xrp}{\Xre-\Lrp}   
\pgfmathsetmacro{\Xrs}{\Xre-\Lrs}   

\coordinate (A) at (0,0);        
\coordinate (B) at (16,0);      
\coordinate (C) at (16,-6);   
\coordinate (D) at (0,-6);    
\coordinate (E) at (0,-12);      
\coordinate (F) at (16,-12);    

\begin{scope}[shift={(A)}]
    \draw[line width=1.0pt] (\Xl-\rDist,\Ytop+\rDist) rectangle (\Xre+\rDist,\Ybot-\rDist);
    
    \coordinate (tA1) at (\Xl,\Ytop);
    \coordinate (tA2) at (\Xl+\Lt,\Ytop);
    \draw[segR] (tA1) -- (tA2);
    \node[above=6pt, text=myred] at ($(tA1)!0.5!(tA2)$) {$t$};
    \draw[braceR1] (\Xl,\Ytop-0.3) -- (\Xl+\Lt,\Ytop-0.3) node[midway, below=6pt, text=myred] {\huge{$\ell_t$}};
    
    \coordinate (rA1) at (\Xr,\Ytop);
    \coordinate (rA2) at (\Xre,\Ytop);
    \draw[segB] (rA1) -- (rA2);
    \node[above=6pt, text=myblue] at ($(rA1)!0.5!(rA2)$) {$r$};
    
    \coordinate (tpA1) at (\Xl,\Ybot);
    \coordinate (tpA2) at (\Xl+\Ltp,\Ybot);
    \draw[segR] (tpA1) -- (tpA2);
    \node[below=6pt, text=myred] at ($(tpA1)!0.5!(tpA2)$) {$t'$};
    
    \coordinate (rpA1) at (\Xrp,\Ybot);
    \coordinate (rpA2) at (\Xre,\Ybot);
    \draw[segB] (rpA1) -- (rpA2);
    \node[below=6pt, text=myblue] at ($(rpA1)!0.5!(rpA2)$) {$r'$};
    \draw[braceB] (\Xrp,\Ybot+0.3) -- (\Xre,\Ybot+0.3) node[midway, above=6pt, text=myblue] {\huge{$\ell_{r'}$}};
\end{scope}

\draw[phaseArrow] ($(A)+(11.7,-0.9)$) -- ($(B)+(-1.7,-0.9)$)
node[midway, above=4pt] {\huge{\texttt{Phase 1}}};

\begin{scope}[shift={(B)}]
    \draw[line width=1.0pt] (\Xl-\rDist,\Ytop+\rDist) rectangle (\Xre+\rDist,\Ybot-\rDist);
    
    \coordinate (tB1) at (\Xl,\Ytop);
    \coordinate (tB2) at (\Xl+\Lt,\Ytop);
    \draw[segR] (tB1) -- (tB2);
    \node[above=6pt, text=myred] at ($(tB1)!0.5!(tB2)$) {$t$};
    
    \coordinate (rB1) at (\Xr,\Ytop);
    \coordinate (rB2) at (\Xre,\Ytop);
    \draw[segB] (rB1) -- (rB2);
    \node[above=6pt, text=myblue] at ($(rB1)!0.5!(rB2)$) {$r$};
    
    \coordinate (sL1) at (\Xl,\Ybot);
    \coordinate (sL2) at (\Xl+\Ltp-\Ls,\Ybot);
    \draw[segR] (sL1) -- (sL2);
    \node[below=6pt, text=myred] at ($(sL1)!0.5!(sL2)$) {$t'$(h)};

    \coordinate (rpB1) at (\Xr,\Ybot);               
    \coordinate (rpB2) at (\Xre-\Ls-\Gap,\Ybot); 
    \draw[segB] (rpB1) -- (rpB2);
    \node[below=6pt, text=myblue] at ($(rpB1)!0.5!(rpB2)$) {$r'$};

    \coordinate (sR1) at (\Xre-\Ls,\Ybot);
    \coordinate (sR2) at (\Xre,\Ybot);
    \draw[segR] (sR1) -- (sR2);
    \node[below=6pt, text=myred] at ($(sR1)!0.5!(sR2)$) {$t'$(l)};
    \draw[braceR2] (\Xre-\Ls,\Ybot+0.3) -- (\Xre,\Ybot+0.3) node[midway, above=6pt, text=myred] {\huge{$\ell_s$}};
\end{scope}

\draw[phaseArrow] ($(B)+(5.0,-3.4)$) -- ($(C)+(5.0,1.6)$);
\node[anchor=east] at ($(B)+(4.8,-3.8)$) {\huge{\texttt{Phase 2}}};

\begin{scope}[shift={(C)}]
    \draw[line width=1.0pt] (\Xl-\rDist,\Ytop+\rDist) rectangle (\Xre+\rDist,\Ybot-\rDist);
    
    \coordinate (tC1) at (\Xl,\Ytop);
    \coordinate (tC2) at (\Xl+\Lt,\Ytop);
    \draw[segR] (tC1) -- (tC2);
    \node[above=6pt, text=myred] at ($(tC1)!0.5!(tC2)$) {$t$};

    \coordinate (rsC1) at (\Xrs,\Ytop);
    \coordinate (rsC2) at (\Xre,\Ytop);
    \draw[segB] (rsC1) -- (rsC2);
    \node[above=6pt, text=myblue] at ($(rsC1)!0.5!(rsC2)$) {$r-qr'$};

    \coordinate (qC1) at (\Xl+\Lt+\Gap,\Ytop);
    \coordinate (qC2) at (\Xl+\Lt+\Gap+\Lq,\Ytop);
    \draw[segG] (qC1) -- (qC2);
    \node[above=6pt, text=mygreen] at ($(qC1)!0.5!(qC2)$) {$q$};
    \draw[braceG] ($(qC1)+(0,-0.3)$) -- ($(qC2)+(0,-0.3)$) node[midway, below=6pt, text=mygreen] {\huge{$\ell_q$}};

    \coordinate (tpC1) at (\Xl,\Ybot);
    \coordinate (tpC2) at (\Xl+\Ltp,\Ybot);
    \draw[segR] (tpC1) -- (tpC2);
    \node[below=6pt, text=myred] at ($(tpC1)!0.5!(tpC2)$) {$t'$};

    \coordinate (rpC1) at (\Xrp,\Ybot);
    \coordinate (rpC2) at (\Xre,\Ybot);
    \draw[segB] (rpC1) -- (rpC2);
    \node[below=6pt, text=myblue] at ($(rpC1)!0.5!(rpC2)$) {$r'$};
\end{scope}

\draw[phaseArrow] ($(C)+(-1.7,-0.9)$) -- ($(D)+(11.7,-0.9)$)
node[midway, above=4pt] {\huge{\texttt{Phase 3}}};

\begin{scope}[shift={(D)}]
    \draw[line width=1.0pt] (\Xl-\rDist,\Ytop+\rDist) rectangle (\Xre+\rDist,\Ybot-\rDist);

    \coordinate (tD1) at (\Xl,\Ytop);
    \coordinate (tD2) at (\Xl+\Lt,\Ytop);
    \draw[segR] (tD1) -- (tD2);
    \node[above=6pt, text=myred] at ($(tD1)!0.5!(tD2)$) {$t$};

    \coordinate (rsD1) at (\Xrs,\Ytop);
    \coordinate (rsD2) at (\Xre,\Ytop);
    \draw[segB] (rsD1) -- (rsD2);
    \node[above=6pt, text=myblue] at ($(rsD1)!0.5!(rsD2)$) {$r-qr'$};

    \coordinate (lD1) at (\Xl,\Ybot);
    \coordinate (lD2) at (\Xre-\Ls-\Gap-\Lrp-1.0,\Ybot);
    \draw[segR] (lD1) -- (lD2);
    \node[below=6pt, text=myred] at ($(lD1)!0.5!(lD2)$) {$t'$(h)};

    \coordinate (rpD1) at (\Xre-\Ls-\Gap-\Lrp,\Ybot);     
    \coordinate (rpD2) at (\Xre-\Ls-\Gap,\Ybot);
    \draw[segB] (rpD1) -- (rpD2);
    \node[below=6pt, text=myblue] at ($(rpD1)!0.5!(rpD2)$) {$r'$};

    \coordinate (sD1) at (\Xre-\Ls,\Ybot);
    \coordinate (sD2) at (\Xre,\Ybot);
    \draw[segR] (sD1) -- (sD2);
    \node[below=6pt, text=myred] at ($(sD1)!0.5!(sD2)$) {$t'$(l)};
\end{scope}

\draw[phaseArrow] ($(D)+(5.0,-3.4)$) -- ($(E)+(5.0,1.6)$);
\node[anchor=west] at ($(D)+(5.2,-3.8)$) {\huge{\texttt{Phase 4}}};

\begin{scope}[shift={(E)}]
    \draw[line width=1.0pt] (\Xl-\rDist,\Ytop+\rDist) rectangle (\Xre+\rDist,\Ybot-\rDist);

    \coordinate (tE1) at (\Xl,\Ytop);
    \coordinate (tE2) at (\Xl+\Lt,\Ytop);
    \draw[segR] (tE1) -- (tE2);
    \node[above=6pt, text=myred] at ($(tE1)!0.5!(tE2)$) {$t$};

    \coordinate (rsE1) at (\Xrs,\Ytop);
    \coordinate (rsE2) at (\Xre,\Ytop);
    \draw[segB] (rsE1) -- (rsE2);
    \node[above=6pt, text=myblue] at ($(rsE1)!0.5!(rsE2)$) {$r-qr'$};

    \coordinate (bigE1) at (\Xl,\Ybot);
    \coordinate (bigE2) at (\Xrp-1.0,\Ybot);
    \draw[segR] (bigE1) -- (bigE2);
    \node[below=6pt, text=myred] at ($(bigE1)!0.5!(bigE2)$) {$t'+qt$};

    \coordinate (rpE1) at (\Xrp,\Ybot);
    \coordinate (rpE2) at (\Xre,\Ybot);
    \draw[segB] (rpE1) -- (rpE2);
    \node[below=6pt, text=myblue] at ($(rpE1)!0.5!(rpE2)$) {$r'$};
\end{scope}

\draw[phaseArrow] ($(E)+(11.7,-0.9)$) -- ($(F)+(-1.7,-0.9)$)
node[midway, above=4pt] {\huge{\texttt{Swap}}};

\begin{scope}[shift={(F)}]
    \draw[line width=1.0pt] (\Xl-\rDist,\Ytop+\rDist) rectangle (\Xre+\rDist,\Ybot-\rDist);

    \coordinate (bigF1) at (\Xl,\Ytop);
    \coordinate (bigF2) at (\Xrp-1.0,\Ytop);
    \draw[segR] (bigF1) -- (bigF2);
    \node[above=6pt, text=myred] at ($(bigF1)!0.5!(bigF2)$) {$t'+qt$};

    \coordinate (rpF1) at (\Xrp,\Ytop);
    \coordinate (rpF2) at (\Xre,\Ytop);
    \draw[segB] (rpF1) -- (rpF2);
    \node[above=6pt, text=myblue] at ($(rpF1)!0.5!(rpF2)$) {$r'$};

    \coordinate (tF1) at (\Xl,\Ybot);
    \coordinate (tF2) at (\Xl+\Lt,\Ybot);
    \draw[segR] (tF1) -- (tF2);
    \node[below=6pt, text=myred] at ($(tF1)!0.5!(tF2)$) {$t$};

    \coordinate (rsF1) at (\Xrs,\Ybot);
    \coordinate (rsF2) at (\Xre,\Ybot);
    \draw[segB] (rsF1) -- (rsF2);
    \node[below=6pt, text=myblue] at ($(rsF1)!0.5!(rsF2)$) {$r-qr'$};
\end{scope}
\end{tikzpicture}}

\caption{An illustration of how the two \texttt{Work} registers are allocated for temporary variables within a single iteration of the EEA. The upper and lower stripes represent the \texttt{Work1} and \texttt{Work2} registers, respectively. The symbols (h) and (l) indicate that the value $t'$ is split into its higher-order bits (h) and lower-order bits (l), which are placed at the corresponding positions on the two sides of the \texttt{Work2} register.}
\label{fig:regshare}
\end{figure}

We next describe the behavior of the system when performing comparison, addition, and subtraction operations on the register pairs $(r, r')$ and $(t, t')$. The arithmetic between $r$ and $r'$ is relatively straightforward: a left shift by $\ell_s$ positions corresponds to multiplying $r'$ by $2^{\ell_s}$. In contrast, the arithmetic involving $t$ and $t'$ requires a more careful interpretation. A convenient way to view this operation is that a left shift of $\ell_s$ positions can be regarded as dividing $t'$ by $2^{\ell_s}$, where the integer part is placed in the left-most portion of \texttt{Work2}, and the fractional part occupies the right-most portion. When $t$ is added to the integer part of $t'$ (which is properly aligned), the arithmetic relation can be expressed as $2^{-\ell_s}t' + t = 2^{-\ell_s}(t' + 2^{\ell_s}t)$.
Thus, the expression $t' + 2^{\ell_s}t$ represents the desired result, corresponding to a left shift by $\ell_s$ positions.
In Phase 4, the comparison between $t'$ and $2^{\ell_s}t$ acts on the left-most $(n + 3 - \ell_{r'} - \ell_s)$ qubits. Here $\ell_s$ denotes its current value before the subsequent right shift of \texttt{Work2}.

With this register allocation, one step of our space-efficient EEA proceeds as in Algorithm \ref{alg:regshare}. To illustrate how the proposed space-efficient EEA operates in practice, we also provide a concrete execution example for $p = 37$ and $x = 13$, shown in Table \ref{tab:run_example}.
\begin{algorithm}[htbp]
\small
\caption{One step of our EEA implementation with register sharing}
\begin{algorithmic}
\Require \texttt{Work1} register stores $\ket{t, q, r}$, \texttt{Work2} register stores $\ket{t', r'}$
\Require \texttt{Length} registers store $\ket{\ell_t, \ell_q, \ell_{r'}, \ell_s}$
\Require \texttt{Control} registers store $\ket{\texttt{Phase1}, \texttt{Phase2}, \texttt{Iter}, \texttt{Sign}}$.

\If{$(\texttt{Phase1}, \texttt{Phase2}) = (0, 0)$} \Comment{Arithmetic logic for the four phases}
    \State Perform a one-position left shift on \texttt{Work2}
    \State $\ell_s \gets \ell_s + 1$
    \State $\texttt{Sign} \gets \texttt{Sign} \oplus (r < 2^{\ell_s}r')$ \Comment{Act on qubits $(\ell_t + \ell_q + 2)$ through $(n + 3 - \ell_s)$}
\ElsIf{$(\texttt{Phase1}, \texttt{Phase2}) = (0, 1)$}
    \State Perform a one-position right shift on \texttt{Work2}
    \State $\ell_s \gets \ell_s - 1$, $\ell_q \gets \ell_q + 1$
    \State $(\texttt{Sign}, r) \gets (\texttt{Sign}, r) - 2^{\ell_s}r'$ \Comment{Use register \texttt{Sign} to store the sign of the subtraction result}
    \If{$\texttt{Sign} = 1$}
        \State $r \gets r + 2^{\ell_s}r'$ \Comment{Ignore overflows}
    \EndIf
    \State $\texttt{Sign} \gets \texttt{Sign} \oplus 1$
    \State Swap $\texttt{Sign}$ with the $(\ell_t + \ell_q + 1)$-th qubit of \texttt{Work1} \Comment{The least significant work qubit of $q$}
\ElsIf{$(\texttt{Phase1}, \texttt{Phase2}) = (1, 0)$}
    \State Swap $\texttt{Sign}$ with the $(\ell_t + \ell_q + 1)$-th qubit of \texttt{Work1}
    \If{$\texttt{Sign} = 0$}
        \State $t' \gets t' - 2^{\ell_s}t$ \Comment{Act on left-most $(\ell_t + 1)$ qubits}
    \EndIf
    \State $\texttt{Sign} \gets \texttt{Sign} \oplus 1$
    \State $(\texttt{Sign}, t') \gets (\texttt{Sign}, t') + 2^{\ell_s}t$ \Comment{Use register \texttt{Sign} to store the sign of the addition result}
    \State Perform a one-position left shift on \texttt{Work2}
    \State $\ell_s \gets \ell_s + 1$, $\ell_q \gets \ell_q - 1$
\ElsIf{$(\texttt{Phase1}, \texttt{Phase2}) = (1, 1)$}
    \State $\texttt{Sign} \gets \texttt{Sign} \oplus (t' \ge 2^{\ell_s}t)$ \Comment{Act on left-most $(n + 3 - \ell_{r'} - \ell_s)$ qubits}
    \State Perform a one-position right shift on \texttt{Work2}
    \State $\ell_s \gets \ell_s - 1$
\EndIf

\If{$\ell_q = 0$ \textbf{and} $\ell_{r'} > 0$} \Comment{Phase update logic; $\ell_{r'} = 0$ indicates algorithm termination}
    \State $\texttt{Phase2} \gets \texttt{Phase2} \oplus \texttt{Sign} \oplus \texttt{Phase1}$, $\texttt{Sign} \gets \texttt{Sign} \oplus \texttt{Phase2}$
\EndIf
\If{$\ell_s = 0$}
    \State $\texttt{Phase1} \gets \texttt{Phase1} \oplus 1$, $\texttt{Phase2} \gets \texttt{Phase2} \oplus 1$
\EndIf

\If{$\ell_q = 0$ \textbf{and} $\ell_s = 0$} \Comment{Register swapping at the end of one EEA iteration}
    \State Swap \texttt{Work1} and \texttt{Work2}
    \State Update $\ell_t$ to the bit length of the new $t$ \Comment{Act on the left-most $(n + 3 - \ell_{r'})$ qubits of \texttt{Work1} and \texttt{Work2}}
    \State Update $\ell_{r'}$ to the bit length of the new $r'$ \Comment{Act on the right-most $(n + 2 - \ell_t)$ qubits of \texttt{Work1} and \texttt{Work2}}
    \State $\texttt{Iter} \gets \texttt{Iter} \oplus 1$
\EndIf
\end{algorithmic}
\label{alg:regshare}
\end{algorithm}

\begin{table}[htbp]
\centering
\footnotesize
\begin{tabular}{c c c c c c c c c c c c c c c c c}
\hline
\textbf{Step} & \texttt{Work1} & \texttt{Work2} & $t$ & $q$ & $r$ & $t'$ & $r'$ & $\ell_t$ & $\ell_q$ & $\ell_{r'}$ & $\ell_s$ & \texttt{Phase1} & \texttt{Phase2} & \texttt{Iter} & \texttt{Sign} \\
\hline
0 & \verb+10||0100101+ & \verb+00000|1101|+ & 1 & 0 & 37 & 0 & 13 & 1 & 0 & 4 & 0 & 0 & 0 & 0 & 0 \\
\hline
1 & \verb+10||0100101+ & \verb+0000|1101|0+ & 1 & 0 & 37 & 0 & 13 & 1 & 0 & 4 & 1 & 0 & 0 & 0 & 0 \\
\hline
2 & \verb+10||0100101+ & \verb+000|1101|00+ & 1 & 0 & 37 & 0 & 13 & 1 & 0 & 4 & 2 & 0 & 1 & 0 & 0 \\
\hline
3 & \verb+10|1|001011+ & \verb+0000|1101|0+ & 1 & 2 & 11 & 0 & 13 & 1 & 1 & 4 & 1 & 0 & 1 & 0 & 0 \\
\hline
4 & \verb+10|10|01011+ & \verb+00000|1101|+ & 1 & 2 & 11 & 0 & 13 & 1 & 2 & 4 & 0 & 1 & 0 & 0 & 0 \\
\hline
5 & \verb+10|1|001011+ & \verb+0000|1101|0+ & 1 & 2 & 11 & 0 & 13 & 1 & 1 & 4 & 1 & 1 & 0 & 0 & 0 \\
\hline
6 & \verb+10||0001011+ & \verb+000|1101|01+ & 1 & 0 & 11 & 2 & 13 & 1 & 0 & 4 & 2 & 1 & 1 & 0 & 1 \\
\hline
7 & \verb+10||0001011+ & \verb+1000|1101|0+ & 1 & 0 & 11 & 2 & 13 & 1 & 0 & 4 & 1 & 1 & 1 & 0 & 0 \\
\hline
8 & \verb+010||001101+ & \verb+10000|1011|+ & 2 & 0 & 13 & 1 & 11 & 2 & 0 & 4 & 0 & 0 & 0 & 1 & 0 \\
\hline
9 & \verb+010||001101+ & \verb+0000|1011|1+ & 2 & 0 & 13 & 1 & 11 & 2 & 0 & 4 & 1 & 0 & 1 & 1 & 0 \\
\hline
10 & \verb+010|1|00010+ & \verb+10000|1011|+ & 2 & 1 & 2 & 1 & 11 & 2 & 1 & 4 & 0 & 1 & 0 & 1 & 0 \\
\hline
11 & \verb+010||000010+ & \verb+1000|1011|1+ & 2 & 0 & 2 & 3 & 11 & 2 & 0 & 4 & 1 & 1 & 1 & 1 & 1 \\
\hline
12 & \verb+110||001011+ & \verb+0100000|10|+ & 3 & 0 & 11 & 2 & 2 & 2 & 0 & 2 & 0 & 0 & 0 & 0 & 0 \\
\hline
13 & \verb+110||001011+ & \verb+100000|10|0+ & 3 & 0 & 11 & 2 & 2 & 2 & 0 & 2 & 1 & 0 & 0 & 0 & 0 \\
\hline
14 & \verb+110||001011+ & \verb+00000|10|01+ & 3 & 0 & 11 & 2 & 2 & 2 & 0 & 2 & 2 & 0 & 0 & 0 & 0 \\
\hline
15 & \verb+110||001011+ & \verb+0000|10|010+ & 3 & 0 & 11 & 2 & 2 & 2 & 0 & 2 & 3 & 0 & 1 & 0 & 0 \\
\hline
16 & \verb+110|1|00011+ & \verb+00000|10|01+ & 3 & 4 & 3 & 2 & 2 & 2 & 1 & 2 & 2 & 0 & 1 & 0 & 0 \\
\hline
17 & \verb+110|10|0011+ & \verb+100000|10|0+ & 3 & 4 & 3 & 2 & 2 & 2 & 2 & 2 & 1 & 0 & 1 & 0 & 0 \\
\hline
18 & \verb+110|101|001+ & \verb+0100000|10|+ & 3 & 5 & 1 & 2 & 2 & 2 & 3 & 2 & 0 & 1 & 0 & 0 & 0 \\
\hline
19 & \verb+110|10|0001+ & \verb+010000|10|1+ & 3 & 4 & 1 & 5 & 2 & 2 & 2 & 2 & 1 & 1 & 0 & 0 & 0 \\
\hline
20 & \verb+110|1|00001+ & \verb+10000|10|10+ & 3 & 4 & 1 & 5 & 2 & 2 & 1 & 2 & 2 & 1 & 0 & 0 & 0 \\
\hline
21 & \verb+110||000001+ & \verb+0100|10|100+ & 3 & 0 & 1 & 17 & 2 & 2 & 0 & 2 & 3 & 1 & 1 & 0 & 1 \\
\hline
22 & \verb+110||000001+ & \verb+00100|10|10+ & 3 & 0 & 1 & 17 & 2 & 2 & 0 & 2 & 2 & 1 & 1 & 0 & 0 \\
\hline
23 & \verb+110||000001+ & \verb+000100|10|1+ & 3 & 0 & 1 & 17 & 2 & 2 & 0 & 2 & 1 & 1 & 1 & 0 & 0 \\
\hline
24 & \verb+100010||010+ & \verb+11000000|1|+ & 17 & 0 & 2 & 3 & 1 & 5 & 0 & 1 & 0 & 0 & 0 & 1 & 0 \\
\hline
25 & \verb+100010||010+ & \verb+1000000|1|1+ & 17 & 0 & 2 & 3 & 1 & 5 & 0 & 1 & 1 & 0 & 0 & 1 & 0 \\
\hline
26 & \verb+100010||010+ & \verb+000000|1|11+ & 17 & 0 & 2 & 3 & 1 & 5 & 0 & 1 & 2 & 0 & 1 & 1 & 0 \\
\hline
27 & \verb+100010|1|00+ & \verb+1000000|1|1+ & 17 & 2 & 0 & 3 & 1 & 5 & 1 & 1 & 1 & 0 & 1 & 1 & 0 \\
\hline
28 & \verb+100010|10|0+ & \verb+11000000|1|+ & 17 & 2 & 0 & 3 & 1 & 5 & 2 & 1 & 0 & 1 & 0 & 1 & 0 \\
\hline
29 & \verb+100010|1|00+ & \verb+1000000|1|1+ & 17 & 2 & 0 & 3 & 1 & 5 & 1 & 1 & 1 & 1 & 0 & 1 & 0 \\
\hline
30 & \verb+100010||000+ & \verb+100100|1|10+ & 17 & 0 & 0 & 37 & 1 & 5 & 0 & 1 & 2 & 1 & 1 & 1 & 1 \\
\hline
31 & \verb+100010||000+ & \verb+0100100|1|1+ & 17 & 0 & 0 & 37 & 1 & 5 & 0 & 1 & 1 & 1 & 1 & 1 & 0 \\
\hline
32 & \verb+Terminated+ & \verb+Terminated+ & 37 & 0 & 1 & 17 & 0 & 6 & 0 & 0 & 0 & 0 & 0 & 0 & 0 \\
\hline
33 & \verb+Terminated+ & \verb+Terminated+ & 37 & 0 & 1 & 17 & 0 & 6 & 0 & 0 & 1 & 0 & 0 & 0 & 0 \\
\hline
34 & \verb+Terminated+ & \verb+Terminated+ & 37 & 0 & 1 & 17 & 0 & 6 & 0 & 0 & 2 & 0 & 0 & 0 & 0 \\
\hline
35 & \verb+Terminated+ & \verb+Terminated+ & 37 & 0 & 1 & 17 & 0 & 6 & 0 & 0 & 3 & 0 & 0 & 0 & 0 \\
\hline
36 & \verb+Terminated+ & \verb+Terminated+ & 37 & 0 & 1 & 17 & 0 & 6 & 0 & 0 & 4 & 0 & 0 & 0 & 0 \\
\hline
\end{tabular}
\caption{An execution example of our space-efficient EEA for $p = 37$ and $x = 13$.
The table mainly illustrates how the two $(n+3)$-qubit registers \texttt{Work1} and \texttt{Work2} are allocated during the execution to store the variables $t, q, r$ and $t', r'$, respectively.}
\label{tab:run_example}
\end{table}

For each step, Table \ref{tab:run_example} shows how the contents of the two \texttt{Work} registers evolve and how different variables are mapped onto \texttt{Work1} and \texttt{Work2} under the register-sharing strategy.
In particular, the variable $t'$ stored in \texttt{Work2} is divided into two parts that occupy the left-most and right-most portions of the register. The division symbols are not explicitly implemented in the quantum circuit; they are all illustrative and inferred from the corresponding \texttt{Length} registers, used to visualize the register allocation.
In the actual quantum implementation, all registers are in superposition, and the location of these division symbols may vary across different input values.

\subsection{Algorithmic optimizations}
\label{subsec:opt}

The baseline implementation of one step, as described in the previous subsection, was constructed to be fully reversible, ensuring that every state transformation can be realized as a unitary operation on a quantum computer. However, this direct implementation contains multiple redundant operations that unnecessarily increase both the circuit depth and gate count. In this subsection, we present a series of algorithmic optimizations applied to the baseline design in order to reduce the quantum circuit depth and overall gate complexity.

\paragraph{Removing redundant arithmetic and shift operations.}
In the baseline version, several arithmetic and shift operations between different phases are functionally redundant. Specifically, we have:
\[
\begin{split}
    \text{In phases 1 and 2:} \quad & \text{comparison, addition, subtraction between } (r, 2^{\ell_s}r'); \\
    \text{In phases 3 and 4:} \quad & \text{comparison, addition, subtraction between } (t', 2^{\ell_s}t); \\
    \text{In phases 1 and 3:} \quad & \text{a one-position left shift on } \texttt{Work2}; \\
    \text{In phases 2 and 4:} \quad & \text{a one-position right shift on } \texttt{Work2}. \\
\end{split}
\]
By combining the control conditions of these operations (or equivalently, the corresponding quantum control wires), we can eliminate redundant operations from the circuit. This optimization results in a total of two groups of location-controlled addition/subtraction operations and two location-controlled swaps in one step. Here, \emph{location-control} refers to arithmetic operations applied conditionally on specific positions of the two \texttt{Work} registers, with the active positions indicated by the \texttt{Length} registers. For the merged arithmetic on $t$'s, the active boundary remains phase dependent: it is $\ell_t+1$ in Phase 3 and $n+3-\ell_{r'}-\ell_s$ in Phase 4. We prepare these two boundaries in the $\ell_t$ and $\ell_{r'}$ registers, respectively, and conditionally swap the registers under the control of \texttt{Phase2}, so that the arithmetic operations always use $\ell_t$ as their boundary. The original length values are restored before the post-shift operations.

\paragraph{Merging location-controlled swaps.}
To further decrease the circuit depth and gate count, we observe that the two location-controlled swaps appearing in phases 2 and 3 act on the same qubit position of \texttt{Work1}. We merge them into a single swap controlled by $\texttt{Phase1}\oplus\texttt{Phase2}$. The required update of $\ell_q$ depends on the phase: $\ell_q$ is increased in phase 2 and decreased in phase 3.

\paragraph{Algorithm pseudocode and overall complexity.}
After applying the above optimizations, the overall asymptotic quantum resource estimation per step is summarized as follows:
\begin{itemize}
    \item Five location-controlled arithmetic operations, including one swap, two additions, and two subtractions; along with length updates at the termination step of an EEA iteration. Each of these operations costs $O(n)$ gates in the implementation presented in Section \ref{sec:circuit}.
    \item Four position shifts and one $(n + 3)$-qubit SWAP, each requiring $O(n)$ quantum gates.
\end{itemize}

With these optimizations, one step of our space-efficient EEA proceeds as in Algorithm~\ref{alg:opt}.
\begin{algorithm}[htbp]
\small
\caption{Optimized one step of our EEA implementation}
\begin{algorithmic}
\Require \texttt{Work1} register stores $\ket{t, q, r}$, \texttt{Work2} register stores $\ket{t', r'}$;
\Require \texttt{Length} registers store $\ket{\ell_t, \ell_q, \ell_{r'}, \ell_s}$;
\Require \texttt{Control} registers store $\ket{\texttt{Phase1}, \texttt{Phase2}, \texttt{Iter}, \texttt{Sign}}$.

\If{$\texttt{Phase1} = 0$} \Comment{Pre-shift operations}
    \State Perform a one-position left shift on \texttt{Work2}, and update $\ell_s \gets \ell_s + 1$
    \State \textbf{if} $\texttt{Phase2} = 1$ \textbf{then} perform a two-position right shift on \texttt{Work2}, and update $\ell_s \gets \ell_s - 2$ \textbf{end if}
    
    \State $(\texttt{Sign}, r) \gets (\texttt{Sign}, r) - 2^{\ell_s}r'$ \Comment{Arithmetic block 1: location-controlled subtraction on $r$'s}
    \State \textbf{if} $\texttt{Phase2} = 1$ \textbf{then} $\texttt{Sign} \gets \texttt{Sign} \oplus 1$ \textbf{end if}
    \State \textbf{if} $\texttt{Phase2} = 0$ \textbf{or} $\texttt{Sign} = 0$ \textbf{then} $r \gets r + 2^{\ell_s}r'$ \textbf{end if}
\EndIf

\If{$\texttt{Phase1} \oplus \texttt{Phase2} = 1$} \Comment{Arithmetic block 2: location-controlled swap}
    \State Swap $\texttt{Sign}$ with the $(\ell_t + \ell_q + 1)$-th qubit of \texttt{Work1}
    \If{$\texttt{Phase1} = 1$}
        \State $\ell_q \gets \ell_q - 1$
    \Else
        \State $\ell_q \gets \ell_q + 1$
    \EndIf
\EndIf

\If{$\texttt{Phase1} = 1$} \Comment{Arithmetic block 3 on $t$'s}
    \State $\ell_t \gets \ell_t + 1$, $\ell_{r'} \gets n + 3 - \ell_{r'} - \ell_s$
    \State \textbf{if} $\texttt{Phase2} = 1$ \textbf{then} swap $\ell_t$ and $\ell_{r'}$ \textbf{end if}

    \State \textbf{if} $\texttt{Phase2} = 1$ \textbf{or} $\texttt{Sign} = 0$ \textbf{then} $t' \gets t' - 2^{\ell_s}t$ \textbf{end if} \Comment{Act on the left-most $\ell_t$ qubits}
    \State $\texttt{Sign} \gets \texttt{Sign} \oplus 1$
    \State $(\texttt{Sign}, t') \gets (\texttt{Sign}, t') + 2^{\ell_s}t$ \Comment{Use the same boundary}

    \State \textbf{if} $\texttt{Phase2} = 1$ \textbf{then} swap $\ell_t$ and $\ell_{r'}$ \textbf{end if}
    \State $\ell_{r'} \gets n + 3 - \ell_{r'} - \ell_s$, $\ell_t \gets \ell_t - 1$
    
    \State Perform a one-position left shift on \texttt{Work2}, and update $\ell_s \gets \ell_s + 1$ \Comment{Post-shift operations}
    \State \textbf{if} $\texttt{Phase2} = 1$ \textbf{then} perform a two-position right shift on \texttt{Work2}, and update $\ell_s \gets \ell_s - 2$ \textbf{end if}
\EndIf

\If{$\ell_q = 0$ \textbf{and} $\ell_{r'} > 0$} \Comment{Phase update logic; $\ell_{r'} = 0$ indicates algorithm termination}
    \State $\texttt{Phase2} \gets \texttt{Phase2} \oplus \texttt{Sign} \oplus \texttt{Phase1}$, $\texttt{Sign} \gets \texttt{Sign} \oplus \texttt{Phase2}$
\EndIf
\State \textbf{if} $\ell_s = 0$ \textbf{then} $\texttt{Phase1} \gets \texttt{Phase1} \oplus 1$, $\texttt{Phase2} \gets \texttt{Phase2} \oplus 1$ \textbf{end if}

\If{$\ell_q = 0$ \textbf{and} $\ell_s = 0$} \Comment{Register swapping at the end of one EEA iteration}
    \State Swap \texttt{Work1} and \texttt{Work2}
    \State Update $\ell_t$ to the bit length of the new $t$, $\ell_{r'}$ to the bit length of the new $r'$
    \State $\texttt{Iter} \gets \texttt{Iter} \oplus 1$
\EndIf
\end{algorithmic}
\label{alg:opt}
\end{algorithm}

\section{Quantum Circuit Implementation of Modular Inversion}
\label{sec:circuit}

In this section, we present the explicit quantum circuit implementation of a single step of the space-efficient EEA with register sharing, as described in Algorithm \ref{alg:opt}.
Throughout the section, we use the same quantum register notation as in the algorithm.

The full circuit construction is, however, too complex to be depicted within a small number of circuit diagrams.
This difficulty is mainly due to the presence of location-controlled arithmetic operations, which are highly non-standard.
To present the construction in a relatively more explicit manner, we first describe the overall circuit structure in Section \ref{subsec:all_circuit}, where the location-controlled arithmetic operations are treated as circuit blocks; we then provide a detailed implementation of these operations in Section \ref{subsec:blocks}. In all circuit diagrams, the black triangle denotes the output wire that carries the result of the corresponding circuit block; and we denote $\ell = \lfloor \log_2 n \rfloor$ for convenience.

\subsection{Overall circuit implementation of one step}
\label{subsec:all_circuit}

According to Algorithm \ref{alg:opt}, the overall circuit implementation of a single step of our space-efficient EEA is shown in Figures \ref{fig:all_circuit_1} and \ref{fig:all_circuit_2}, presented in a continuous layout.
Both circuit diagrams are divided into several algorithmic components, indicated by dashed boxes.
Each component corresponds exactly to one arithmetic block specified in Algorithm \ref{alg:opt}.

In both circuit diagrams, $\mathsf{Shift}$ denotes a position shift of the workspace: a positive value represents a left shift, while a negative value represents a right shift.
The circuit blocks $\mathsf{Add}$, $\mathsf{Sub}$, and $\mathsf{Swap}$ denote location-controlled addition, subtraction, and swap operations, respectively.
At the end of each EEA iteration, the circuit block labeled $\mathsf{SWAP}$ represents a full SWAP operation applied to all qubits in the registers \texttt{Work1} and \texttt{Work2}; this block is distinct from the $\mathsf{Swap}$ block to avoid ambiguity.
Finally, the $\mathsf{Len}$ block denotes the update operations on the length registers.

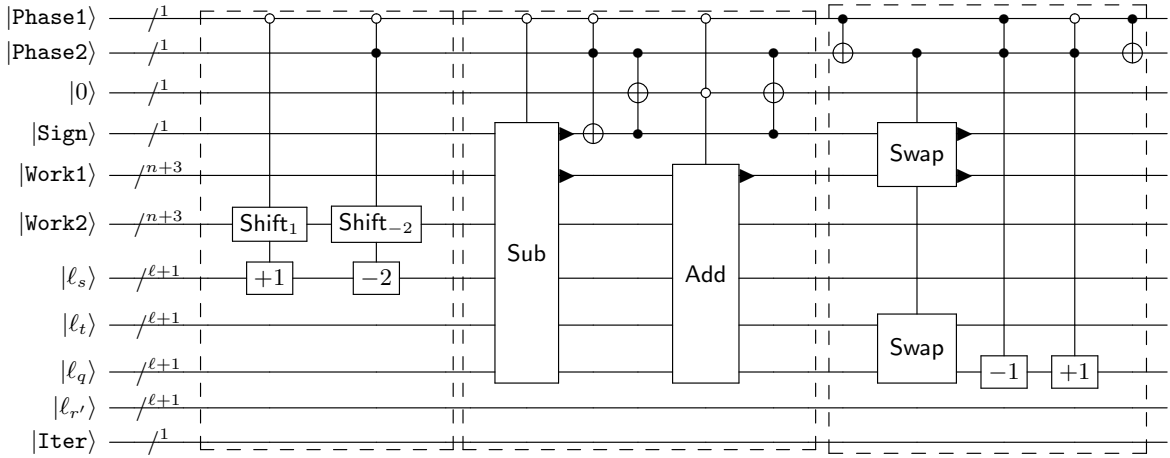
\begin{figure}[ht]
\footnotesize
\centering
\scalebox{\iffull 1 \else 0.8 \fi}{
\(\Qcircuit @C=1.0em @R=0.8em {
    & \lstick{\ket{\texttt{Phase1}}} & \qw & \qw\slash^1 & \qw & \qw 
    & \ctrlo{5} & \ctrlo{1} & \qw & \qw & \ctrlo{3} 
    & \ctrlo{1} & \qw & \ctrlo{2} & \qw & \qw 
    & \ctrl{1} & \qw & \ctrl{1} & \ctrlo{1} & \ctrl{1} 
    & \qw \\
    & \lstick{\ket{\texttt{Phase2}}} & \qw & \qw\slash^1 & \qw & \qw 
    & \qw & \ctrl{4} & \qw & \qw & \qw 
    & \ctrl{2} & \ctrl{1} & \qw & \ctrl{1} & \qw 
    & \targ & \ctrl{2} & \ctrl{7} & \ctrl{7} & \targ 
    & \qw \\
    & \lstick{\ket{0}} & \qw & \qw\slash^1 & \qw & \qw 
    & \qw & \qw & \qw & \qw & \qw 
    & \qw & \targ & \ctrlo{2} & \targ & \qw 
    & \qw & \qw & \qw & \qw & \qw 
    & \qw \\
    & \lstick{\ket{\texttt{Sign}}} & \qw & \qw\slash^1 & \qw & \qw 
    & \qw & \qw & \qw & \qw & \multigateouttwohi{5}{\mathsf{Sub}} 
    & \targ & \ctrl{-1} & \qw & \ctrl{-1} & \qw 
    & \qw & \multigateouttwo{1}{\mathsf{Swap}} & \qw & \qw & \qw 
    & \qw \\
    & \lstick{\ket{\texttt{Work1}}} & \qw & \qw\slash^{n+3} & \qw & \qw 
    & \qw & \qw & \qw & \qw & \ghost{\mathsf{Sub}} 
    & \qw & \qw & \multigateout{4}{\mathsf{Add}} & \qw & \qw 
    & \qw & \ghost{\mathsf{Swap}}\qwx[3] & \qw & \qw & \qw 
    & \qw \\
    & \lstick{\ket{\texttt{Work2}}} & \qw & \qw\slash^{n+3} & \qw & \qw 
    & \gate{\mathsf{Shift}_1}\qwx[1] & \gate{\mathsf{Shift}_{-2}}\qwx[1] & \qw & \qw & \ghost{\mathsf{Sub}} 
    & \qw & \qw & \ghost{\mathsf{Add}} & \qw & \qw
    & \qw & \qw & \qw & \qw & \qw 
    & \qw \\
    & \lstick{\ket{\ell_s}} & \qw & \qw\slash^{\ell + 1} & \qw & \qw 
    & \gate{+1} & \gate{-2} & \qw & \qw & \ghost{\mathsf{Sub}} 
    & \qw & \qw & \ghost{\mathsf{Add}} & \qw & \qw
    & \qw & \qw & \qw & \qw & \qw 
    & \qw \\
    & \lstick{\ket{\ell_t}} & \qw & \qw\slash^{\ell + 1} & \qw & \qw 
    & \qw & \qw & \qw & \qw & \ghost{\mathsf{Sub}} 
    & \qw & \qw & \ghost{\mathsf{Add}} & \qw & \qw 
    & \qw & \multigate{1}{\mathsf{Swap}} & \qw & \qw & \qw 
    & \qw \\
    & \lstick{\ket{\ell_q}} & \qw & \qw\slash^{\ell + 1} & \qw & \qw 
    & \qw & \qw & \qw & \qw & \ghost{\mathsf{Sub}} 
    & \qw & \qw & \ghost{\mathsf{Add}} & \qw & \qw 
    & \qw & \ghost{\mathsf{Swap}} & \gate{-1} & \gate{+1} & \qw 
    & \qw \\
    & \lstick{\ket{\ell_{r'}}} & \qw & \qw\slash^{\ell + 1} & \qw & \qw 
    & \qw & \qw & \qw & \qw & \qw 
    & \qw & \qw & \qw & \qw & \qw 
    & \qw & \qw & \qw & \qw & \qw 
    & \qw \\
    \push{\rule{0em}{1.2em}} & \lstick{\ket{\texttt{Iter}}} & \qw & \qw\slash^{1} & \qw & \qw 
    & \qw & \qw & \qw & \qw & \qw 
    & \qw & \qw & \qw & \qw & \qw 
    & \qw & \qw & \qw & \qw & \qw 
    & \qw
    \gategroup{1}{6}{11}{9}{0.6em}{--}
    \gategroup{1}{10}{11}{16}{0.6em}{--}
    \gategroup{1}{17}{11}{21}{0.9em}{--}
}\)}
\caption{Overall circuit implementation (first half) of a single step of our space-efficient EEA.
The three dashed boxes, from left to right, represent:
(1) pre-shift operations;
(2) location-controlled subtraction on $r$'s;
(3) location-controlled swap.}
\label{fig:all_circuit_1}
\end{figure}

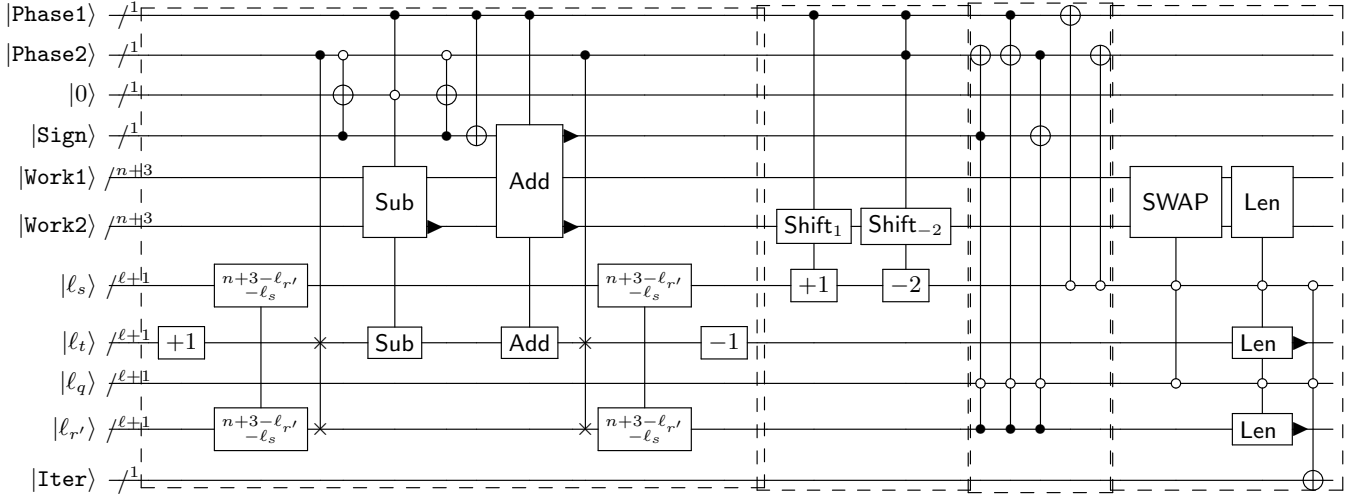
\begin{figure}[ht]
\footnotesize
\centering
\scalebox{\iffull 1 \else 0.8 \fi}{
\(\Qcircuit @C=0.4em @R=0.8em {
    & \lstick{\ket{\texttt{Phase1}}} & \qw & \qw\slash^1 & \qw & \qw
    & \qw & \qw & \qw
    & \qw & \ctrl{2} & \qw & \ctrl{3} & \ctrl{3}
    & \qw & \qw & \qw & \qw
    & \qw & \qw
    & \ctrl{5} & \ctrl{1} & \qw
    & \qw & \ctrl{1} & \qw & \targ & \qw
    & \qw & \qw & \qw & \qw & \qw \\
    & \lstick{\ket{\texttt{Phase2}}} & \qw & \qw\slash^1 & \qw & \qw
    & \qw & \qw & \ctrl{6}
    & \ctrlo{1} & \qw & \ctrlo{1} & \qw & \qw
    & \qw & \ctrl{6} & \qw & \qw
    & \qw & \qw
    & \qw & \ctrl{4} & \qw
    & \targ & \targ & \ctrl{2} & \qw & \targ
    & \qw & \qw & \qw & \qw & \qw \\
    & \lstick{\ket{0}} & \qw & \qw\slash^1 & \qw & \qw
    & \qw & \qw & \qw
    & \targ & \ctrlo{2} & \targ & \qw & \qw
    & \qw & \qw & \qw & \qw
    & \qw & \qw
    & \qw & \qw & \qw
    & \qw & \qw & \qw & \qw & \qw
    & \qw & \qw & \qw & \qw & \qw \\
    & \lstick{\ket{\texttt{Sign}}} & \qw & \qw\slash^1 & \qw & \qw
    & \qw & \qw & \qw
    & \ctrl{-1} & \qw & \ctrl{-1} & \targ & \multigateouttwo{2}{\mathsf{Add}}
    & \qw & \qw & \qw & \qw
    & \qw & \qw
    & \qw & \qw & \qw
    & \ctrl{-2} & \qw & \targ & \qw & \qw
    & \qw & \qw & \qw & \qw & \qw \\
    & \lstick{\ket{\texttt{Work1}}} & \qw & \qw\slash^{n+3} & \qw & \qw
    & \qw & \qw & \qw
    & \qw & \ghost{\mathsf{Sub}} & \qw & \qw & \ghost{\mathsf{Add}}
    & \qw & \qw & \qw & \qw
    & \qw & \qw
    & \qw & \qw & \qw
    & \qw & \qw & \qw & \qw & \qw
    & \qw & \multigate{1}{\mathsf{SWAP}} & \multigate{1}{\mathsf{Len}} & \qw & \qw \\
    & \lstick{\ket{\texttt{Work2}}} & \qw & \qw\slash^{n+3} & \qw & \qw
    & \qw & \qw & \qw
    & \qw & \multigateout{-1}{\mathsf{Sub}}\qwx[2] & \qw & \qw & \ghost{\mathsf{Add}}\qwx[2]
    & \qw & \qw & \qw & \qw
    & \qw & \qw
    & \gate{\mathsf{Shift}_1}\qwx[1] & \gate{\mathsf{Shift}_{-2}}\qwx[1] & \qw
    & \qw & \qw & \qw & \qw & \qw
    & \qw & \ghost{\mathsf{SWAP}} & \ghost{\mathsf{Len}} & \qw & \qw \\
    & \lstick{\ket{\ell_s}} & \qw & \qw\slash^{\ell + 1} & \qw & \qw
    & \qw & \gate{\substack{n+3-\ell_{r'}\\[-0.2em]-\ell_s}}\qwx[3] & \qw
    & \qw & \qw & \qw & \qw & \qw
    & \qw & \qw & \gate{\substack{n+3-\ell_{r'}\\[-0.2em]-\ell_s}}\qwx[3] & \qw
    & \qw & \qw
    & \gate{+1} & \gate{-2} & \qw
    & \qw & \qw & \qw & \ctrlo{-6} & \ctrlo{-5}
    & \qw & \ctrlo{-1} & \ctrlo{-1} & \ctrlo{2} & \qw \\
    & \lstick{\ket{\ell_t}} & \qw & \qw\slash^{\ell + 1} & \qw & \qw
    & \gate{+1} & \qw & \qswap\qwx[2]
    & \qw & \gate{\mathsf{Sub}} & \qw & \qw & \gate{\mathsf{Add}}
    & \qw & \qswap\qwx[2] & \qw & \gate{-1}
    & \qw & \qw
    & \qw & \qw & \qw
    & \qw & \qw & \qw & \qw & \qw
    & \qw & \qw & \gateout{\mathsf{Len}}\qwx[-1] & \qw & \qw \\
    & \lstick{\ket{\ell_q}} & \qw & \qw\slash^{\ell + 1} & \qw & \qw
    & \qw & \qw & \qw
    & \qw & \qw & \qw & \qw & \qw
    & \qw & \qw & \qw & \qw
    & \qw & \qw
    & \qw & \qw & \qw
    & \ctrlo{-5} & \ctrlo{-7} & \ctrlo{-5} & \qw & \qw
    & \qw & \ctrlo{-2} & \ctrlo{-1} & \ctrlo{2} & \qw \\
    & \lstick{\ket{\ell_{r'}}} & \qw & \qw\slash^{\ell + 1} & \qw & \qw
    & \qw & \gate{\substack{n+3-\ell_{r'}\\[-0.2em]-\ell_s}} & \qswap
    & \qw & \qw & \qw & \qw & \qw
    & \qw & \qswap & \gate{\substack{n+3-\ell_{r'}\\[-0.2em]-\ell_s}} & \qw
    & \qw & \qw
    & \qw & \qw & \qw
    & \ctrl{-1} & \ctrl{-1} & \ctrl{-1} & \qw & \qw
    & \qw & \qw & \gateout{\mathsf{Len}}\qwx[-1] & \qw & \qw \\
    & \lstick{\ket{\texttt{Iter}}} & \qw & \qw\slash^1 & \qw & \qw
    & \qw & \qw & \qw
    & \qw & \qw & \qw & \qw & \qw
    & \qw & \qw & \qw & \qw
    & \qw & \qw
    & \qw & \qw & \qw
    & \qw & \qw & \qw & \qw & \qw
    & \qw & \qw & \qw & \targ & \qw
    \gategroup{1}{6}{11}{19}{0.6em}{--}
    \gategroup{1}{20}{11}{23}{0.8em}{--}
    \gategroup{1}{24}{11}{28}{1.0em}{--}
    \gategroup{1}{29}{11}{33}{0.8em}{--}
}\)}
\caption{Overall circuit implementation (second half) of a single step of our space-efficient EEA, continued from Figure~\ref{fig:all_circuit_1}.
The four dashed boxes, from left to right, represent:
(1) location-controlled addition on $t$'s;
(2) post-shift operations;
(3) phase update;
(4) swapping \texttt{Work1} and \texttt{Work2}, together with the corresponding length updates at the end of one EEA iteration.}
\label{fig:all_circuit_2}
\end{figure}

\subsection{Step-dependent active windows}
\label{subsec:actwindow}

To reduce the quantum gate complexity for implementing our space-efficient EEA, we introduce step-dependent active windows for each location-controlled operation. As illustrated in the figures of the previous subsection, we define $k$ and $K$ as known bounds on the active indices of the \texttt{Work} registers.

For instance, at step $T$, a location-controlled Swap circuit acts on the $(\ell_t + \ell_q + 1)$-th qubit of the \texttt{Work1} register. Over all possible inputs $x$ used to compute the modular inverse $x^{-1} \bmod p$, the quantity $\ell_t + \ell_q + 1$ is bounded below by $k := k(T)$ and above by $K := K(T)$. Therefore, applying location-controlled operations to qubits in the \texttt{Work1} register with indices outside the interval $[k, K]$ is unnecessary. We refer to this interval as the \emph{active window}.

We next specify the step-dependent active windows for each type of location-controlled operation. All bounds are taken over all \emph{possible inputs} $x$, namely, those inputs $x$ for which the location-controlled circuits are activated by the outer control qubits. 
We defer the proofs to Appendix \ref{app:actwindow}. Here the constants are set to be
\[
\lambda := (2 + \sqrt3)^{1/3}, 
\quad c := 1/\log_2\lambda = 3/\log_2(2 + \sqrt3),
\quad \delta:=\log_\lambda\left(\frac{4\sqrt{3}-3}{13}\lambda^2\right)
=0.726208\ldots .
\]

\begin{itemize}

\item \textbf{Location-controlled addition/subtraction on $r$'s.}  
These operations act on qubits indexed from $(\ell_t + \ell_q + 2)$ to $(n + 3 - \ell_s)$. 
The upper bound $K_1 = K_1(T) = n + 3$; the lower bound $k_1 = k_1(T)$ is defined by
\[
    \ell_t + \ell_q + 2 \ge k_1(T)
    := \max\left\{\ceil{\dfrac{T-n-1-4\delta}{4c - 1}}, 1\right\} + 2.
\]

\item \textbf{Location-controlled swap.} 
These operations act on the $(\ell_t + \ell_q + 1)$-th qubit. 
The two-sided bounds are defined by
\[
\begin{split}
    \ell_t + \ell_q + 1 
    & \ge k_2(T) := \max\left\{\ceil{\dfrac{T - 3(n + 1) - 4\delta}{4c - 3}}, 1\right\} + 1, \\
    \ell_t + \ell_q + 1 
    & \le K_2(T) := \min\left\{\floor{T/2} + 2, n + 2\right\}.
\end{split}
\]

\item \textbf{Location-controlled addition/subtraction on $t$'s.}  
These are one-sided controlled operations. So the left bound $k_3(T) = 1.$ The right boundary is $\ell_t+1$ in Phase~3 and $n+3-\ell_{r'}-\ell_s$ in Phase~4. Therefore, their common right boundary satisfies
\[
\max\left\{\ell_t+1,\,n+3-\ell_{r'}-\ell_s\right\}
\le K_3(T):=\min\left\{\floor{T/4}+2,\,n+1\right\}.
\]

\item \textbf{Length-update circuit for updating $\ell_t$.}  
We execute the length-update circuit only at step indices $T$ that are multiples of four.
This circuit acts on qubits indexed from $\ell_t$ to $(n + 3 - \ell_{r'})$, and the corresponding two-sided active window is defined by
\[
\begin{split}
    \ell_t 
    & \ge k_4(T) := \max\left\{\ceil{\frac{T - 4(n + 1) - 4\delta}{4c - 4}}, 1\right\}, \\
    n + 3 - \ell_{r'} 
    & \le K_4(T) := \min\left\{T/4 + 3, n + 2\right\}.
\end{split}
\]

\item \textbf{Length-update circuit for updating $\ell_{r'}$.}  
Let $\ell_{t}^*, \ell_{r'}^*$ denote the updated value of $\ell_t, \ell_{r'}$ after the length update. 
This circuit acts on qubits indexed from $\ell_t^* + 2$ to $(n + 4 - \ell_{r'}^*)$, with upper bound $K_5(T) = n + 3$, and lower bound defined by
\[
\begin{split}
    \ell_t^* + 2 
    & \ge k_5(T) := \max\left\{\ceil{\dfrac{T-4\delta}{4c}}, 1\right\}.
\end{split}
\]

\end{itemize}

\subsection{Unary iteration}
\label{subsec:unary_iteration}

We will use unary iteration~\cite{babbush2018encoding} to implement location-controlled operations on a known interval of the two \texttt{Work} registers.
Let $\mathsf{L}$ be a \texttt{Length} register promised to store a value $a\in\{k,k+1,\ldots,K\}$, let $\mathsf{W}$ be a \texttt{Work} register, and let $c$ be an optional external control qubit.
The operation has the form
\begin{align*}
    \mathcal{U}_{[k,K]}
    &=
    \ket{0}\!\bra{0}_{c}\otimes I
    +
    \ket{1}\!\bra{1}_{c}\otimes
    \sum_{j=k}^{K}
    \ket{j}\!\bra{j}_{\mathsf{L}}\otimes U_j ,
\end{align*}
where $U_j$ is a unitary operation on $\mathsf{W}$ selected by the value $j$.
Equivalently, the operation performs $\ket{c}\ket{a}\ket{\psi} \mapsto \ket{c}\ket{a} U_a^{\,c}\ket{\psi}$ for all $a\in [k, K]$.

The main primitive used by unary iteration is a reversible AND.
The compute and uncompute forms of this gate are shown in Figure~\ref{fig:and_compute_uncompute}.
In the Toffoli/CNOT counting convention used here, computing an AND costs one Toffoli gate, while uncomputing an AND using measurement costs one CNOT gate with no Toffoli gates.
\begin{figure}[htbp]
\centering
\includegraphics[width=0.65\linewidth]{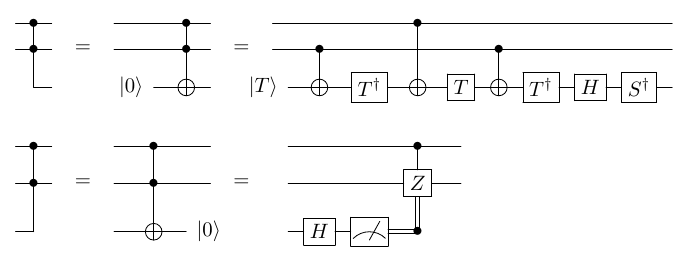}
\caption{Compute and uncompute forms of the reversible AND primitive, reproduced from~\cite{babbush2018encoding}.}
\label{fig:and_compute_uncompute}
\end{figure}

For a general interval, unary iteration can be viewed as a binary tree whose leaves are the labels $k,k+1,\ldots,K$.
Each internal node stores a subset of labels.
If the current node has control bit $g$, we choose the most significant bit position $b$ that is not constant on its labels, thereby separating them into two nonempty subsets according to the value of $A_b$.
For an increasing-order traversal, we first compute $h=g(1-A_b)$ with one reversible AND, use $h$ as the control for the $A_b=0$ child, then switch the same temporary bit by a CNOT, $h\gets h\oplus g$,
so that $h=gA_b$ controls the $A_b=1$ child.
After the second child is processed, $h$ is uncomputed.
At a leaf $j$, the current control bit is exactly
\begin{align*}
    e_j=c[a=j],
\end{align*}
and the circuit applies $U_j$ controlled on $e_j$.
For a decreasing-order traversal, the two children are visited in the opposite order.
Figure~\ref{fig:interval_unary_iteration} illustrates this construction for the interval $\{5,6,7,8\}$.

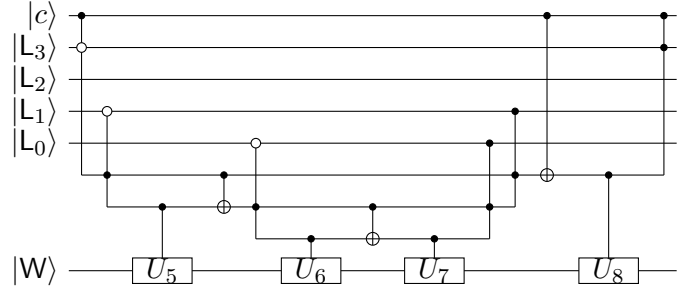
\begin{figure}[htbp]
\centering
\begin{tikzpicture}[scale=0.800000,x=1pt,y=1pt]
\filldraw[color=white] (0.000000, -7.500000) rectangle (286.000000, 127.500000);

\draw[color=black] (0.000000,120.000000) -- (286.000000,120.000000);
\draw[color=black] (0.000000,120.000000) node[left] {$\ket{c}$};
\draw[color=black] (0.000000,105.000000) -- (286.000000,105.000000);
\draw[color=black] (0.000000,105.000000) node[left] {$\ket{\mathsf{L}_3}$};
\draw[color=black] (0.000000,90.000000) -- (286.000000,90.000000);
\draw[color=black] (0.000000,90.000000) node[left] {$\ket{\mathsf{L}_2}$};
\draw[color=black] (0.000000,75.000000) -- (286.000000,75.000000);
\draw[color=black] (0.000000,75.000000) node[left] {$\ket{\mathsf{L}_1}$};
\draw[color=black] (0.000000,60.000000) -- (286.000000,60.000000);
\draw[color=black] (0.000000,60.000000) node[left] {$\ket{\mathsf{L}_0}$};
\draw[color=white] (0.000000,45.000000) -- (6.000000,45.000000);
\draw[color=black] (6.000000,45.000000) -- (280.000000,45.000000);
\draw[color=white] (280.000000,45.000000) -- (286.000000,45.000000);
\draw[color=white] (0.000000,30.000000) -- (18.000000,30.000000);
\draw[color=black] (18.000000,30.000000) -- (210.000000,30.000000);
\draw[color=white] (210.000000,30.000000) -- (286.000000,30.000000);
\draw[color=white] (0.000000,15.000000) -- (88.000000,15.000000);
\draw[color=black] (88.000000,15.000000) -- (198.000000,15.000000);
\draw[color=white] (198.000000,15.000000) -- (286.000000,15.000000);
\draw[color=black] (0.000000,0.000000) -- (286.000000,0.000000);
\draw[color=black] (0.000000,0.000000) node[left] {$\ket{\mathsf{W}}$};
\draw (6.000000,120.000000) -- (6.000000,45.000000);
\begin{scope}
\draw[fill=none] (6.000000, 45.000000) +(-45.000000:0.000000pt) -- +(45.000000:0.000000pt) -- +(135.000000:0.000000pt) -- +(225.000000:0.000000pt) -- cycle;
\clip (6.000000, 45.000000) +(-45.000000:0.000000pt) -- +(45.000000:0.000000pt) -- +(135.000000:0.000000pt) -- +(225.000000:0.000000pt) -- cycle;
\filldraw (6.000000, 45.000000) circle(0.250000pt);
\end{scope}
\filldraw (6.000000, 120.000000) circle(1.500000pt);
\draw[fill=white] (6.000000, 105.000000) circle(2.250000pt);
\draw (18.000000,75.000000) -- (18.000000,30.000000);
\begin{scope}
\draw[fill=none] (18.000000, 30.000000) +(-45.000000:0.000000pt) -- +(45.000000:0.000000pt) -- +(135.000000:0.000000pt) -- +(225.000000:0.000000pt) -- cycle;
\clip (18.000000, 30.000000) +(-45.000000:0.000000pt) -- +(45.000000:0.000000pt) -- +(135.000000:0.000000pt) -- +(225.000000:0.000000pt) -- cycle;
\filldraw (18.000000, 30.000000) circle(0.250000pt);
\end{scope}
\filldraw (18.000000, 45.000000) circle(1.500000pt);
\draw[fill=white] (18.000000, 75.000000) circle(2.250000pt);
\draw (44.000000,30.000000) -- (44.000000,0.000000);
\begin{scope}
\draw[fill=white] (44.000000, 0.000000) +(-45.000000:19.798990pt and 8.485281pt) -- +(45.000000:19.798990pt and 8.485281pt) -- +(135.000000:19.798990pt and 8.485281pt) -- +(225.000000:19.798990pt and 8.485281pt) -- cycle;
\clip (44.000000, 0.000000) +(-45.000000:19.798990pt and 8.485281pt) -- +(45.000000:19.798990pt and 8.485281pt) -- +(135.000000:19.798990pt and 8.485281pt) -- +(225.000000:19.798990pt and 8.485281pt) -- cycle;
\draw (44.000000, 0.000000) node {$U_5$};
\end{scope}
\filldraw (44.000000, 30.000000) circle(1.500000pt);
\draw (73.000000,45.000000) -- (73.000000,30.000000);
\begin{scope}
\draw[fill=white] (73.000000, 30.000000) circle(3.000000pt);
\clip (73.000000, 30.000000) circle(3.000000pt);
\draw (70.000000, 30.000000) -- (76.000000, 30.000000);
\draw (73.000000, 27.000000) -- (73.000000, 33.000000);
\end{scope}
\filldraw (73.000000, 45.000000) circle(1.500000pt);
\draw (88.000000,60.000000) -- (88.000000,15.000000);
\begin{scope}
\draw[fill=none] (88.000000, 15.000000) +(-45.000000:0.000000pt) -- +(45.000000:0.000000pt) -- +(135.000000:0.000000pt) -- +(225.000000:0.000000pt) -- cycle;
\clip (88.000000, 15.000000) +(-45.000000:0.000000pt) -- +(45.000000:0.000000pt) -- +(135.000000:0.000000pt) -- +(225.000000:0.000000pt) -- cycle;
\filldraw (88.000000, 15.000000) circle(0.250000pt);
\end{scope}
\filldraw (88.000000, 30.000000) circle(1.500000pt);
\draw[fill=white] (88.000000, 60.000000) circle(2.250000pt);
\draw (114.000000,15.000000) -- (114.000000,0.000000);
\begin{scope}
\draw[fill=white] (114.000000, 0.000000) +(-45.000000:19.798990pt and 8.485281pt) -- +(45.000000:19.798990pt and 8.485281pt) -- +(135.000000:19.798990pt and 8.485281pt) -- +(225.000000:19.798990pt and 8.485281pt) -- cycle;
\clip (114.000000, 0.000000) +(-45.000000:19.798990pt and 8.485281pt) -- +(45.000000:19.798990pt and 8.485281pt) -- +(135.000000:19.798990pt and 8.485281pt) -- +(225.000000:19.798990pt and 8.485281pt) -- cycle;
\draw (114.000000, 0.000000) node {$U_6$};
\end{scope}
\filldraw (114.000000, 15.000000) circle(1.500000pt);
\draw (143.000000,30.000000) -- (143.000000,15.000000);
\begin{scope}
\draw[fill=white] (143.000000, 15.000000) circle(3.000000pt);
\clip (143.000000, 15.000000) circle(3.000000pt);
\draw (140.000000, 15.000000) -- (146.000000, 15.000000);
\draw (143.000000, 12.000000) -- (143.000000, 18.000000);
\end{scope}
\filldraw (143.000000, 30.000000) circle(1.500000pt);
\draw (172.000000,15.000000) -- (172.000000,0.000000);
\begin{scope}
\draw[fill=white] (172.000000, 0.000000) +(-45.000000:19.798990pt and 8.485281pt) -- +(45.000000:19.798990pt and 8.485281pt) -- +(135.000000:19.798990pt and 8.485281pt) -- +(225.000000:19.798990pt and 8.485281pt) -- cycle;
\clip (172.000000, 0.000000) +(-45.000000:19.798990pt and 8.485281pt) -- +(45.000000:19.798990pt and 8.485281pt) -- +(135.000000:19.798990pt and 8.485281pt) -- +(225.000000:19.798990pt and 8.485281pt) -- cycle;
\draw (172.000000, 0.000000) node {$U_7$};
\end{scope}
\filldraw (172.000000, 15.000000) circle(1.500000pt);
\draw (198.000000,60.000000) -- (198.000000,15.000000);
\begin{scope}
\draw[fill=none] (198.000000, 15.000000) +(-45.000000:0.000000pt) -- +(45.000000:0.000000pt) -- +(135.000000:0.000000pt) -- +(225.000000:0.000000pt) -- cycle;
\clip (198.000000, 15.000000) +(-45.000000:0.000000pt) -- +(45.000000:0.000000pt) -- +(135.000000:0.000000pt) -- +(225.000000:0.000000pt) -- cycle;
\filldraw (198.000000, 15.000000) circle(0.250000pt);
\end{scope}
\filldraw (198.000000, 30.000000) circle(1.500000pt);
\filldraw (198.000000, 60.000000) circle(1.500000pt);
\draw (210.000000,75.000000) -- (210.000000,30.000000);
\begin{scope}
\draw[fill=none] (210.000000, 30.000000) +(-45.000000:0.000000pt) -- +(45.000000:0.000000pt) -- +(135.000000:0.000000pt) -- +(225.000000:0.000000pt) -- cycle;
\clip (210.000000, 30.000000) +(-45.000000:0.000000pt) -- +(45.000000:0.000000pt) -- +(135.000000:0.000000pt) -- +(225.000000:0.000000pt) -- cycle;
\filldraw (210.000000, 30.000000) circle(0.250000pt);
\end{scope}
\filldraw (210.000000, 45.000000) circle(1.500000pt);
\filldraw (210.000000, 75.000000) circle(1.500000pt);
\draw (225.000000,120.000000) -- (225.000000,45.000000);
\begin{scope}
\draw[fill=white] (225.000000, 45.000000) circle(3.000000pt);
\clip (225.000000, 45.000000) circle(3.000000pt);
\draw (222.000000, 45.000000) -- (228.000000, 45.000000);
\draw (225.000000, 42.000000) -- (225.000000, 48.000000);
\end{scope}
\filldraw (225.000000, 120.000000) circle(1.500000pt);
\draw (254.000000,45.000000) -- (254.000000,0.000000);
\begin{scope}
\draw[fill=white] (254.000000, 0.000000) +(-45.000000:19.798990pt and 8.485281pt) -- +(45.000000:19.798990pt and 8.485281pt) -- +(135.000000:19.798990pt and 8.485281pt) -- +(225.000000:19.798990pt and 8.485281pt) -- cycle;
\clip (254.000000, 0.000000) +(-45.000000:19.798990pt and 8.485281pt) -- +(45.000000:19.798990pt and 8.485281pt) -- +(135.000000:19.798990pt and 8.485281pt) -- +(225.000000:19.798990pt and 8.485281pt) -- cycle;
\draw (254.000000, 0.000000) node {$U_8$};
\end{scope}
\filldraw (254.000000, 45.000000) circle(1.500000pt);
\draw (280.000000,120.000000) -- (280.000000,45.000000);
\begin{scope}
\draw[fill=none] (280.000000, 45.000000) +(-45.000000:0.000000pt) -- +(45.000000:0.000000pt) -- +(135.000000:0.000000pt) -- +(225.000000:0.000000pt) -- cycle;
\clip (280.000000, 45.000000) +(-45.000000:0.000000pt) -- +(45.000000:0.000000pt) -- +(135.000000:0.000000pt) -- +(225.000000:0.000000pt) -- cycle;
\filldraw (280.000000, 45.000000) circle(0.250000pt);
\end{scope}
\filldraw (280.000000, 120.000000) circle(1.500000pt);
\filldraw (280.000000, 105.000000) circle(1.500000pt);
\end{tikzpicture}
\caption{Pruned unary iteration on the promised interval $A\in\{5,6,7,8\}$.
The labels are split by the bits of $A$ only when the split separates the remaining candidates.
The temporary control is reused to visit the two children of each split, and the leaves apply $U_5,U_6,U_7,U_8$ in increasing order.}
\label{fig:interval_unary_iteration}
\end{figure}

If $M=K-k+1$, the tree has $M$ leaves and $M-1$ internal nodes.
Thus the selection logic contributes $M-1=K-k$ Toffoli gates, while the leaf operations contribute the cost of the controlled $U_j$ gates themselves.
For $M\ge2$, the recursion depth, and hence the number of temporary control qubits, is at most $\left\lceil \log_2(M-1)\right\rceil+1$.
For the trivial case $M=1$, no temporary control qubit is needed.
This replaces independent equality tests for all $j$, which would cost $O(M\log_2 n)$ Toffoli gates for an $O(\log_2 n)$-bit index register, by a linear-size selection circuit using only $O(\log_2 M)$ temporary control qubits.

\subsection{Location-controlled operations}
\label{subsec:blocks}

In this subsection, the two \texttt{Work} registers are indexed, from left to right, by $u_1, u_2, \ldots, u_{n+3}$ and $v_1, v_2, \ldots, v_{n+3}$, respectively.

\paragraph{Location-controlled $\mathsf{Swap}$ circuit.}
We first consider the location-controlled swap operation appearing in Figure \ref{fig:all_circuit_1}.
The goal is to swap the \texttt{Sign} qubit with the $(\ell_t+\ell_q+1)$-th qubit of \texttt{Work1}, while the whole operation is controlled by an external qubit \texttt{Ctrl}.
Let $J = \ell_t+\ell_q+1$, where $k$ and $K$ are known constants such that $k\le J\le K$.
After temporarily storing $J$ in the $\ell_q$ register, the location-controlled swap can be written as
\begin{align*}
    \mathsf{Swap}_{[k,K]}
    &=
    \ket{0}\!\bra{0}_{\texttt{Ctrl}}\otimes I
    +
    \ket{1}\!\bra{1}_{\texttt{Ctrl}}
    \otimes
    \sum_{j=k}^{K}
    \ket{j}\!\bra{j}_{\ell_q}
    \otimes
    \mathsf{SWAP}\!\left(\texttt{Sign},u_j\right),
\end{align*}
where the subscript $\ell_q$ denotes the temporary register storing $J$.

This operation can be implemented directly using unary iteration.
We first reversibly transform the $\ell_q$ register as $\ket{\ell_t}\ket{\ell_q} \longmapsto \ket{\ell_t}\ket{\ell_t+\ell_q+1}$.
Using \texttt{Ctrl} as the global control, pruned unary iteration over the labels $k,k+1,\ldots,K$ then produces, one value at a time, the control bit $e_j=\texttt{Ctrl}\land [J=j]$.
When this unary control is active, we swap \texttt{Sign} and $u_j$.
After it finishes we restore the $\ell_q$ register to its original value.
Figure~\ref{fig:lc_swap} illustrates this construction for the example $k=5$ and $K=8$.

\begin{figure}[htbp]
    \centering
    \resizebox{0.75\textwidth}{!}{
        \input{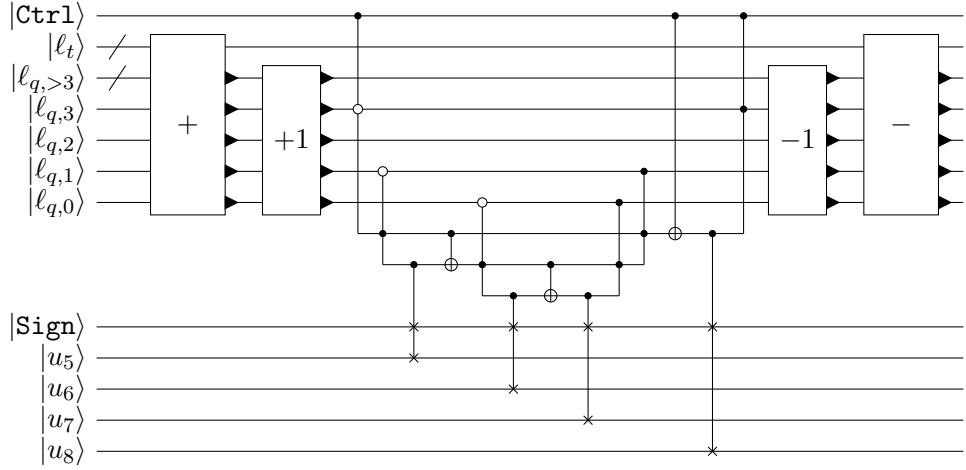}
    }
    \caption{Location-controlled swap circuit for $k=5$ and $K=8$. After computing $J=\ell_t+\ell_q+1$, pruned unary iteration over the labels $5,6,7,8$ selects the work qubit to be swapped with \texttt{Sign}.}
    \label{fig:lc_swap}
\end{figure}

The unary iterator contributes $K-k+O(1)$ Toffoli gates and $2(K-k)+O(1)$ CNOT gates, the controlled swaps contribute $K-k+O(1)$ Toffoli gates and $2(K-k)+O(1)$ CNOT gates. The reversible transformations of the $\ell_q$ register contribute $O(\log_2 n)$ Toffoli/CNOT gates.
Thus the location-controlled $\mathsf{Swap}$ circuit uses $2(K-k)+O(\log_2 n)$ Toffoli gates and $4(K-k)+O(\log_2 n)$ CNOT gates.

\paragraph{Location-controlled $\mathsf{Add, Sub}$ circuit.}
We next describe the location-controlled addition and subtraction circuits.
As a representative example, consider the operation
\begin{align}
    (\texttt{Sign},r)\gets (\texttt{Sign},r)-2^{\ell_s}r',
    \label{eq:lc_sub_example}
\end{align}
where \texttt{Sign} records the sign of the subtraction result.
In terms of physical positions in the two \texttt{Work} registers, this operation acts only on the interval
\begin{align*}
    L &= \ell_t+\ell_q+2, &
    R &= n+3-\ell_s .
\end{align*}
That is, the relevant qubits are $u_L,\ldots,u_R$ in \texttt{Work1} and $v_L,\ldots,v_R$ in \texttt{Work2}; outside this interval all qubits are left unchanged.
The shift by $\ell_s$ is already reflected in the endpoint $R=n+3-\ell_s$, so the aligned bits $v_L,\ldots,v_R$ represent the contribution of $2^{\ell_s}r'$ to the subtraction.
For a known active window $[k,K]$, we promise that $k\le L\le R\le K$.

If the endpoints $L$ and $R$ were fixed, the subtraction could be implemented by a ripple-carry subtractor on the qubit pairs $(u_j,v_j)$ for $j\in[L,R]$ with a dedicated carry qubit $c$ initialized to $\ket{0}$.
For the subtraction in Equation~\eqref{eq:lc_sub_example}, we first apply the $\mathsf{UMA}^{\dagger}$ cells from high to low, update \texttt{Sign} from the carry qubit, and then apply the $\mathsf{MAJ}^{\dagger}$ cells from low to high to uncompute the carry.

It remains to coherently select the interval $[L,R]$.
The endpoint registers storing $L$ and $R$ are obtained from the length registers by reversible affine transformations, which are undone after the arithmetic block.
The circuit uses one accumulator bit $s$, initialized to $\ket{0}$.
In the $\mathsf{UMA}^{\dagger}$ pass, the circuit scans the physical positions from high to low, namely $j=K,K-1,\ldots,k$.
For each scanned position $j$, pruned unary iteration on the endpoint registers coherently produces the two local controls
\begin{align*}
    \lambda_j &= \texttt{Ctrl}\land [R=j], &
    \mu_j &= \texttt{Ctrl}\land [L=j].
\end{align*}

The circuit first updates the accumulator by $s \gets s\oplus \lambda_j$.
If the updated value of $s$ is $1$, the circuit applies $\mathsf{UMA}^{\dagger}(c,u_j,v_j)$; otherwise it does nothing on $(c,u_j,v_j)$.
After this controlled arithmetic cell, the circuit updates the accumulator by $s \gets s\oplus \mu_j$.
Therefore, when the scan reaches $j=R$, the bit $\lambda_j$ flips $s$ from $0$ to $1$ before the arithmetic cell, so the $\mathsf{UMA}^{\dagger}$ cell at $j=R$ is applied.
For positions strictly inside the interval, the accumulator remains equal to $1$, so the corresponding arithmetic cells are also applied.
When the scan reaches $j=L$, the $\mathsf{UMA}^{\dagger}$ cell is still applied, and then $\mu_j$ flips $s$ back to $0$.
The same update rule is used for every scanned position $j\in[k,K]$.
Before the active interval is reached, all arithmetic cells are disabled, so the carry qubit $c$ remains equal to $0$ at the interval entrance.

After the $\mathsf{UMA}^{\dagger}$ pass, $c$ stores the carry bit of the selected subtraction.
We update \texttt{Sign} by
\begin{align*}
    \texttt{Sign}\gets \texttt{Sign}\oplus c.
\end{align*}

The $\mathsf{MAJ}^{\dagger}$ pass visits $j=k,k+1,\ldots,K$. 
At each position, the accumulator is first updated by $s \gets s\oplus \mu_j$.
If the updated value of $s$ is $1$, the circuit applies $\mathsf{MAJ}^{\dagger}(c,u_j,v_j)$.
It then updates the accumulator by $s \gets s\oplus \lambda_j$.
Therefore the $\mathsf{MAJ}^{\dagger}$ cells are applied exactly for $j\in[L,R]$, and $s$ and $c$ are restored to $\ket{0}$ at the end.

Figure~\ref{fig:lc_add} illustrates the complete pattern for the example $k=5$ and $K=7$.

\begin{figure}[htbp]
    \centering
    \resizebox{1\textwidth}{!}{
        \input{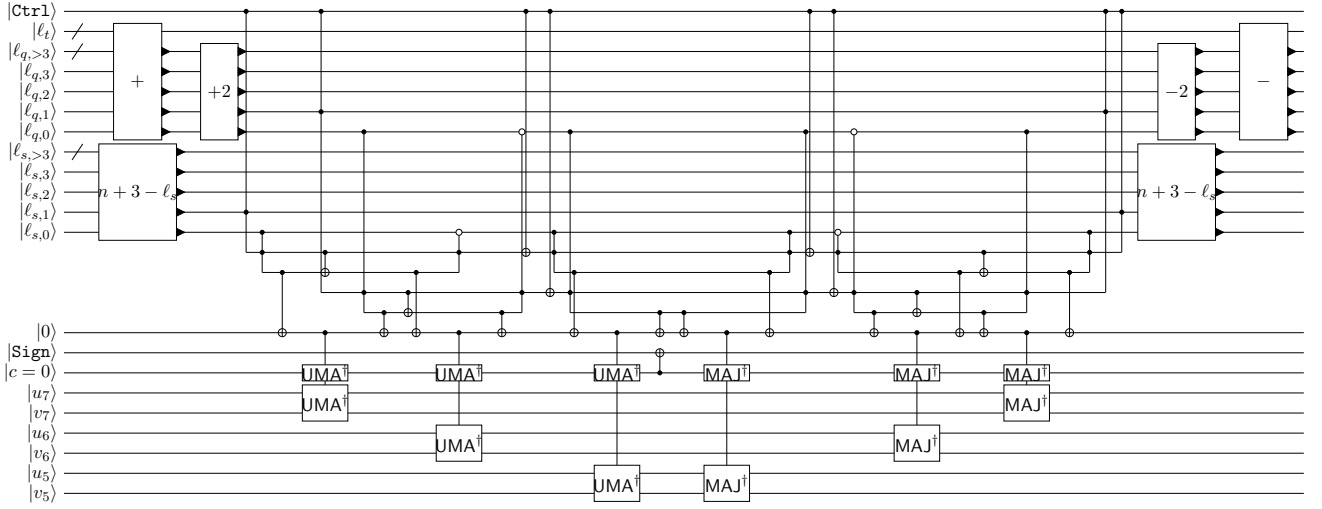}
    }
    \caption{Location-controlled subtraction circuit for $k=5$ and $K=7$. The $\mathsf{UMA}^{\dagger}$ pass computes the carry, \texttt{Sign} is updated from $c$, and the $\mathsf{MAJ}^{\dagger}$ pass restores the carry qubit.}
    \label{fig:lc_add}
\end{figure}

Finally, we explain how the controlled $\mathsf{MAJ}$ and $\mathsf{UMA}$ cells used above are implemented.
We do not add the location-control bit to every gate in the cell; only the Toffoli gate that updates the carry qubit is controlled.
Figure~\ref{fig:controlled_maj_uma_cells} shows the corresponding controlled cells.

\begin{figure}[ht]
    \centering
    \begin{subfigure}{0.45\linewidth}
        \centering
        \includegraphics[height=2.8cm]{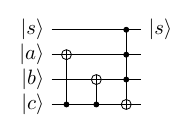}
        \caption{Controlled $\mathsf{MAJ}$ cell.}
        \label{fig:controlled_maj_cell}
    \end{subfigure}
    \hfill
    \begin{subfigure}{0.45\linewidth}
        \centering
        \includegraphics[height=2.8cm]{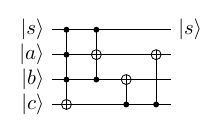}
        \caption{Controlled $\mathsf{UMA}$ cell.}
        \label{fig:controlled_uma_cell}
    \end{subfigure}
    \caption{Controlled arithmetic cells used in the location-controlled add/sub circuits.}
    \label{fig:controlled_maj_uma_cells}
\end{figure}

The gate count consists of the following contributions:
\begin{itemize}
    \item Generating $[L=j]$ and $[R=j]$ during both scans requires four applications of
    unary iteration. They use $4(K-k)+O(1)$ Toffoli gates and
    $8(K-k)+O(1)$ CNOT gates in total.
    \item Applying one controlled $\mathsf{UMA}^{\dagger}$ cell at each position in the
    high-to-low scan uses
    $4(K-k)+O(1)$ Toffoli gates and $2(K-k)+O(1)$ CNOT gates.
    \item Applying one controlled $\mathsf{MAJ}^{\dagger}$ cell at each position in the
    low-to-high scan uses
    $3(K-k)+O(1)$ Toffoli gates and $2(K-k)+O(1)$ CNOT gates.
    \item Computing $L$ and $R$ from $\ell_q$ and $\ell_s$, and subsequently restoring
    these registers, uses
    $O(\log_2 n)$ Toffoli and CNOT gates.
\end{itemize}
Adding these terms, the location-controlled $\mathsf{Sub}$ circuit implementing
Equation~\eqref{eq:lc_sub_example} uses $11(K-k)+O(\log_2 n)$ Toffoli gates and
$12(K-k)+O(\log_2 n)$ CNOT gates.

\paragraph{Location-controlled length-update circuit.}
Finally, we describe the length-update circuit.
We focus on updating $\ell_t$; the circuit for updating $\ell_{r'}$ is analogous.
At the end of one EEA iteration, after the full swap of \texttt{Work1} and \texttt{Work2}, the left part of \texttt{Work1} stores the current value of $t$, while the left part of \texttt{Work2} stores the current value of $t'$.
Let $L_{\mathrm{new}} = \operatorname{len}(t), L_{\mathrm{old}} = \operatorname{len}(t')$.
Since the value stored in $\ell_t$ before the update is $L_{\mathrm{old}}$, it suffices to implement
\begin{align*}
    \ell_t
    \gets
    \ell_t\oplus L_{\mathrm{old}}\oplus L_{\mathrm{new}} .
\end{align*}
Thus the length update is reduced to two serial constant-XOR writes: first XOR the old length, and then XOR the new length.

We first describe the subroutine for XORing $L_{\mathrm{new}}$ into $\ell_t$.
This subroutine finds the highest valid nonzero position among the qubits $u_k,\ldots,u_K$.
Only positions in the left region not occupied by $r'$ should be inspected.
Let $J=n+3-\ell_{r'}$ be the right boundary of this valid region with the promise $k\le \ell_t \le J\le K$ from the active window under consideration.
For each $j\in[k,K]$, define the masked bit
\begin{align*}
    a_j=\texttt{Ctrl}\land [j\le J]\land u_j .
\end{align*}
Also define bits
\begin{align*}
    z_j=\bigwedge_{h=j}^{K}\neg a_h .
\end{align*}
Thus $z_j=1$ means that no valid nonzero bit occurs among positions $j,j+1,\ldots,K$.
If the highest valid nonzero position is $m$, then
\begin{align*}
    z_K=\cdots=z_{m+1}=1,
    \qquad
    z_m=\cdots=z_k=0 .
\end{align*}
If no valid nonzero bit exists in the window, then all $z_j$ are equal to $1$.

The highest position can be written using a telescoping XOR.
We first XOR the classical constant $K$ into the length register.
Then, for $j=K,K-1,\ldots,k+1$, whenever $z_j=1$, we change the current candidate from $j$ to $j-1$ by XORing the classical constant $j\oplus(j-1)$.
Finally, if $z_k=1$, then the entire window contains no valid nonzero bit, and we clear the remaining candidate $k$ by XORing $k$.
Equivalently, the value written by the subroutine is
\begin{align*}
    L_{\mathrm{new}}
    =
    K
    \oplus
    \bigoplus_{j=k+1}^{K}
        z_j\cdot \bigl(j\oplus(j-1)\bigr)
    \oplus
    z_k\cdot k ,
\end{align*}
where a term such as $z_j\cdot C$ means that the classical bit string $C$ is XORed into the length register controlled on $z_j$.
If the highest valid nonzero position is $m$, the expression telescopes to $m$.
If no such position exists, it telescopes to zero.

It remains to implement the constant writes controlled by the bits $z_j$ without storing all of them in clean ancillae.
Borrow dirty qubits $g_k,\ldots,g_K$, which will be restored at the end.
For $j>k$, let $V_j$ denote the Clifford operation that XORs the classical constant $j\oplus(j-1)$ into the length register, and let $V_k$ denote the operation that XORs $k$.
Since each $V_j$ is a product of CNOT gates, $V_j^2=I$.
Assume for the moment that we can perform the transformation
\begin{align*}
    g_j\gets g_j\oplus z_j
    \qquad (j=k,\ldots,K).
\end{align*}
Then the desired controlled write is implemented by placing this transformation between two uses of $V_j$ controlled by the borrowed qubit:
\begin{align*}
    V_j^{g_j},\qquad
    g_j\gets g_j\oplus z_j,\qquad
    V_j^{g_j}.
\end{align*}
Indeed, if the initial value of the borrowed qubit is $G_j$, then the two controlled writes are
\begin{align*}
    V_j^{G_j}V_j^{G_j\oplus z_j}=V_j^{z_j},
\end{align*}
so the unknown initial value cancels.
After all constant writes are completed, we perform the same transformation again to restore every $g_j$ to its original value.
Figure~\ref{fig:length-write} illustrates this method for the example $k=5$ and $K=7$.
\begin{figure}[htbp]
    \centering
    \resizebox{1\textwidth}{!}{
        \begin{tikzpicture}[scale=0.800000,x=1pt,y=1pt]
\filldraw[color=white] (0.000000, -7.500000) rectangle (594.000000, 142.500000);
\draw[color=black] (0.000000,135.000000) -- (594.000000,135.000000);
\draw[color=black] (0.000000,135.000000) node[left] {$\ket{\ell_t}$};
\draw[color=black] (0.000000,120.000000) -- (594.000000,120.000000);
\draw[color=black] (0.000000,120.000000) node[left] {$\ket{\texttt{Ctrl}}$};
\draw[color=black] (0.000000,105.000000) -- (594.000000,105.000000);
\draw[color=black] (0.000000,105.000000) node[left] {$\ket{\ell_{r^\prime}}$};
\draw[color=black] (0.000000,90.000000) -- (594.000000,90.000000);
\draw[color=black] (0.000000,90.000000) node[left] {$\ket{u_7}$};
\draw[color=black] (0.000000,75.000000) -- (594.000000,75.000000);
\draw[color=black] (0.000000,75.000000) node[left] {$\ket{u_6}$};
\draw[color=black] (0.000000,60.000000) -- (594.000000,60.000000);
\draw[color=black] (0.000000,60.000000) node[left] {$\ket{u_5}$};
\draw[color=black] (0.000000,45.000000) -- (594.000000,45.000000);
\draw[color=black] (0.000000,45.000000) node[left] {$\ket{0^4}$};
\draw[color=black] (0.000000,30.000000) -- (594.000000,30.000000);
\draw[color=black] (0.000000,30.000000) node[left] {$\ket{g_7}$};
\draw[color=black] (0.000000,15.000000) -- (594.000000,15.000000);
\draw[color=black] (0.000000,15.000000) node[left] {$\ket{g_6}$};
\draw[color=black] (0.000000,0.000000) -- (594.000000,0.000000);
\draw[color=black] (0.000000,0.000000) node[left] {$\ket{g_5}$};
\draw (21.000000, 39.000000) -- (29.000000, 51.000000);
\begin{scope}
\draw[fill=white] (25.000000, 135.000000) +(-45.000000:26.870058pt and 8.485281pt) -- +(45.000000:26.870058pt and 8.485281pt) -- +(135.000000:26.870058pt and 8.485281pt) -- +(225.000000:26.870058pt and 8.485281pt) -- cycle;
\clip (25.000000, 135.000000) +(-45.000000:26.870058pt and 8.485281pt) -- +(45.000000:26.870058pt and 8.485281pt) -- +(135.000000:26.870058pt and 8.485281pt) -- +(225.000000:26.870058pt and 8.485281pt) -- cycle;
\draw (25.000000, 135.000000) node {$\oplus 7$};
\end{scope}
\draw (87.000000,135.000000) -- (87.000000,30.000000);
\begin{scope}
\draw[fill=white] (87.000000, 135.000000) +(-45.000000:43.840620pt and 8.485281pt) -- +(45.000000:43.840620pt and 8.485281pt) -- +(135.000000:43.840620pt and 8.485281pt) -- +(225.000000:43.840620pt and 8.485281pt) -- cycle;
\clip (87.000000, 135.000000) +(-45.000000:43.840620pt and 8.485281pt) -- +(45.000000:43.840620pt and 8.485281pt) -- +(135.000000:43.840620pt and 8.485281pt) -- +(225.000000:43.840620pt and 8.485281pt) -- cycle;
\draw (87.000000, 135.000000) node {$\oplus(7\oplus6)$};
\end{scope}
\filldraw (87.000000, 30.000000) circle(1.500000pt);
\draw (161.000000,135.000000) -- (161.000000,15.000000);
\begin{scope}
\draw[fill=white] (161.000000, 135.000000) +(-45.000000:43.840620pt and 8.485281pt) -- +(45.000000:43.840620pt and 8.485281pt) -- +(135.000000:43.840620pt and 8.485281pt) -- +(225.000000:43.840620pt and 8.485281pt) -- cycle;
\clip (161.000000, 135.000000) +(-45.000000:43.840620pt and 8.485281pt) -- +(45.000000:43.840620pt and 8.485281pt) -- +(135.000000:43.840620pt and 8.485281pt) -- +(225.000000:43.840620pt and 8.485281pt) -- cycle;
\draw (161.000000, 135.000000) node {$\oplus(6\oplus5)$};
\end{scope}
\filldraw (161.000000, 15.000000) circle(1.500000pt);
\draw (223.000000,135.000000) -- (223.000000,0.000000);
\begin{scope}
\draw[fill=white] (223.000000, 135.000000) +(-45.000000:26.870058pt and 8.485281pt) -- +(45.000000:26.870058pt and 8.485281pt) -- +(135.000000:26.870058pt and 8.485281pt) -- +(225.000000:26.870058pt and 8.485281pt) -- cycle;
\clip (223.000000, 135.000000) +(-45.000000:26.870058pt and 8.485281pt) -- +(45.000000:26.870058pt and 8.485281pt) -- +(135.000000:26.870058pt and 8.485281pt) -- +(225.000000:26.870058pt and 8.485281pt) -- cycle;
\draw (223.000000, 135.000000) node {$\oplus 5$};
\end{scope}
\filldraw (223.000000, 0.000000) circle(1.500000pt);
\draw (285.000000,120.000000) -- (285.000000,0.000000);
\begin{scope}
\draw[fill=white] (285.000000, 60.000000) +(-45.000000:43.840620pt and 93.338095pt) -- +(45.000000:43.840620pt and 93.338095pt) -- +(135.000000:43.840620pt and 93.338095pt) -- +(225.000000:43.840620pt and 93.338095pt) -- cycle;
\clip (285.000000, 60.000000) +(-45.000000:43.840620pt and 93.338095pt) -- +(45.000000:43.840620pt and 93.338095pt) -- +(135.000000:43.840620pt and 93.338095pt) -- +(225.000000:43.840620pt and 93.338095pt) -- cycle;
\draw (285.000000, 60.000000) node {$g_j\oplus z_j$};
\end{scope}
\draw (359.000000,135.000000) -- (359.000000,30.000000);
\begin{scope}
\draw[fill=white] (359.000000, 135.000000) +(-45.000000:43.840620pt and 8.485281pt) -- +(45.000000:43.840620pt and 8.485281pt) -- +(135.000000:43.840620pt and 8.485281pt) -- +(225.000000:43.840620pt and 8.485281pt) -- cycle;
\clip (359.000000, 135.000000) +(-45.000000:43.840620pt and 8.485281pt) -- +(45.000000:43.840620pt and 8.485281pt) -- +(135.000000:43.840620pt and 8.485281pt) -- +(225.000000:43.840620pt and 8.485281pt) -- cycle;
\draw (359.000000, 135.000000) node {$\oplus(7\oplus6)$};
\end{scope}
\filldraw (359.000000, 30.000000) circle(1.500000pt);
\draw (433.000000,135.000000) -- (433.000000,15.000000);
\begin{scope}
\draw[fill=white] (433.000000, 135.000000) +(-45.000000:43.840620pt and 8.485281pt) -- +(45.000000:43.840620pt and 8.485281pt) -- +(135.000000:43.840620pt and 8.485281pt) -- +(225.000000:43.840620pt and 8.485281pt) -- cycle;
\clip (433.000000, 135.000000) +(-45.000000:43.840620pt and 8.485281pt) -- +(45.000000:43.840620pt and 8.485281pt) -- +(135.000000:43.840620pt and 8.485281pt) -- +(225.000000:43.840620pt and 8.485281pt) -- cycle;
\draw (433.000000, 135.000000) node {$\oplus(6\oplus5)$};
\end{scope}
\filldraw (433.000000, 15.000000) circle(1.500000pt);
\draw (495.000000,135.000000) -- (495.000000,0.000000);
\begin{scope}
\draw[fill=white] (495.000000, 135.000000) +(-45.000000:26.870058pt and 8.485281pt) -- +(45.000000:26.870058pt and 8.485281pt) -- +(135.000000:26.870058pt and 8.485281pt) -- +(225.000000:26.870058pt and 8.485281pt) -- cycle;
\clip (495.000000, 135.000000) +(-45.000000:26.870058pt and 8.485281pt) -- +(45.000000:26.870058pt and 8.485281pt) -- +(135.000000:26.870058pt and 8.485281pt) -- +(225.000000:26.870058pt and 8.485281pt) -- cycle;
\draw (495.000000, 135.000000) node {$\oplus 5$};
\end{scope}
\filldraw (495.000000, 0.000000) circle(1.500000pt);
\draw (557.000000,120.000000) -- (557.000000,0.000000);
\begin{scope}
\draw[fill=white] (557.000000, 60.000000) +(-45.000000:43.840620pt and 93.338095pt) -- +(45.000000:43.840620pt and 93.338095pt) -- +(135.000000:43.840620pt and 93.338095pt) -- +(225.000000:43.840620pt and 93.338095pt) -- cycle;
\clip (557.000000, 60.000000) +(-45.000000:43.840620pt and 93.338095pt) -- +(45.000000:43.840620pt and 93.338095pt) -- +(135.000000:43.840620pt and 93.338095pt) -- +(225.000000:43.840620pt and 93.338095pt) -- cycle;
\draw (557.000000, 60.000000) node {$g_j\oplus z_j$};
\end{scope}
\draw[color=black] (594.000000,135.000000) node[right] {$\ket{\ell_t\oplus L_{\mathrm{new}}}$};
\draw[color=black] (594.000000,120.000000) node[right] {$\ket{\texttt{Ctrl}}$};
\draw[color=black] (594.000000,105.000000) node[right] {$\ket{\ell_{r^\prime}}$};
\draw[color=black] (594.000000,90.000000) node[right] {$\ket{u_7}$};
\draw[color=black] (594.000000,75.000000) node[right] {$\ket{u_6}$};
\draw[color=black] (594.000000,60.000000) node[right] {$\ket{u_5}$};
\draw[color=black] (594.000000,45.000000) node[right] {$\ket{0^4}$};
\draw[color=black] (594.000000,30.000000) node[right] {$\ket{g_7}$};
\draw[color=black] (594.000000,15.000000) node[right] {$\ket{g_6}$};
\draw[color=black] (594.000000,0.000000) node[right] {$\ket{g_5}$};
\end{tikzpicture}
    }
    \caption{Writing $L_{\mathrm{new}}$ using borrowed qubits for $k=5$ and $K=7$. The boxes labeled $g_j\oplus z_j$ denote the transformation $g_j\gets g_j\oplus z_j$ for all $j$ in the window.}
    \label{fig:length-write}
\end{figure}
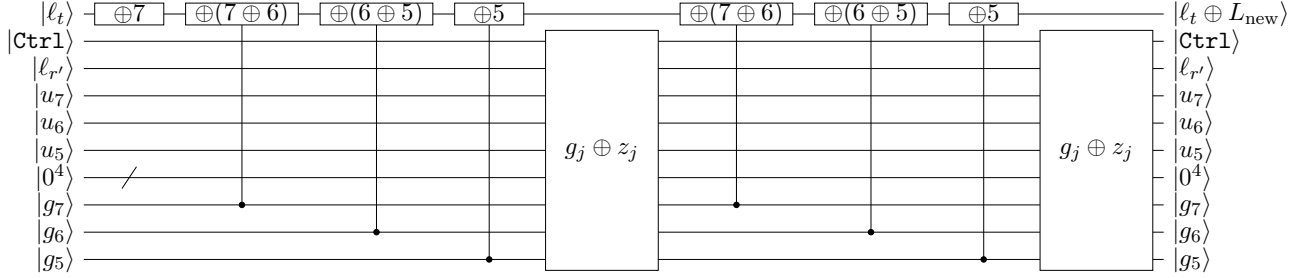

We now explain how to implement the transformation $g_j\gets g_j\oplus z_j$ for all $j$.
The recurrence is
\begin{align*}
    z_K &= \neg a_K, &
    z_j &= \neg a_j\, z_{j+1}
    \qquad (j<K).
\end{align*}
To implement this transformation, first apply, for $j=k,k+1,\ldots,K-1$, the Toffoli gates $g_j\gets g_j\oplus (\neg a_j)g_{j+1}$.
Then apply the base CNOT gate $g_K\gets g_K\oplus \neg a_K$.
Finally apply the same Toffoli gates in reverse order, for $j=K-1,K-2,\ldots,k$.
For a fixed $j<K$, the two Toffoli gates acting on $g_j$ are separated by the transformation $g_{j+1}\gets g_{j+1}\oplus z_{j+1}$.
Therefore their combined action on $g_j$ is
\begin{align*}
    g_j
    &\gets
    g_j
    \oplus
    (\neg a_j)g_{j+1}
    \oplus
    (\neg a_j)(g_{j+1}\oplus z_{j+1}) \notag\\
    &=
    g_j\oplus (\neg a_j)z_{j+1}
    =
    g_j\oplus z_j .
\end{align*}
The bit $a_j$ is produced only locally when the corresponding gate is applied.
During a scan of the window, unary iteration on the boundary register supplies the temporary range control
\begin{align*}
    b_j=\texttt{Ctrl}\land [j\le J].
\end{align*}
The circuit then computes $a_j=b_j\land u_j$ into one clean temporary qubit, uses it as the negative control in the corresponding Toffoli or CNOT, and immediately uncomputes it.
The required range controls are supplied by one scan in the order $k,k+1,\ldots,K$ and one scan in the reverse order $K,K-1,\ldots,k$.

Figure~\ref{fig:length-z-transform} shows this transformation for the example $k=5$ and $K=7$.
The boundary value is computed in place on the $\ell_{r'}$ register, and the temporary bit $a$ is created only when the corresponding update gate is applied.
\begin{figure}[htbp]
    \centering
    \resizebox{1\textwidth}{!}{
        \input{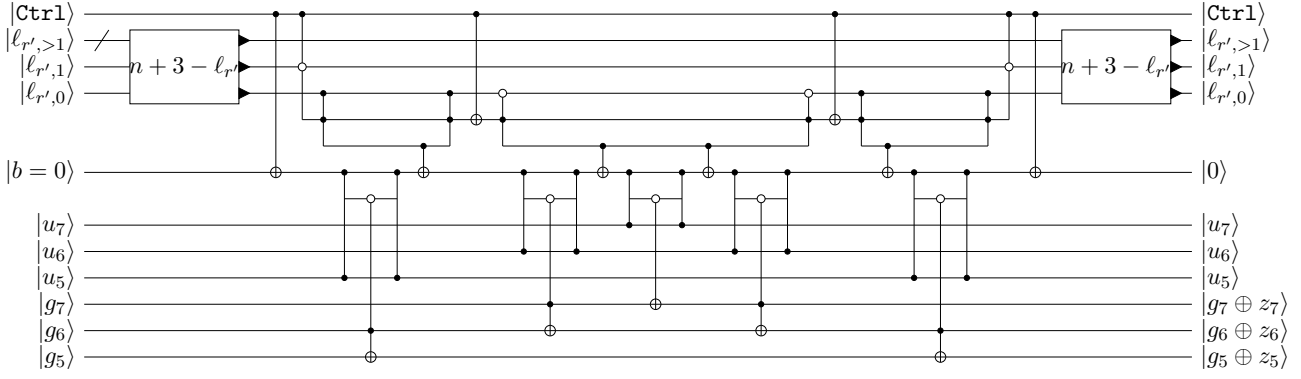}
    }
    \caption{Computing $g_j\gets g_j\oplus z_j$ for $k=5$ and $K=7$.}
    \label{fig:length-z-transform}
\end{figure}

To update $\ell_t$, we invoke this subroutine twice.
First we use the analogous construction with $v_j$ in place of $u_j$ and borrow the qubits $u_j$ as the dirty register $g_j$, which XORs $L_{\mathrm{old}}$ into $\ell_t$.
Then we use the construction described above and borrow the qubits $v_j$ as the dirty register $g_j$, which XORs $L_{\mathrm{new}}$ into $\ell_t$.
The two calls are serial, and the borrowed dirty qubits are restored after each call.

We now give the explicit resource count.
First consider one application of $g_j\gets g_j\oplus z_j$ for all $j\in[k,K]$.
Its gate count has three contributions:
\begin{itemize}
    \item The forward and reverse unary iterations use $2(K-k)+O(1)$ Toffoli gates and
    $4(K-k)+O(1)$ CNOT gates.
    \item The forward and reverse applications of
    $g_j\gets g_j\oplus(\neg a_j)g_{j+1}$ contribute $2(K-k)$ Toffoli gates in total,
    while the update at $j=K$ contributes one CNOT gate.
    \item The temporary bit $a_j=b_j\land u_j$ is computed once in each direction and
    uncomputed immediately after use. Over both directions, this uses
    $2(K-k)+O(1)$ Toffoli gates and $2(K-k)+O(1)$ CNOT gates.
\end{itemize}
Thus one application of $g_j\gets g_j\oplus z_j$ uses
$6(K-k)+O(\log_2n)$ Toffoli gates and $6(K-k)+O(\log_2n)$ CNOT gates.

To XOR either $L_{\mathrm{old}}$ or $L_{\mathrm{new}}$ into $\ell_t$, the transformation
$g_j\gets g_j\oplus z_j$ is first applied between the two sets of controlled constant XORs
and is then applied again to restore the qubits $g_j$. These two applications use
$12(K-k)+O(\log_2n)$ Toffoli gates and $12(K-k)+O(\log_2n)$ CNOT gates. Each set of
controlled constant XORs uses $2(K-k)+O(\log_2 n)$ CNOT gates and no Toffoli gates.
Therefore, XORing one of the two lengths into $\ell_t$ uses
$12(K-k)+O(\log_2n)$ Toffoli gates and $16(K-k)+O(\log_2n)$ CNOT gates.

The complete length update performs this procedure once for $L_{\mathrm{old}}$ and once for
$L_{\mathrm{new}}$. Its total cost is $24(K-k)+O(\log_2n)$ Toffoli gates and
$32(K-k)+O(\log_2n)$ CNOT gates.

The construction uses exactly $K-k+1$ dirty auxiliary qubits, all of which are restored to
their initial states. The same clean auxiliary qubits can be reused throughout the circuit:
one stores $a_j$, one stores $b_j$, and unary iteration requires at most
$\lceil\log_2(K-k)\rceil+1$ additional clean qubits for $K-k\ge 1$. Thus the construction
uses at most $\lceil\log_2(K-k)\rceil+3$ clean auxiliary qubits.

\section{Affine Point Addition with Mid-Circuit Measurement}
\label{sec:measurement_point_addition}

In the previous sections, we constructed a space-efficient modular inversion circuit based on the extended Euclidean algorithm.
The inversion circuit uses two $n$-qubit field registers together with $O(\log_2 n)$ auxiliary qubits.
In this section, we explain how to incorporate this component into an affine point-addition circuit.
By allowing mid-circuit measurements and classical feed-forward operations, the affine point-addition block can be implemented using only three $n$-qubit field registers, plus the $O(\log_2 n)$ auxiliary qubits required by the EEA inversion subroutine.

The overall circuit is shown in Figure~\ref{fig:measurement_point_addition}.
The three main quantum registers are denoted by $X$, $Y$, and $A$, each storing an element of $\mathbb{F}_p$.
The remaining $O(\log_2 n)$ EEA ancilla qubits are not shown explicitly in the figure.

\begin{figure}[ht]
    \centering
    \includegraphics[width=\linewidth]{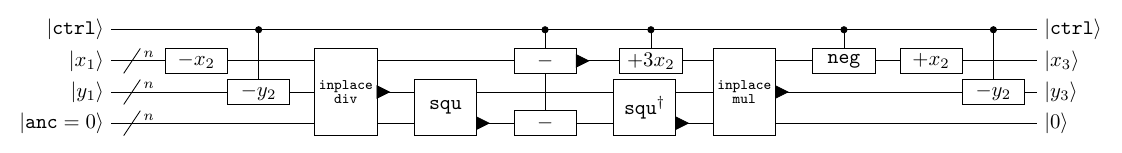}
    \caption{Affine point-addition circuit with mid-circuit measurements using three $n$-qubit field registers.}
    \label{fig:measurement_point_addition}
\end{figure}

The constant additions, subtractions, negations, and squarings in Figure~\ref{fig:measurement_point_addition} are standard in-place modular operations.
The only non-standard components are the in-place division and multiplication subroutines.
They have the following abstract interfaces:
\begin{align}
    \operatorname{IDiv}:\quad
    \ket{x}_X\ket{y}_Y\ket{0}_A
    &\longmapsto
    \ket{x}_X\ket{y/x}_Y\ket{0}_A,
    \qquad x\ne 0,
    \label{eq:idiv_interface}
    \\
    \operatorname{IMul}:\quad
    \ket{x}_X\ket{y}_Y\ket{0}_A
    &\longmapsto
    \ket{x}_X\ket{xy}_Y\ket{0}_A,
    \qquad x\ne 0.
    \label{eq:imul_interface}
\end{align}

It remains to explain how the in-place maps in Equations~\eqref{eq:idiv_interface} and~\eqref{eq:imul_interface} can be implemented with only three field registers.
The constructions are similar, so we describe the in-place division subroutine in detail.
Let $S$ denote the $O(\log_2 n)$ auxiliary register used by the EEA inversion circuit.
The goal is to implement
\begin{align*}
    \ket{x}_X\ket{y}_Y\ket{0}_A\ket{0}_S
    \longmapsto
    \ket{x}_X\ket{y/x}_Y\ket{0}_A\ket{0}_S,
    \qquad x\ne 0.
\end{align*}

First, run the forward EEA inversion circuit on $X$, using $A$ as the large clean workspace and $S$ as the small auxiliary register:
\begin{align*}
    \ket{x}_X\ket{y}_Y\ket{0}_A\ket{0}_S
    \longmapsto
    \ket{x^{-1}}_X\ket{y}_Y\ket{0}_A\ket{\Gamma(x)}_S.
\end{align*}
Here $\ket{\Gamma(x)}_S$ denotes the intermediate EEA state that is kept until the inverse EEA circuit is applied.
Next, use a standard modular multiplication to write the quotient into $A$:
\begin{align*}
    \ket{x^{-1}}_X\ket{y}_Y\ket{0}_A
    \longmapsto
    \ket{x^{-1}}_X\ket{y}_Y\ket{y/x}_A.
\end{align*}
At this point $A$ is no longer clean, and therefore cannot be used as the large workspace for the inverse EEA circuit.
Instead, we recycle the register $Y$ by measurement.
Apply $H^{\otimes n}$ to $Y$ and measure it in the computational basis.
If the classical outcome is $b\in\{0,1\}^n$, the remaining quantum state acquires the phase $(-1)^{b\cdot y}$, where the dot product is over $\mathbb{F}_2$.
The measured register $Y$ is then reset to $\ket{0}$.

Now $Y$ is a clean $n$-qubit register, so it can be used as the large workspace for the inverse EEA circuit.
Applying the inverse EEA circuit gives
\begin{align}
    (-1)^{b\cdot y}
    \ket{x^{-1}}_X\ket{0}_Y\ket{y/x}_A\ket{\Gamma(x)}_S
    \longmapsto
    (-1)^{b\cdot y}
    \ket{x}_X\ket{0}_Y\ket{y/x}_A\ket{0}_S.
    \label{eq:idiv_step_inverse_eea}
\end{align}
The phase in Equation~\eqref{eq:idiv_step_inverse_eea} is removed by recomputing the measured value $y$ from $x$ and $y/x$.
Specifically, apply the modular multiplication
\begin{align}
    \ket{x}_X\ket{0}_Y\ket{y/x}_A
    \longmapsto
    \ket{x}_X\ket{y}_Y\ket{y/x}_A,
    \label{eq:idiv_step_recompute_y}
\end{align}
then apply the classically controlled Pauli correction
\begin{align*}
    Z^b
    =
    Z^{b_0}\otimes Z^{b_1}\otimes\cdots\otimes Z^{b_{n-1}}
\end{align*}
to the register $Y$.
This contributes another factor $(-1)^{b\cdot y}$ and therefore cancels the measurement phase.
Finally, reverse the multiplication in Equation~\eqref{eq:idiv_step_recompute_y}, obtaining
\begin{align*}
    \ket{x}_X\ket{0}_Y\ket{y/x}_A.
\end{align*}
Swapping $Y$ and $A$ completes the in-place division:
\begin{align*}
    \ket{x}_X\ket{y/x}_Y\ket{0}_A.
\end{align*}
The corresponding circuit is shown in Figure~\ref{fig:measurement_idiv}.

\begin{figure}[ht]
    \centering
    \includegraphics[width=1\linewidth]{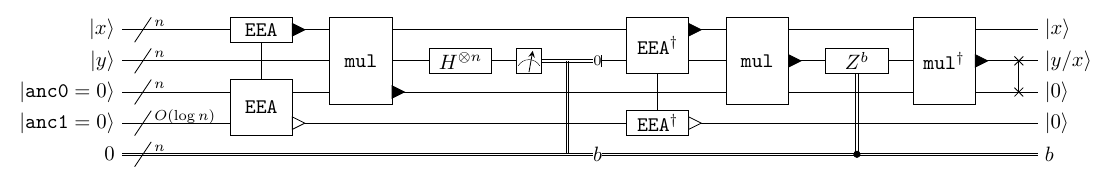}
    \caption{In-place division with mid-circuit measurement and classical feed-forward operations.}
    \label{fig:measurement_idiv}
\end{figure}

The in-place multiplication subroutine is implemented by the same mechanism.
One first computes $xy$ into the clean register $A$, measures and resets the old $Y=y$ register, uses the remaining registers to recompute $y$ for the $Z^b$ phase correction, uncomputes this recomputation, and finally swaps $Y$ with $A$.
Thus both $\operatorname{IDiv}$ and $\operatorname{IMul}$ require only the three field registers $X,Y,A$, together with the $O(\log_2 n)$ EEA auxiliary register $S$.
Therefore the affine point-addition circuit in Figure~\ref{fig:measurement_point_addition} uses $3n+O(\log_2 n)$
logical qubits.

\section{Resource Estimation and the ECDLP}
\label{sec:resource_estimation}

We adopt the resource-estimation methodology introduced by Roetteler et al.~\cite{roetteler2017quantum}.
In Shor’s algorithm for the ECDLP, the dominant quantum cost arises from the sequence of controlled elliptic-curve point additions performed in superposition. 
We consider the point additions in affine coordinates over the finite field $\mathbb{F}_p$, where $p$ is an $n$-bit prime.
At this level, point addition reduces to a fixed sequence of modular additions, subtractions, multiplications, squarings, and inversions.
Among these operations, modular inversion typically dominates both the space complexity and the Toffoli gate count.

\subsection{Space complexity improvements}
\label{subsec:resource_space}

In this work, we focus exclusively on reducing the space complexity required for modular inversion, rather than exploring trade-offs among Toffoli count, circuit depth, and space complexity as previous works \cite{roetteler2017quantum,haner2020improved,babbush2026ECC,schrottenloher2026optimized}.
Specifically, we present an exact, space-efficient, in-place reversible quantum circuit for modular inversion that implements Equation~\eqref{eq:inv_api}. The circuit is based on the register-sharing technique originally proposed by Proos and Zalka~\cite{proos2003shor}.
This approach reduces the number of qubits required for modular inversion to
\[
    2n + 6\floor{\log_2 n} + 19.
\]
Of this total, $2n + 4\floor{\log_2 n} + 15$ qubits correspond to the base circuit shown in Figures~\ref{fig:all_circuit_1} and~\ref{fig:all_circuit_2}, while the remaining $2\floor{\log_2 n} + 4$ qubits are required for the unary iteration described in Section~\ref{subsec:unary_iteration}.

For elliptic curve point addition shown in Figure~\ref{fig:measurement_point_addition}, the overall space complexity is then
\[
    3n + 6\floor{\log_2 n} + 19.
\]
\subsection{Asymptotic gate counts}
\label{subsec:inversion_costs}

\paragraph{Modular inversion.} We first derive the asymptotic Toffoli and CNOT gate counts of the proposed modular inversion circuit. Since the circuit consists of several types of location-controlled operations, we first recall (in Section~\ref{subsec:blocks}) the gate cost incurred by each circuit block at algorithmic step $T$. These per-step costs are summarized in Table~\ref{tab:inv_gate_cost}. We then combine them with the active-window analysis developed in Section~\ref{subsec:actwindow} to obtain the overall asymptotic gate complexity.

\begin{table}[htbp]
\centering
\begin{tabular}{|c|c|c|}
\hline
\textbf{Circuit block on step $T$}                                                          & \textbf{Toffoli cost} & \textbf{CNOT cost}    \\ \hline
\begin{tabular}[c]{@{}c@{}}Location-controlled\\ addition/subtraction on $r$'s\end{tabular} & $11(K_1(T) - k_1(T))$ & $12(K_1(T) - k_1(T))$ \\ \hline
Location-controlled swap                                                                    & $2(K_2(T) - k_2(T))$  & $4(K_2(T) - k_2(T))$  \\ \hline
\begin{tabular}[c]{@{}c@{}}Location-controlled\\ addition/subtraction on $t$'s\end{tabular} & $9(K_3(T) - k_3(T))$  & $10(K_3(T) - k_3(T))$ \\ \hline
\begin{tabular}[c]{@{}c@{}}Length-update circuit for\\ $\ell_t$ (only on $4|T$)\end{tabular}    & $24(K_4(T) - k_4(T))$ & $32(K_4(T) - k_4(T))$ \\ \hline
\begin{tabular}[c]{@{}c@{}}Length-update circuit for\\ $\ell_{r'}$ (only on $4|T$)\end{tabular} & $24(K_5(T) - k_5(T))$ & $32(K_5(T) - k_5(T))$ \\ \hline
\end{tabular}
\caption{Per-step Toffoli and CNOT gate counts of the major circuit blocks in the modular inversion algorithm.}
\label{tab:inv_gate_cost}
\end{table}

\begin{itemize}
\item \textbf{Location-controlled addition/subtraction on $r$'s.}  
These operations cost
\[
    2\sum_{T=1}^{4\ceil{cn}} 11\bigl(K_1(T) - k_1(T)\bigr) + O(n\log_2 n)
    = (44c + 11)n^2 + O(n\log_2 n)
\]
Toffoli gates, and $(48c + 12)n^2 + O(n\log_2 n)$ CNOT gates.

\item \textbf{Location-controlled swap.} 
These operations cost
\[
    \sum_{T=1}^{4\ceil{cn}} 2\bigl(K_2(T) - k_2(T)\bigr) + O(n\log_2 n)
    = (4c + 1)n^2 + O(n\log_2 n)
\]
Toffoli gates, and $(8c + 2)n^2 + O(n\log_2 n)$ CNOT gates.

\item \textbf{Location-controlled addition/subtraction on $t$'s.} 
These operations cost
\[
    2\sum_{T=1}^{4\ceil{cn}} 9\bigl(K_3(T) - k_3(T)\bigr) + O(n\log_2 n)
    = (72c - 36)n^2 + O(n\log_2 n)
\]
Toffoli gates, and $(80c - 40)n^2 + O(n\log_2 n)$ CNOT gates.

\item \textbf{Length-update circuit for updating $\ell_t$.} 
We execute the length-update circuit only at step indices $T$ that are multiples of four. These operations cost
\[
    \sum_{T'=1}^{\ceil{cn}} 24\bigl(K_4(4T') - k_4(4T')\bigr) + O(n\log_2 n)
    = 12c n^2 + O(n\log_2 n).
\]
Toffoli gates, and $16cn^2 + O(n\log_2 n)$ CNOT gates.

\item \textbf{Length-update circuit for updating $\ell_{r'}$.} 
We execute the length-update circuit only at step indices $T$ that are multiples of four. These operations cost
\[
    \sum_{T'=1}^{\ceil{cn}} 24\bigl(K_5(4T') - k_5(4T')\bigr) + O(n\log_2 n)
    = 12cn^2 + O(n\log_2 n).
\]
Toffoli gates, and $(16c - 4) n^2 + O(n\log_2 n)$ CNOT gates.

\end{itemize}

In addition to the circuit blocks analyzed above, the modular inversion circuit (illustrated in Figures \ref{fig:all_circuit_1} and \ref{fig:all_circuit_2}) contains four controlled $\mathrm{Shift}$ circuits. Each $\mathrm{Shift}$ circuit consists of $n+O(1)$ Toffoli gates and $2n+O(1)$ CNOT gates. Therefore, these four $\mathrm{Shift}$ circuits contribute an additional $16cn^2+O(n)$ Toffoli gates and $32cn^2+O(n)$ CNOT gates.

Combining all contributions, the total number of Toffoli gates required by our (one-sided without reversion) modular inversion circuit grows asymptotically as
\[
    (160c - 24)n^2 + O(n\log_2 n) < 228.636 n^2 + O(n\log_2 n).
\]
The asymptotic bound of CNOT gates required is
\[
    (200c - 30)n^2 + O(n\log_2 n) < 285.795 n^2 + O(n\log_2 n).
\]

\paragraph{Point addition.} As discussed in Section~\ref{sec:measurement_point_addition}, we count one affine point addition with mid-circuit measurement as consisting of 1 in-place division, 1 in-place multiplication, 2 modular squarings, and a constant number of modular additions/subtractions/negations. The corresponding circuit is shown schematically in Figure \ref{fig:measurement_point_addition}. 
Among these modular arithmetic operations, the modular additions/subtractions/negations contribute only $O(n\log_2 n)$ Toffoli gates in total and are therefore absorbed into the lower-order term. The in-place division, as well as the in-place multiplication, consists of two EEA inversion circuit, one forward and one inverse, and three modular multiplications as shown in Figure~\ref{fig:measurement_idiv}. 

We use an improved modular multiplication/squaring circuit described in Appendix \ref{app:mul}, which costs $17n^2 + O(n\log_2 n)$ Toffoli gates. Therefore, a single affine point addition with mid-circuit measurement can be implemented by at most
\[
    (4\cdot 228.636 + 6\cdot 17 + 2\cdot 17)n^2 + O(n\log_2 n) < 1051n^2 + O(n\log_2 n)
\]
Toffoli gates.
\subsection{Numerical experiments}
\label{subsec:numerical}
All resource-estimation experiments were conducted on a Linux server equipped with an AMD EPYC 7302 processor (up to 3.0\,GHz). Circuit generation, compilation and optimization were performed using IBM Qiskit. The constructed circuits were transpiled into a gate basis consisting of \texttt{CCX}, \texttt{CX}, and \texttt{X} gates, in order to obtain accurate Toffoli-gate and CNOT-gate counts.

The implementation follows the explicit circuit constructions in
Sections~\ref{sec:circuit} and~\ref{sec:measurement_point_addition}.   Since modular
inversion is the most involved subroutine and a main contributor to the
point-addition cost, we report its implementation separately.

\begin{itemize}
    \item \textbf{Modular-inversion subroutine.}
    We implement the one-step EEA circuit as a collection of Qiskit modules.  These modules include the pre-shift and
    post-shift blocks, the location-controlled arithmetic and swap blocks, the
    phase-update logic, and the length-update circuits.  A full inversion
    instance is generated by following the step schedule of
    Algorithm~\ref{alg:opt}.  Each step instantiates only its active window.
    The step circuits are counted in chunks and are then combined with the
    shared-register EEA wrapper.

    \item \textbf{Point-addition circuit.}
    We implement the point-addition circuit blockwise. Figure~\ref{fig:measurement_point_addition}
specifies the top level affine point addition schedule, including the constant
coordinate updates, the in-place division and multiplication calls, the
square-minus block, and the final corrections.  The point coordinates $x_2$ and
$y_2$ are treated as classically precomputed constants.  The in-place division
and multiplication blocks are assembled according to
Figure~\ref{fig:measurement_idiv}, from the shared-register EEA circuit,
modular multiplication blocks, inverse multiplication blocks, and the explicit
mid-circuit measurements and feed-forward phase corrections shown in
the figure.
\end{itemize}

To obtain the compiled gate counts of these implemented circuits, we use a
layered blockwise counting pipeline.  This is necessary because directly
expanding the full point-addition circuit into a single Qiskit circuit over
elementary operations is impractical.  We first recursively decompose and count
reusable Qiskit subblocks, including constant-coordinate updates, controlled
modular additions and subtractions, modular doubling and halving, and the
shared-register EEA wrapper.  We then assemble larger arithmetic blocks from
these counted components with their exact multiplicities.  For example, a
modular multiplication block contains $n$ controlled modular additions and
$n-1$ modular doublings, while its inverse contains $n$ controlled modular
subtractions and $n-1$ modular halvings.  The squaring block reuses the same
multiplication count and adds the explicit CNOT cost for copy and uncompute.
This yields a complete compiled gate count for the point-addition circuit while
preserving the block multiplicities.

In addition to compiling and counting the reusable subblocks, we apply a local
template-based optimization before each subblock count is stored and reused.
The optimizer performs local cancellation and Qiskit template replacement on
small reversible segments.  To make the template identities directly applicable,
we restrict these segments to the gate library
$\{X,\mathrm{CX},\mathrm{CCX}\}$.  A rewritten segment is accepted only if its gate count does not
increase.

To improve efficiency and preserve the circuit structure, we apply the
optimization segment by segment. We scan each decomposed subblock linearly and
collect consecutive $X$, $\mathrm{CX}$, and $\mathrm{CCX}$ gates
into one reversible segment.  When the scan reaches any other instruction, the
current segment is optimized, the out-of-library instruction is counted
directly, and a new segment is started.  Thus, in circuits with measurements and feed-forward phase corrections, these operations are not
rewritten.  They simply separate the neighboring reversible segments.  If a
segment exceeds the runtime limits, it is counted unchanged.

To validate the implementation, we developed test suites for both the modular-inversion and point-addition circuits. The tests include representative examples over prime fields with $p$ ranging from small values such as $p=3$ to 512-bit primes, and with circuit widths $n$ ranging up to $512$. They additionally verify register layouts, circuit schedules, arithmetic sub-blocks, and compiled block assembly.

The corresponding circuit implementation is included in our open-source codebase.\iffull
\footnote{\label{footnote:repo}GitHub repository: \href{https://github.com/ZeroWang030221/Space-Efficient-Quantum-Algorithm-for-Elliptic-Curve-Discrete-Logarithms-with-Resource-Estimation}{https://github.com/ZeroWang030221/Space-Efficient-Quantum-Algorithm-for-Elliptic-Curve-Discrete-Logarithms-with-Resource-Estimation}}
\else
\footnote{\label{footnote:repo}GitHub repository: \href{https://anonymous.4open.science/r/Space-Efficient-Quantum-Algorithm-for-Elliptic-Curve-Discrete-Logarithms-with-Resource-Estimation-D264}{https://anonymous.4open.science/r/Space-Efficient-Quantum-Algorithm-for-Elliptic-Curve-Discrete-Logarithms-with-Resource-Estimation-D264}}
\fi
For each problem size $n$, we generated the full circuit instance up to the theoretical upper bound on the number of steps derived in this work, transpiled it into the \texttt{CCX}, \texttt{CX}, and \texttt{X} basis, and recorded the resulting gate counts. We report results for $n \in \{64,128,160,192,224,256,384,512\}$. The numerical values are listed in Table~\ref{tab:gate-counts}. 

\begin{table}[htbp]
\centering
\begin{tabular}{c|cc|cc}
\hline
\multirow{2}{*}{$n$}
& \multicolumn{2}{c|}{Forward modular inversion}
& \multicolumn{2}{c}{Point addition} \\
\cline{2-5}
& Toffoli ($\times 10^6$)
& CNOT ($\times 10^6$)
& Toffoli ($\times 10^6$)
& CNOT ($\times 10^6$) \\
\hline
64  & 1.20  & 1.85  & 5.37   & 8.72   \\
128 & 4.24  & 6.60  & 19.20  & 31.72  \\
160 & 6.33  & 9.95  & 28.80  & 48.13  \\
192 & 8.87  & 14.01 & 40.51  & 68.03  \\
224 & 11.86 & 18.77 & 54.26  & 92.21  \\
256 & 15.35 & 24.30 & 70.29  & 118.50 \\
384 & 33.11 & 52.91 & 152.50 & 260.44 \\
512 & 58.15 & 92.94 & 268.26 & 456.81 \\
\hline
\end{tabular}
\caption{Gate counts for modular inversion and point addition circuit over $n$-bit prime fields.}
\label{tab:gate-counts}
\end{table}
\subsection{Resource estimates for Shor's ECDLP circuit}
\label{subsec:component_costs}

For the full ECDLP attack, we do not repeat the controlled point-addition circuit
bit-by-bit for all $2n$ control bits. Instead, we follow the signed-window technique
of~\cite[Section~5.1]{haner2020improved}. Namely, we partition the
$2n$ control bits arising in the double-scalar multiplication into windows of size $w$,
and implement each window by one signed windowed point addition.

In the signed-window construction, one of the $w$ bits is used as a sign bit, while the
remaining $w-1$ bits address a precomputed table of point multiples. As shown schematically
in Fig.~10 of~\cite{haner2020improved}, one such windowed point addition uses the arithmetic
core of an affine Weierstrass point addition. The coordinates of the selected precomputed point
are obtained through three table queries, and its sign is handled by one modular negation. In our space-optimized circuit,
these three queries are implemented by five sequential table look-ups.\footnote{The circuit
first queries $(x_2,y_2)$, then $x_2$, and finally $(x_2,y_2)$. If two free $n$-qubit registers were
available, the two coordinates in each $(x_2,y_2)$ query could be loaded together. Since only
one free $n$-qubit register is available in our space-optimized circuit, each $(x_2,y_2)$ query is split into two sequential
look-ups, giving $2+1+2=5$ look-ups in total.} Since the dominant $n^2$ term in our
Toffoli accounting comes from modular divisions, multiplications and squaring, replacing the
classically known addend by an addend obtained from the look-up table does not change the
asymptotic scaling.

To obtain a clean asymptotic bound, we choose window size $w = 2\log_2 n$. The five sequential table look-ups contribute $5\cdot 2^w$
Toffoli gates to each windowed point addition. The full double-scalar multiplication uses
approximately $2n/w$ windows, so its dominant Toffoli count becomes
\[
    \frac{2n}{w}\left(1051n^2 + 5\cdot 2^w + O(n\log_2n)\right)
    = 1056\frac{n^3}{\log_2n} + O(n^2).
\]

For $n=256$, following the analysis in \cite{babbush2026ECC}, the Toffoli count of the complete Shor ECDLP circuit is estimated as
\begin{equation}
28 \cdot \left(5 \cdot 2^{16} + Q_A\right).
\end{equation}
Here, $Q_A$ denotes the Toffoli count of a single controlled affine point-addition circuit. For a window size of $w=16$, the number of signed-window point additions is $\frac{2n}{w}-4=28$. The five sequential table look-ups described above incur an additional $5\cdot 2^w$ Toffoli gates per window, which becomes $5\cdot 2^{16}$ when $w=16$. Substituting our concrete value of $Q_A(256)$ shows that our algorithm requires $2^{30.88}$ Toffoli gates for the curve secp256k1.

Finally, we note that compression techniques and multi-run tradeoffs \cite{ekeraa2019revisiting} could potentially yield further improvements to the constant factors in our overall resource estimates.

\iffull
\section*{Acknowledgments}
We thank Pierre-Luc Dallaire-Demers for pointing out an error in an earlier version of the EEA step bound proof, which we have replaced with a matrix–continuant proof based on a finite polyhedral antinorm
certificate; see Appendix~\ref{app:step_num}. We thank Andreas Renz for pointing out a numerical inconsistency in an earlier version regarding the concrete Toffoli estimate for secp256k1. This work was supported by the Innovation Program for Quantum Science and Technology (Grant Number 2024ZD0300500). ZY, ZW, YS and TL are supported by the National Natural Science Foundation of China (Grant Numbers 62372006 and 92365117).
\fi

\ifllncs
    \bibliographystyle{splncs04}
\fi

\newcommand{\etalchar}[1]{$^{#1}$}
\newcommand{\arxiv}[1]{arXiv:\href{https://arxiv.org/abs/#1}{\ttfamily{#1}}\?}\newcommand{\arXiv}[1]{arXiv:\href{https://arxiv.org/abs/#1}{\ttfamily{#1}}\?}\def\?#1{\if.#1{}\else#1\fi}

\appendix
\section{Proof Details}

\subsection{Bounds on the total number of steps}
\label{app:step_num}

Suppose that after $m$ iterations of the Extended Euclidean Algorithm, we obtain $r_{m} = 1$ and $r_{m+1} = 0$, with intermediate quotients $q_1, q_2, \ldots, q_m$, where $q_i = \lfloor r_{i-1} / r_i \rfloor$ for $i = 1, 2, \ldots, m$. Then the total number of steps required in our four-phase algorithm can be expressed as
\begin{equation}
    N = 4\sum_{i=1}^{m} \left(\lfloor \log_2 q_i \rfloor + 1\right).
\label{eq:step}
\end{equation}
In this appendix, we prove that $4n \le N \le 4\lceil cn \rceil$, where $c=3/\log_2(2+\sqrt3)$.

\paragraph{The lower bound.}
The lower bound of $N$ can be derived directly. Since $r_{i-1} = q_i r_i + r_{i+1} < (q_i + 1)r_i$, it follows that
\[
    N \ge 4\sum_{i=1}^{m}\log_2(q_i + 1) > 4\sum_{i=1}^{m}\log_2\frac{r_{i-1}}{r_i} = 4\log_2 p \ge 4(n - 1).
\]
Because $N$ must be a multiple of $4$, we conclude that $N \ge 4n$.

\paragraph{Asymptotic sharpness of the upper bound.} Before we moving on to the proof of the upper bound, we first express why the constant $c=3/\log_2(2 + \sqrt 3)$ in the upper bound $N \le 4\lceil cn \rceil$ is asymptotically optimal. We consider a different roadmap from deriving the lower bound: instead of computing the upper bound of $N$, we fix the total number of steps $N$ and then compute the minimum possible value of $p$. 

For integer \(q \ge 1\), define the EEA transformer matrix $M(q) :=
\begin{pmatrix}
    q & 1 \\
    1 & 0
\end{pmatrix}.$
The EEA iteration can therefore be written as $(r_{i-1}, r_i)^T = M(q_i)(r_{i}, r_{i+1})^T.$ Iterating this identity and using \((r_m, r_{m+1})=(1, 0)\) gives
\begin{equation}
    \begin{pmatrix}
        p \\ x
    \end{pmatrix} = M(q_1) \cdots M(q_m)
    \begin{pmatrix}
        1 \\ 0
    \end{pmatrix} = M(q_1) \cdots M(q_{m-1})
    \begin{pmatrix}
        q_{m} \\ 1
    \end{pmatrix}.
    \label{eq:continuant}
\end{equation}
In particular, 
\begin{equation}
    p = K(q_1, \ldots, q_m) := e_1^T M(q_1) \ldots M(q_m) e_1.
    \label{eq:p_continuant}
\end{equation}

Given equation~\eqref{eq:p_continuant}, writing $d(q) := \lfloor\log_2 q \rfloor + 1$, our problem is to find a slowest growth pattern for products of the matrices $M(q)$ with a fixed weight sum $N/4 = \sum_{i = 1}^m d(q_i)$, where a factor $M(q)$ has weight $d(q)$. 
Experiments for small total quotient bit-lengths suggest that, with
\(\sum d(q)\) fixed, the quotient sequence $(1, 2)$ has the minimum averaged growing rate, i.e. the nearly alternating quotient sequence \((2, 1, 2, 1, \ldots, 2, 1, 2, 2)\) asymptotically minimizes the resulting value of \(p\).

Since \(p\) is the upper-left and largest entry of \(P := M(q_1) \cdots M(q_m)\), we have $p\le \rho(P)\le \lVert P\rVert_\infty\le 2p$, here $\rho(P)$ denotes the spectral radius, i.e. the maximum absolute value of eigenvalues, and $\lVert P\rVert_\infty$ denotes the maximum row sum. So \(p\) has the same order of magnitude with the spectral radius of the matrix product $P$. Therefore, the repeated block $(1, 2)$ in the quotient sequence suggests the asymptotic growth rate: it has total weight $3$, while $M(1) M(2) = \begin{pmatrix}
    3 & 1 \\
    2 & 1
\end{pmatrix}$ has spectral radius $2 + \sqrt 3$. This motivates the normalization $\rho := 2 + \sqrt3$ and $\lambda := \rho^{1/3}$, so that one unit of weight corresponds to a factor $\lambda$. 

We then give the formal argument on the asymptotic sharpness of the constant $c$. For $k \ge 0$, consider a quotient sequence 
\[
    Q_k = \bigl(2, \underbrace{1, 2, 1, 2, \ldots, 1, 2}_{k\text{ copies of }(1, 2)}, 2\bigr),
\]
its corresponding step number is $N_k = 12k + 16$. Let $M := M(1)M(2) = \begin{pmatrix}
    3 & 1 \\
    2 & 1
\end{pmatrix}$ be a two-iteration EEA transformation matrix with maximal eigenvalue $\rho = 2 + \sqrt3$. From basic linear algebra argument, the value of $p = e_1^T M(2) M^k M(2) e_1 = (2, 1) M^k (2, 1)^T$ will asymptotically be bounded by $\Theta(\rho^k)$. Consequently, $\log_2 p_k = k\log_2 \rho + O(1)$, and
\[
    \lim_{k\to\infty} \frac{N_k}{\log_2p_k}
    = \frac{12}{\log_2\rho} = 4c.
\]
Therefore, the coefficient $c$ in the upper bound for $N \le 4\lceil cn\rceil$ cannot be decreased for general admissible quotient sequences.

In the rest of this section, we will discuss the roadmap for proving this sharp upper bound, and then give the formal proof.

\paragraph{Roadmap for the proof of the upper bound.} We consider the invariant cone
\[
    \mathcal C := \left\{
    \begin{pmatrix}
        a \\ b
    \end{pmatrix}:
    a \ge b \ge 0 \right\},
\]
such that every EEA transformer matrix \(M(q)\) preserves \(\mathcal C\), since $M(q)(a, b)^T = (qa + b, a)^T \in\mathcal C$ for all $q \ge 1, a \ge b \ge 0$.

To derive the minimum possible value of $p$, we use the toolkit of a finite polyhedral antinorm certificate on the invariant cone $\mathcal{C}$, following the extremal-antinorm method for lower spectral radius from~\cite{guglielmi2013exact}. 
An antinorm on an invariant cone is a non-zero, non-negative, concave, and positively homogeneous function. It is positive everywhere inside the cone, except at the origin. In particular, the minimum of finitely many positive linear functions, $F(v) = \min_{u\in U} uv$ for some (positive-coordinated) vector set $U$, is a positive polyhedral antinorm.

Set $\rho = 2 + \sqrt{3}, \lambda = \rho^{1/3} = (2 + \sqrt{3})^{1/3}$ as discussed in the asymptotic sharpness of the upper bound. We want to find an antinorm $F = \min_{u\in U} uv$ with the property that the inequality $F(M(q)v) \ge \lambda^{d(q)}F(v)$ holds for every $q \ge 1$ and every $v$ in the invariant cone $\mathcal{C}$.
Iterating this inequality converts the total weight $N/4$ directly into the lower growth factor $\lambda^{N/4}$, which lower-bounds the possible value of $p$. 

For $d \ge 0$, let $B_d := \lambda^{-(d+1)} M(2^d)$ be the normalized transformation associated with the smallest quotient $q = 2^{d-1}$ of bit-length $(d+1)$. 
To achieve the inequality $F(M(q)v) \ge \lambda^{d(q)}F(v)$ for $d \le 4$, we need to construct a finite set \(U = \{u_0, \ldots, u_r\}\) such that 
\begin{equation}
    u_jB_d v\ge \min_{u\in U} uv = F(v),
    \text{  for all } u_j\in U, d\in \{0, 1, 2, 3\}\text{ and } v\in\mathcal C.
    \label{eq:antinorm}
\end{equation}
Once the condition~\eqref{eq:antinorm} is satisfied, we have $F(B_dv) \ge F(v)\text{, therefore }F(M(2^{d})v)\ge \lambda^{d+1}F(v)$. 

The first two vectors of $U$ encode the two-step return generated by the repeated block $(1, 2)$. Namely, we choose $u_0$ to be the left Perron eigenvector of $M = M(1)M(2)$. 
The $(1, 2)$ cycle determines $u_1 = \lambda^{-1}u_0M(1)$ and gives $u_0 = \lambda^{-2}u_1M(2)$. Thus, \(u_0\) and \(u_1\) are precisely the two supporting directions visited during one period of the optimal quotient sequence. 

However, although the two vectors explain the candidate periodic pattern, they do not yet give a certificate for arbitrary quotients. For a current row set $U$, we say that a transformed row $w = uB_d$ is already controlled by the antinorm $F_U(v) := \min_{u'\in U} u'v$ if there is
some $h\in\operatorname{conv}(U)$ such that $wv\ge hv$ for every $v\in\mathcal C$. Indeed, this implies $wv\ge hv\ge\min_{u'\in U}u'v = F_U(v)$. 

Starting from $U = \{u_0, u_1\}$, the transformation $B_2$, corresponding to the quotient $q = 4$, produces $u_2 := u_1B_2 = \lambda^{-3}u_1M(4)$. This row is not controlled by the two-row antinorm as $u_2(1, 1)^T
< \min_{0\le j\le 1}u_j(1,1)^T$, thus must be added to $U$. Applying $B_0$ to the new row similarly gives $u_3 := u_2B_0 = \lambda^{-1}u_2M(1)$, and $u_3(1,0)^T < \min_{0\le j\le 2}u_j(1,0)^T$, so a fourth row is required for $U$. For the resulting set $U = \{u_0, u_1, u_2, u_3\}$, we can verify that the row $u_jB_s$ dominates on $\mathcal C$ a row in $\operatorname{conv}(U)$, for every $j, s\in\{0,1,2,3\}$. Hence all normalized transformations $B_s$ associated with bit-lengths at most four are controlled, and the finite row-generation procedure stops.

For the case $d \ge 5$, the inequality $F(M(2^{d})v)\ge \lambda^{d+1}F(v)$ is not handled by this finite closure table, but instead covered by the coarser estimate $F(M(q)v) \ge (q/\lambda) F(v) \ge \lambda^{d+1}F(v)$, because $q\ge 2^{d-1}\ge\lambda^{d+1}$ holds for all $d\ge 5$.

\paragraph{Formal proof of the upper bound.} Set $\rho = 2 + \sqrt{3}, \lambda = \rho^{1/3} = (2 + \sqrt{3})^{1/3}$ and $\mathcal{C} = \{(a, b)^T: a\ge b\ge 0\}$. Define the following four positive row vectors:
\[
    u_0:=(1+\sqrt3,1),
    \qquad
    u_1:=\lambda^{-1}u_0M(1),
    \qquad
    u_2:=\lambda^{-3}u_1M(4),
    \qquad
    u_3:=\lambda^{-1}u_2M(1).
\]
Equivalently, using \(\lambda^3-1 = 1+\sqrt3\),
\[
    u_0=(\lambda^3-1,1),
    \quad
    u_1=\left(\lambda^2,\,
    \frac{\lambda^3-1}{\lambda}\right),
    \quad
    u_2=\left(
    \frac5\lambda-\frac1{\lambda^4},\,
    \frac1\lambda\right),
    \quad
    u_3=\left(
    \frac{6\lambda^3-1}{\lambda^5},\,
    \frac{5\lambda^3-1}{\lambda^5}\right).
\]
For \(v\in\mathcal C\), define
\[
    F(v):=\min_{0\le j\le3}u_jv.
\]
Since the coordinates of the four vectors are positive, it is easy to prove that the function \(F\) is positive on \(\mathcal C\setminus\{0\}\), homogeneous, concave, and coordinatewise nondecreasing. For the function $F$, we have the following lemma.

\begin{lemma}
    \label{lem:weighted-matrix-growth}
    Let \(q\ge 1\) be a positive integer, $d = \lfloor\log_2 q\rfloor + 1$. For any vector \(v \in\mathcal C\), we have
    \begin{equation}
        F(M(q)v)\ge \lambda^d F(v).
    \label{eq:lemma A.1.1}
    \end{equation}
    Moreover, if \(d\ge 2\), then
    \begin{equation}
        F\left((q, 1)^T\right) \ge \lambda^{d-2} F\left((2, 1)^T\right).
    \label{eq:lemma A.1.2}
    \end{equation}
\end{lemma}

The overall upper bound of EEA step number follows directly from Lemma~\ref{lem:weighted-matrix-growth}. More precisely, we have that
\begin{alignat*}{2}
    \frac{p}{2} \cdot F \left((2, 1)^T\right) 
    & = F \left((p, p / 2)^T\right) \ge F \left((p, x)^T\right) && \qquad (\text{By Algorithm~\ref{alg:full_EEA}, }x < p/2) \\
    & = F \left(M(q_1)\cdots M(q_{m-1})
    \begin{pmatrix}
        q_m \\ 1
    \end{pmatrix}
    \right) && \qquad (\text{By Equation~\ref{eq:continuant}}) \\
    & \ge \lambda^{N / 4 - d_m} F\left((q_m, 1)^T \right) && \qquad (\text{By Inequality~\ref{eq:lemma A.1.1}}) \\
    & \ge \lambda^{N / 4 - 2} F\left((2, 1)^T \right). && \qquad (\text{By Inequality~\ref{eq:lemma A.1.2}, where }r_{m-1} \ge 2 \text{ implies } d_m\ge 2) \\
\end{alignat*}
This implies that
\[
    N/4 \le 2 + \frac{\log_2 p - 1}{\log_2\lambda} 
    < 2 + \frac{n - 1}{\log_2\lambda} < 1 + \frac{n}{\log_2\lambda}.
\]
As $N/4$ is an integer, we can conclude that $N \le 4\lceil{cn}\rceil$, where $c = 1/\log_2\lambda = 3/\log_2(2 + \sqrt{3})$. 

What left is to prove Lemma~\ref{lem:weighted-matrix-growth}.

\begin{proof}[Proof of Lemma~\ref{lem:weighted-matrix-growth}]
For row vectors \(a,b\), write $a\succeq_{\mathcal C}b$ if
$(a-b)(1,0)^T
\ge0$ and $(a-b)(1,1)^T
\ge0$. Since each element in $\mathcal C$ can be represented as a linear combination of $(1,0)^T$ and $(1,1)^T$ with nonnegative coefficients, the relation \(a\succeq_{\mathcal C}b\) implies $av\ge bv$ for every $v\in\mathcal C$.

We first prove \eqref{eq:lemma A.1.1}. Define the matrix \(G = (G_{j,s})\) for \(j=0,\dots,3\) and \(s=0,\dots,3\) by
\[
G_{j,s} := \lambda^{-s-1}u_jM(2^{s}).
\] and the matrix \(H = (H_{j,s})\) with entries
\[
H = 
\begin{pmatrix}
u_1 & \frac45 u_0+\frac15 u_1 & u_0 & u_0 \\[2mm]
u_1 & u_0 & u_2 & u_0 \\
u_3 & u_0 & u_0 & u_0 \\
u_0 & u_0 & u_2 & u_0
\end{pmatrix}.
\]
It is trivial to verify that
\[
G_{j,s} \succeq_{\mathcal C} H_{j,s}\quad \text{and} \quad H_{j,s}v\ge F(v), \quad\text{for all } j,s\in\left\{0,...,3\right\}.
\]
Consequently,
\[
F(M(2^{d})v)
=
\lambda^{d+1}\min_{0\le j\le3}G_{j,d}v
\ge \lambda^{d+1}F(v),\quad \text{for } 0\le d\le3.
\]

Since \(F\) is coordinatewise nondecreasing, $F(M(q)v)\ge F(M(2^{d-1})v) \ge \lambda^dF(v).$ Thus we prove \eqref{eq:lemma A.1.1} for $1\le d \le 4$.

It remains to consider \(d\ge5\).  Write $v=(a,b)^T, a\ge b\ge0.$ A direct check of the four defining rows gives
$F((1,0)^T)=u_3(1,0)^T,F((1,1)^T)=u_2(1,1)^T=\lambda F((1,0)^T).$
By monotonicity and homogeneity,
\begin{equation}
F(M(q)v)
=
F
\begin{pmatrix}
qa+b\\a
\end{pmatrix}
\ge
F
\begin{pmatrix}
qa\\0
\end{pmatrix}
=
qaF
\begin{pmatrix}
1\\0
\end{pmatrix}
=\frac {qa}\lambda F
\begin{pmatrix}
1\\1
\end{pmatrix}
=\frac q\lambda F
\begin{pmatrix}
a\\a
\end{pmatrix}
\ge \frac q\lambda F(v).
\label{eq:lemma A.1.d_ge_5}
\end{equation}

For \(d=5\), $2^4=16 >\lambda^6=(2+\sqrt3)^2=7+4\sqrt3,$ and the inequality propagates to all larger \(d\), since
\(2>\lambda\).  Hence $q/\lambda \ge 2^{d-1}/\lambda \ge \lambda^d$. Together with \eqref{eq:lemma A.1.d_ge_5} this proves \eqref{eq:lemma A.1.1} for every \(d\ge5\).

We next prove \eqref{eq:lemma A.1.2}. Since \(q\ge2^{d-1}\),
it suffices by monotonicity to consider \(q=2^{d-1}\):
\begin{itemize}
    \item For \(d=2\), the claim is immediate.
    \item For \(d=3\), direct
substitution gives $F((4,1)^T)\ge\lambda F((2,1)^T)$. 
    \item For \(d\ge4\), the same four-row calculation gives $8F(1,0)^T\ge\lambda^2F(2,1)^T.$ Therefore,
    \[
    F
\begin{pmatrix}
2^{d-1}\\1
\end{pmatrix}\ge
2^{d-1}
F
\begin{pmatrix}
1\\0
\end{pmatrix}=
2^{d-4}
\left(
8F
\begin{pmatrix}
1\\0
\end{pmatrix}\right)\ge
2^{d-4}\lambda^2
F
\begin{pmatrix}
2\\1
\end{pmatrix}\ge
\lambda^{d-2}
F
\begin{pmatrix}
2\\1
\end{pmatrix},
    \]
where the last inequality uses \(2>\lambda\).
\end{itemize}
 
This completes the proof.
\end{proof}
\subsection{Step-dependent active windows}
\label{app:actwindow}

Fix an input $x$, let $(q_1, \ldots, q_{k-1})$ denote the quotient sequence produced by EEA, and define $b_i := \floor{\log_2 q_i}$ for each iteration $i$.
Under the four-phase schedule, each EEA iteration $i$ is expanded into $4(b_i + 1)$ steps. The total number of steps is $N = 4\sum_{i=1}^{k-1} (b_i + 1) \le N_{\max} := 4\ceil{cn}$, where $\lambda :=(2+\sqrt{3})^{1/3}$ and $c:=1/\log_2\lambda
=3/\log_2(2+\sqrt{3})$ (see Appendix \ref{app:step_num}).

For a global step index $T \in \{1, 2, \ldots, N_{\max}\}$, let $j = j(x, T)$ denote the index of the EEA iteration containing step $T$, and let $u = u(x, T)$ denote the position of $T$ within that iteration. More precisely, the pair $(j(x,T), u(x,T))$ is uniquely determined by the conditions
\[
    \sum_{i=1}^{j-1} 4(b_i + 1) < T \le \sum_{i=1}^{j} 4(b_i + 1),
    \quad
    u = T - \sum_{i=1}^{j-1} 4(b_i + 1).
\]
For convenience, write
\[
    N_j:=\sum_{i=1}^{j-1}4(b_i+1)=T-u.
\]
The bit-length of the current quotient $\ell_q(x, T)$ depends deterministically on indices $(j, u)$:
\begin{equation}\label{eqn:phase_decompose}
\begin{split}
    \ell_q(x, T) & = \begin{cases}
        0 & u(x, T)\in (0, b_j + 1] \cup (3(b_j + 1), 4(b_j + 1)],\\
        u(x, T) - (b_j + 1) & u(x, T)\in (b_j + 1, 2(b_j + 1)],\\
        3(b_j + 1) - u & u(x, T)\in (2(b_j + 1), 3(b_j + 1)].
    \end{cases}
\end{split}
\end{equation}

We next state a lemma that characterizes the minimal possible value of $t_j$, conditioned on the total number of steps $N_j$ taken before reaching the $j$-th iteration of EEA. Let $F$ be the function constructed in Appendix~\ref{app:step_num}, and define
\[
\eta
:=\frac{F\left((2,1)^T\right)}
{\lambda^2F\left((1,1)^T\right)}
=\frac{4\sqrt{3}-3}{13}\lambda^2,
\qquad
\delta:=-\log_\lambda\eta
=0.726208\ldots .
\]
\begin{lemma}\label{log2t_j>phi*N/4}
    Suppose that $N_j$ is fixed. Then, over all possible EEA iteration indices $j$ and all admissible quotient sequences $(q_1, \ldots, q_{j-1})$, we have $t_j\ge \eta\lambda^{N_j/4}
    =\lambda^{N_j/4-\delta}.$ Consequently,
    \[
    \floor{\log_2t_j}+1>
    \frac{N_j/4-\delta}{c}
    =\frac{N_j-4\delta}{4c},
    \qquad
    \floor{\log_2t_j}+1\ge
    \ceil{\frac{N_j-4\delta}{4c}}.
    \]
\end{lemma}

\begin{proof}
    By Algorithm~\ref{alg:full_EEA}, the EEA input $x$ is chosen to be smaller than $p/2$, and hence $q_1\ge2$. Lemma~\ref{lem:weighted-matrix-growth} established in Appendix~\ref{app:step_num} gives $F(M(q)v)\ge\lambda^{\floor{\log_2q}+1}F(v)$ and $F\left((q, 1)^T\right) \ge \lambda^{d-2} F\left((2, 1)^T\right)$ whenever $q\ge2$ has bit-length $d$. Given the EEA iteration rule $t_{i+1}=t_{i-1}+q_it_i$, we can write $(t_{i+1},t_i)^T=M(q_i)(t_i,t_{i-1})^T.$ Therefore, for $j\ge2$,
\begin{align*}
F\left((t_j,t_{j-1})^T\right) = F\left(
M(q_{j-1})\cdots M(q_2)\binom{q_1}{1}
\right)\ge
\lambda^{N_j/4-(b_1+1)}
F\left((q_1,1)^T\right)
\ge
\lambda^{N_j/4-2}
F\left((2,1)^T\right).
\end{align*}
Since $0\le t_{j-1}\le t_j$ and $F$ is coordinatewise nondecreasing and positively homogeneous,
\[
F\left((t_j,t_{j-1})^T\right)
\le
F\left((t_j,t_{j})^T\right)
=t_jF\left((1,1)^T\right).
\]
Combining the two inequalities gives
$t_j\ge\eta\lambda^{N_j/4}$. The case $j=1$ follows from $t_1=1\ge\eta$. 
\end{proof}

We next collect three elementary estimates used repeatedly below. First,
$t_jq_j\le t_{j+1}<p<2^n$ implies
\begin{equation}\label{eqn:quotient-length-bound}
    \ell_t+b_j\le n,
    \qquad b_j+1\le n+1-\ell_t.
\end{equation}
Second, iterating
$t_{i+1}=t_{i-1}+q_it_i<(q_i+1)t_i\le2^{b_i+1}t_i$ gives
\begin{equation}\label{eqn:t-length-upper}
    t_j\le\prod_{i=1}^{j-1}2^{b_i+1}=2^{N_j/4},
    \qquad \ell_t\le N_j/4+1.
\end{equation}
Finally, iterating
$r_{i-1}=r_{i+1}+q_ir_i<(q_i+1)r_i\le2^{b_i+1}r_i$ gives
\[
    \frac{p}{r_j}<\prod_{i=1}^{j}2^{b_i+1}
    =2^{N_j/4+b_j+1}.
\]
Since $p\ge2^{n-1}$ and $N_j/4+b_j+1$ is an integer, the bit length of $r'=r_j$ satisfies
\begin{equation}\label{eqn:rprime-length-lower}
    \ell_{r'}\ge n-N_j/4-(b_j+1).
\end{equation}

We now prove the bounds stated in Section~\ref{subsec:actwindow}.

\paragraph{Location-controlled addition/subtraction on $r$'s.}
These operations are activated only in Phases~1 and~2, i.e., when
$0<u\le2(b_j+1)$. We use Equation~\eqref{eqn:quotient-length-bound} and consider two cases.
\begin{itemize}
    \item \textbf{Case 1.} Suppose that $u \le b_j + 1$. In this case, we have $\ell_q = 0$. By Lemma \ref{log2t_j>phi*N/4}, it follows that
    \[
    \ell_t>
    \frac{T-u-4\delta}{4c}
    \ge
    \frac{T-(n+1-\ell_t)-4\delta}{4c}.
    \]
    Hence $\ell_t+\ell_q=\ell_t>
    \dfrac{T-n-1-4\delta}{4c-1}.$

    \item \textbf{Case 2.} Suppose that $b_j + 1 < u \le 2(b_j + 1)$. Then $\ell_q = u - (b_j + 1)$, and we obtain that
    \[
        \ell_t+\ell_q>
        \frac{T-u-4\delta}{4c}
        +u-(b_j+1).
    \]
    The right-hand side is minimized when $u = b_j + 1$, in which case the expression reduces to that of Case 1.
\end{itemize}
Combining the two cases, we conclude that $\ell_t + \ell_q + 2 \ge \max\left\{\ceil{\dfrac{T-n-1-4\delta}{4c - 1}}, 1\right\} + 2$.

\paragraph{Location-controlled swaps.}
For the lower bound, a location-controlled swap is only activated in Phase 2 and 3, i.e. $b_j + 1 < u\le 3(b_j + 1)$, we then analyze the claimed lower bound by considering two cases.
\begin{itemize}
    \item \textbf{Case 1.} Suppose that $b_j + 1 < u \le 2(b_j + 1)$. In this case, we also have $\ell_t + \ell_q \ge \dfrac{T-n-1-4\delta}{4c - 1}$.

    \item \textbf{Case 2.} Suppose that $2(b_j + 1) < u \le 3(b_j + 1)$. In this case, we have $\ell_q = 3(b_j + 1) - u$. By Lemma \ref{log2t_j>phi*N/4}, it follows that
    \[
        \ell_t \ge \frac{T - u   -4\delta}{4c}
        \ge \frac{T-3(n+1-\ell_t)-4\delta}{4c},
    \]
    hence $\ell_t + \ell_q \ge \ell_t \ge \dfrac{T - 3(n + 1) - 4\delta}{4c - 3}$.
\end{itemize}
The Case~1 bound is no smaller than the Case~2 bound throughout the active range. Indeed, their difference equals $\dfrac{2(4cn+4c+4\delta-T)}
{(4c-1)(4c-3)}>0,$ where we used $T\le N_{\max}<4cn+4$. Combining the two cases, we conclude that $\ell_t + \ell_q + 1 \ge \max\left\{\ceil{\dfrac{T - 3(n + 1) - 4\delta}{4c - 3}}, 1\right\} + 1$.

For the upper bound, Equations~\eqref{eqn:t-length-upper} and~\eqref{eqn:phase_decompose} give
$\ell_t\le(T-u)/4+1$ and $\ell_q\le u/2$. Hence,
\[
    \ell_t + \ell_q + 1 \le \left\lfloor (T - u)/4 + u/2 + 2 \right\rfloor
    = \left\lfloor (T + u)/4 \right\rfloor + 2
    \le \left\lfloor T/2 \right\rfloor + 2.
\]
Finally, since the \texttt{Work1} register stores a value $r \ge 1$, we conclude that $\ell_t + \ell_q + 1 \le n + 2$.

\paragraph{Location-controlled addition/subtraction on $t$'s.}
In Phase~3, Equation~\eqref{eqn:t-length-upper} directly shows that the right boundary $\ell_t+1$ is at most $\floor{T/4}+2$.

In Phase~4, where $3(b_j+1)<u\le4(b_j+1)$, the right boundary is
$n+3-\ell_{r'}-\ell_s$. Immediately before the subsequent right shift of \texttt{Work2},
\[
    \ell_s=4(b_j+1)-u+1.
\]
Combining this identity with Equation~\eqref{eqn:rprime-length-lower} yields
\begin{align*}
    n+3-\ell_{r'}-\ell_s
    &\le N_j/4+u-3(b_j+1)+2 \\
    &\le N_j/4+u/4+2
     =T/4+2,
\end{align*}
where the second inequality follows from $u\le4(b_j+1)$. Since the right boundary is an integer, it is at most $\floor{T/4}+2$.

Finally, in Phase~3 the constraint $t\le p<2^n$ gives $\ell_t+1\le n+1$. In Phase~4, both $\ell_{r'}\ge1$ and $\ell_s\ge1$, so $n+3-\ell_{r'}-\ell_s\le n+1$. Thus, over all activated location-controlled additions/subtractions on $t$'s at step $T$, the common right boundary obeys
\[
    \max\left\{\ell_t+1,\,n+3-\ell_{r'}-\ell_s\right\}
    \le K_3(T):=\min\left\{\floor{T/4}+2,\,n+1\right\}.
\]

\paragraph{Length-update circuit.}
For the lower bound $k_4(T)$, the length-update circuit is activated only at the end of an EEA iteration, namely when $u = 4(b_j + 1)$. By Lemma \ref{log2t_j>phi*N/4}, we have
\[
    \ell_t \ge \frac{T - u - 4\delta}{4c} \ge \frac{T - 4(n + 1 - \ell_t) - 4\delta}{4c},
\]
which implies $\ell_t \ge \dfrac{T - 4(n + 1) - 4\delta}{4c - 4}$.

For the upper bound $K_4(T)$, Equation~\eqref{eqn:rprime-length-lower} and
$u=4(b_j+1)$ give $\ell_{r'}\ge n-T/4$. Therefore,
\[
    n+3-\ell_{r'}\le T/4+3.
\]
Moreover, whenever this length-update circuit is activated, the current
$r'=r_j$ is nonzero, so $\ell_{r'}\ge1$ and
$n+3-\ell_{r'}\le n+2$. Hence
\[
    n+3-\ell_{r'}
    \le K_4(T):=\min\left\{T/4+3,\,n+2\right\}.
\]

For the lower bound $k_5(T)$, Lemma \ref{log2t_j>phi*N/4} directly yields
$\ell_t^* \ge \dfrac{T-4\delta}{4c}$. For the upper bound $K_5(T)$, it is clear that $n + 4 - \ell_{r'}^* \le n + 3$.

\section{Improved Modular Multiplication Circuit}
\label{app:mul}

Our modular multiplier used in in-place division/multiplication implements
\[
    \mathrm{Mul}_{p}:
    \ket{x}\ket{y}\ket{0}
    \longmapsto
    \ket{x}\ket{y}\ket{xy \bmod p},
\]
based on the usual Horner double-and-add scan as the low-space multiplier of~\cite{roetteler2017quantum}.  The circuit contains \(n\) controlled modular additions and \(n-1\) modular doublings.
\begin{itemize}
    \item \textbf{Controlled modular additions.}
    The controlled modular addition
\[
    \ket{c}\ket{A}\ket{B}\ket{0}_{f}
    \longmapsto
    \ket{c}\ket{A}\ket{B+cA\bmod p}\ket{0}_{f},
\]
is implemented by first doing the controlled addition
\(B\leftarrow B+cA \pmod {2^n}\), while XORing the overflow carry into a
flag \(f\). We use a specialized controlled
quantum--quantum ripple-carry adder for this step, rather than controlling every gate of an ordinary adder from \cite{cuccaro2004new}.  
Let \(\gamma_i\) denote the carry into bit \(i\), with \(\gamma_0=0\), and $\gamma_{i+1}=\operatorname{maj}(A_i,B_i,\gamma_i)$.

During the forward sweep, the circuit computes the carry chain. At bit \(i\),
the addend and accumulator wires are temporarily transformed to $A_i\oplus \gamma_i$ and $B_i\oplus \gamma_i$, while the carry wire is advanced to \(\gamma_{i+1}\). In the reverse sweep, we
first undo the carry Toffoli, so the carry wire again contains \(\gamma_i\).  A
CNOT from this carry wire to the accumulator wire restores $B_i\oplus \gamma_i \longmapsto B_i$. At this point the addend wire still contains \(A_i\oplus \gamma_i\).  We then
apply the controlled sum-writing Toffoli
\[
    \operatorname{CCX}\bigl(c,\; A_i\oplus \gamma_i;\; B_i\bigr).
\]
Thus, if \(c=0\), the accumulator bit remains \(B_i\); if \(c=1\), it becomes $B_i\oplus A_i\oplus \gamma_i$, which is exactly the \(i\)-th sum bit.  Finally, a CNOT from the carry wire to
the addend wire restores $A_i\oplus \gamma_i \longmapsto A_i$.
Therefore, the controlled writing of the \(n\) sum bits contributes precisely
\(n\) additional Toffoli gates. In addition, after the forward sweep the final carry \(\gamma_n\) is copied into the overflow flag only when \(c=1\), using one more CCX gate,
\[
    f \longleftarrow f\oplus c\gamma_n.
\]
Hence this controlled variant of ripple-carry adder costs $3n+O(1)$ Toffoli gates.

Let $S=B+cA$, then after the first step the accumulator contains
\[
    \widetilde B=S\bmod 2^n,
\]
and the overflow bit
\[
    h=\lfloor S/2^n\rfloor
\]
has been XORed into the flag.  Since \(A,B_{\rm old}<p<2^n\), we have
\(S\le 2p-2<2^{n+1}\), so \(h\in\{0,1\}\).  The true reduction condition is
\(S\ge p\).  If \(h=0\), then \(\widetilde B=S\), so this condition is exactly
\([\widetilde B\ge p]\).  If \(h=1\), then reduction is certainly needed, while
\(\widetilde B=S-2^n<p\), so \([\widetilde B\ge p]=0\).  Hence the reduction
bit is
\[
    r=h\oplus[\widetilde B\ge p].
\]
We therefore update the flag by XORing in the fixed-modulus comparison
\([\widetilde B\ge p]\), which is evaluated by
the classical-quantum comparator from \cite{gidney2025classical}. Here we use the dirty-assisted adder which costs \(3n+O(1)\) Toffoli gates by borrowing the non-target
register \(A\) as dirty workspace. Controlled on this flag, we subtract the fixed modulus
\(p\) from \(B\), again using the same Gidney-style controlled
classical--quantum subtractor, with cost \(3n+O(1)\).
Finally, after the correction, the flag can be uncomputed by one Toffoli gate applying
\[
    f\leftarrow f\oplus c\,[B_{\rm final}<A].
\]
$[B_{\rm final}<A]$ can be  evaluated using the high-bit-only
Cuccaro quantum comparator, derived from the Cuccaro ripple-carry adder~\cite{cuccaro2004new},  costing \(2n+O(1)\) Toffoli gates. All clean and borrowed dirty workspace, including the flag \(f\), is restored
at the end of the block.
  
Therefore, one controlled modular addition costs
\[
\begin{aligned}
    T_{\mathrm{add}}(n)=(3n+O(1))+(3n+O(1))+(3n+O(1))+(2n+O(1))=11n+O(1)
\end{aligned}
\]
Toffoli gates.

\item \textbf{Modular doublings.} For modular doubling, we follow the implementation from \cite{roetteler2017quantum}. The binary doubling itself is implemented by Clifford
shift/CNOT operations. Since \(p\) is odd, the reduction flag can be cleared
from the least significant output bit.  The only parts that require Toffoli gates are
one fixed comparison with \(p\) and one conditional fixed subtraction, which can be implemented by the same classical-quantum adder from \cite{gidney2025classical}.

Therefore, one modular doubling costs
\[
    T_{\mathrm{dbl}}(n)
    =(3n+O(1))+(3n+O(1))
    =6n+O(1)
\]
Toffoli gates.
\end{itemize}

Consequently, the total number of Toffoli gates required by our modular multiplication circuit grows asymptotically as
\[
\begin{aligned}
    T(n)
    &= n\,T_{\mathrm{add}}(n)
       +(n-1)\,T_{\mathrm{dbl}}(n)
    = 17n^2+O(n).
\end{aligned}
\]
The modular squaring circuit is same to the modular multiplication, but with the addend equal to the input register itself. 
Therefore, it also costs at most $17n^2+O(n)$ Toffoli gates.

\end{document}